%% file: paper.tex
\documentclass[acmsmall,screen]{acmart}

\setcopyright{rightsretained}
\acmPrice{}
\acmDOI{10.1145/3571198}
\acmYear{2023}
\copyrightyear{2023}
\acmSubmissionID{popl23main-p48-p}
\acmJournal{PACMPL}
\acmVolume{7}
\acmNumber{POPL}
\acmArticle{5}
\acmMonth{1}
\received{2022-07-07}
\received[accepted]{2022-11-07}

\input{preamble}

\begin{document}

\title{ADEV: Sound Automatic Differentiation of Expected Values of Probabilistic Programs}

\author{Alexander K. Lew}\authornote{Equal contribution}
\email{alexlew@mit.edu} 
\affiliation{
  \institution{MIT}
  \city{Cambridge}
  \state{MA}
  \country{USA}                
}

\author{Mathieu Huot}\authornotemark[1]
\email{mathieu.huot@cs.ox.ac.uk}  
\affiliation{
  \institution{Oxford University} 
  \city{Oxford}
  \country{UK}          
}
\author{Sam Staton}
\email{sam.staton@cs.ox.ac.uk}  
\affiliation{
  \institution{Oxford University} 
  \city{Oxford}
  \country{UK}          
}
\author{Vikash K. Mansinghka}
\email{vkm@mit.edu}  
\affiliation{
  \institution{MIT} 
  \city{Cambridge}
  \state{MA}
  \country{USA}          
}
\input{00abstract}



\maketitle

\input{01intro}

\input{02background}
\input{03combinator}


\input{04discrete}
\input{05continuous}
\input{06general}
\input{09recap}
\input{07relatedwork}
\input{08discussion}
\input{acks}


\pagebreak
\bibliography{refs}

\pagebreak
\input{appendix/appendix}

\end{document}

%% file: preamble.tex
\usepackage{booktabs}   
\usepackage{subcaption} 
\usepackage{natbib}
 \usepackage{tcolorbox}
\usepackage{xcolor}
\usepackage{color, colortbl}
\usepackage{algorithm2e}
\usepackage{float}
\usepackage{enumitem}
\usepackage{tikz}
\usetikzlibrary{shapes.geometric, arrows}

\usepackage[draft]{minted}

\usepackage[bb=boondox]{mathalpha}
\DeclareSymbolFont{rsfso}{U}{rsfso}{m}{n}
\DeclareSymbolFontAlphabet{\mathscr}{rsfso}

\newcommand{\revision}[1]{#1} 

\newtheorem{theorem}{Theorem}[section]
\newtheorem{corollary}[theorem]{Corollary}
\newtheorem{lemma}[theorem]{Lemma}

\newtheorem{definition}{Definition}[section]

\newtheorem{notation}{Notation}[section]

\newtheorem{remark}{Remark}[section]

\newcommand{\KK}{\mathbb{K}}
\newcommand{\RR}{\mathbb{R}}
\newcommand{\NN}{\mathbb{N}}
\newcommand{\II}{\mathbb{I}}
\newcommand{\BB}{\mathbb{B}}
\newcommand{\grad}{\nabla}

\newcommand{\Set}[0]{\mathbf{Set}}
\newcommand{\SubSet}[0]{\mathbf{SubSet}}
\newcommand{\QBS}[0]{\mathbf{QBS}}

\newcommand{\plots}[0]{\mathcal{P}}

\newcommand{\ter}[0]{t}

\newcommand{\sem}[1]{\llbracket #1\rrbracket}
\newcommand{\type}[0]{\tau}

\newcommand{\var}[0]{x}

\newcommand{\expect}[0]{\mathbb{E}}

\newcommand{\op}[0]{op}

\newcommand{\pop}[1]{\ensuremath{\partial_i}\op}

\newcommand{\ad}[1]{\ensuremath{\mathcal{D}\{#1\}}}

\newcommand{\loss}{\mathcal{L}}
\newcommand{\syntaxcolor}[0]{\color{blue!70!black}}
\newcommand{\constantcolor}[1]{\mathbf{\syntaxcolor{#1}}}
\newcommand{\colforsyntax}[1]{\mathbf{\syntaxcolor{#1}}}

\newcommand{\True}[0]{\colforsyntax{True}}
\newcommand{\False}[0]{\colforsyntax{False}}
\newcommand{\haskdo}[0]{\colforsyntax{do}}
\newcommand{\dhaskdo}[0]{\colforsyntax{do}_\mathcal{D}}

\newcommand{\llet}[0]{\colforsyntax{let}}
\newcommand{\iin}[0]{\colforsyntax{in}}
\newcommand{\iif}[0]{\colforsyntax{if}}
\newcommand{\then}[0]{\colforsyntax{then}}
\newcommand{\eelse}[0]{\colforsyntax{else}}
\newcommand{\iifthenelse}[3]{\iif~#1~\then~#2~\eelse~#3}
\newcommand{\return}[0]{\colforsyntax{return}}
\newcommand{\dreturn}[0]{\colforsyntax{return}_\mathcal{D}}
\newcommand{\vreturn}[0]{\return_\ver}
\newcommand{\vhaskdo}[0]{\haskdo_\ver}
\newcommand{\fst}[0]{\colforsyntax{fst}}
\newcommand{\snd}[0]{\colforsyntax{snd}}
\newcommand{\splits}[0]{\mathbf{split}}
\newcommand{\logs}{\texttt{log}}
\newcommand{\exps}{\texttt{exp}}
\newcommand{\sins}{\texttt{sin}}
\newcommand{\coss}{\texttt{cos}}
\newcommand{\pows}{\texttt{pow}}
\newcommand{\der}{\mathcal{D}}
\newcommand{\ver}{\mathcal{V}}

\newcommand{\sample}[0]{\constantcolor{uniform}}
\newcommand{\normalreparam}[0]{\constantcolor{normal_\texttt{REPARAM}}}
\newcommand{\normalreinforce}[0]{\constantcolor{normal_\texttt{REINFORCE}}}
\newcommand{\geometricreinforce}[0]{\constantcolor{geometric_\texttt{REINFORCE}}}
\newcommand{\flipreinforce}[0]{\constantcolor{flip_\texttt{REINFORCE}}}
\newcommand{\flipenum}[0]{\constantcolor{flip_\texttt{ENUM}}}

\newcommand{\baseline}[0]{\constantcolor{baseline}}

\newcommand{\addcost}[0]{\constantcolor{add\_cost}}
\newcommand{\poissonweak}[0]{\constantcolor{poisson_\texttt{WEAK\_DERIV}}}
\newcommand{\gammaimplicit}[0]{\constantcolor{gamma_\texttt{IMPLICIT}}}
\newcommand{\importance}[0]{\constantcolor{importance}}
\newcommand{\reinforce}[0]{\constantcolor{reinforce}}
\newcommand{\leaveoneout}[0]{\constantcolor{leave\_one\_out}}
\newcommand{\smc}[0]{\constantcolor{smc}}
\newcommand{\implicitdiff}[0]{\constantcolor{implicit\_differentiation}}
\newcommand{\reparamreject}[0]{\constantcolor{reparam\_rejection\_sampler}}

\newcommand{\forget}[1]{\lfloor #1 \rfloor}

\newcommand{\mbe}{\constantcolor{\mathbbb{E}}}

\newcommand{\eRR}{\widetilde{\RR}}
\newcommand{\deRR}{\eRR_\mathcal{D}}
\newcommand{\posreal}[0]{\RR_{> 0}}
\newcommand{\eplus}{+^{\eRR}}
\newcommand{\etimes}{\times^{\eRR}}
\newcommand{\eexp}{exp^{\eRR}}
\newcommand{\minibatch}{\constantcolor{minibatch}}
\newcommand{\exact}{\constantcolor{exact}}
\newcommand{\monadsymbol}{P}
\newcommand{\pmonad}{\monadsymbol\,}
\newcommand{\dpmonad}{\monadsymbol_\der\,}
\newcommand{\vpmonad}{\monadsymbol_\ver\,}
\newcommand{\rel}{\mathscr{R}}

\newcommand{\Gl}[0]{\mathbf{Gl}}

\makeatletter
\def\PYG@reset{\let\PYG@it=\relax \let\PYG@bf=\relax%
    \let\PYG@ul=\relax \let\PYG@tc=\relax%
    \let\PYG@bc=\relax \let\PYG@ff=\relax}
\def\PYG@tok#1{\csname PYG@tok@#1\endcsname}
\def\PYG@toks#1+{\ifx\relax#1\empty\else%
    \PYG@tok{#1}\expandafter\PYG@toks\fi}
\def\PYG@do#1{\PYG@bc{\PYG@tc{\PYG@ul{%
    \PYG@it{\PYG@bf{\PYG@ff{#1}}}}}}}
\def\PYG#1#2{\PYG@reset\PYG@toks#1+\relax+\PYG@do{#2}}

\@namedef{PYG@tok@w}{\def\PYG@tc##1{\textcolor[rgb]{0.73,0.73,0.73}{##1}}}
\@namedef{PYG@tok@c}{\let\PYG@it=\textit\def\PYG@tc##1{\textcolor[rgb]{0.24,0.48,0.48}{##1}}}
\@namedef{PYG@tok@cp}{\def\PYG@tc##1{\textcolor[rgb]{0.61,0.40,0.00}{##1}}}
\@namedef{PYG@tok@k}{\let\PYG@bf=\textbf\def\PYG@tc##1{\textcolor[rgb]{0.00,0.50,0.00}{##1}}}
\@namedef{PYG@tok@kp}{\def\PYG@tc##1{\textcolor[rgb]{0.00,0.50,0.00}{##1}}}
\@namedef{PYG@tok@kt}{\def\PYG@tc##1{\textcolor[rgb]{0.69,0.00,0.25}{##1}}}
\@namedef{PYG@tok@o}{\def\PYG@tc##1{\textcolor[rgb]{0.40,0.40,0.40}{##1}}}
\@namedef{PYG@tok@ow}{\let\PYG@bf=\textbf\def\PYG@tc##1{\textcolor[rgb]{0.67,0.13,1.00}{##1}}}
\@namedef{PYG@tok@nb}{\def\PYG@tc##1{\textcolor[rgb]{0.00,0.50,0.00}{##1}}}
\@namedef{PYG@tok@nf}{\def\PYG@tc##1{\textcolor[rgb]{0.00,0.00,1.00}{##1}}}
\@namedef{PYG@tok@nc}{\let\PYG@bf=\textbf\def\PYG@tc##1{\textcolor[rgb]{0.00,0.00,1.00}{##1}}}
\@namedef{PYG@tok@nn}{\let\PYG@bf=\textbf\def\PYG@tc##1{\textcolor[rgb]{0.00,0.00,1.00}{##1}}}
\@namedef{PYG@tok@ne}{\let\PYG@bf=\textbf\def\PYG@tc##1{\textcolor[rgb]{0.80,0.25,0.22}{##1}}}
\@namedef{PYG@tok@nv}{\def\PYG@tc##1{\textcolor[rgb]{0.10,0.09,0.49}{##1}}}
\@namedef{PYG@tok@no}{\def\PYG@tc##1{\textcolor[rgb]{0.53,0.00,0.00}{##1}}}
\@namedef{PYG@tok@nl}{\def\PYG@tc##1{\textcolor[rgb]{0.46,0.46,0.00}{##1}}}
\@namedef{PYG@tok@ni}{\let\PYG@bf=\textbf\def\PYG@tc##1{\textcolor[rgb]{0.44,0.44,0.44}{##1}}}
\@namedef{PYG@tok@na}{\def\PYG@tc##1{\textcolor[rgb]{0.41,0.47,0.13}{##1}}}
\@namedef{PYG@tok@nt}{\let\PYG@bf=\textbf\def\PYG@tc##1{\textcolor[rgb]{0.00,0.50,0.00}{##1}}}
\@namedef{PYG@tok@nd}{\def\PYG@tc##1{\textcolor[rgb]{0.67,0.13,1.00}{##1}}}
\@namedef{PYG@tok@s}{\def\PYG@tc##1{\textcolor[rgb]{0.73,0.13,0.13}{##1}}}
\@namedef{PYG@tok@sd}{\let\PYG@it=\textit\def\PYG@tc##1{\textcolor[rgb]{0.73,0.13,0.13}{##1}}}
\@namedef{PYG@tok@si}{\let\PYG@bf=\textbf\def\PYG@tc##1{\textcolor[rgb]{0.64,0.35,0.47}{##1}}}
\@namedef{PYG@tok@se}{\let\PYG@bf=\textbf\def\PYG@tc##1{\textcolor[rgb]{0.67,0.36,0.12}{##1}}}
\@namedef{PYG@tok@sr}{\def\PYG@tc##1{\textcolor[rgb]{0.64,0.35,0.47}{##1}}}
\@namedef{PYG@tok@ss}{\def\PYG@tc##1{\textcolor[rgb]{0.10,0.09,0.49}{##1}}}
\@namedef{PYG@tok@sx}{\def\PYG@tc##1{\textcolor[rgb]{0.00,0.50,0.00}{##1}}}
\@namedef{PYG@tok@m}{\def\PYG@tc##1{\textcolor[rgb]{0.40,0.40,0.40}{##1}}}
\@namedef{PYG@tok@gh}{\let\PYG@bf=\textbf\def\PYG@tc##1{\textcolor[rgb]{0.00,0.00,0.50}{##1}}}
\@namedef{PYG@tok@gu}{\let\PYG@bf=\textbf\def\PYG@tc##1{\textcolor[rgb]{0.50,0.00,0.50}{##1}}}
\@namedef{PYG@tok@gd}{\def\PYG@tc##1{\textcolor[rgb]{0.63,0.00,0.00}{##1}}}
\@namedef{PYG@tok@gi}{\def\PYG@tc##1{\textcolor[rgb]{0.00,0.52,0.00}{##1}}}
\@namedef{PYG@tok@gr}{\def\PYG@tc##1{\textcolor[rgb]{0.89,0.00,0.00}{##1}}}
\@namedef{PYG@tok@ge}{\let\PYG@it=\textit}
\@namedef{PYG@tok@gs}{\let\PYG@bf=\textbf}
\@namedef{PYG@tok@gp}{\let\PYG@bf=\textbf\def\PYG@tc##1{\textcolor[rgb]{0.00,0.00,0.50}{##1}}}
\@namedef{PYG@tok@go}{\def\PYG@tc##1{\textcolor[rgb]{0.44,0.44,0.44}{##1}}}
\@namedef{PYG@tok@gt}{\def\PYG@tc##1{\textcolor[rgb]{0.00,0.27,0.87}{##1}}}
\@namedef{PYG@tok@err}{\def\PYG@bc##1{{\setlength{\fboxsep}{\string -\fboxrule}\fcolorbox[rgb]{1.00,0.00,0.00}{1,1,1}{\strut ##1}}}}
\@namedef{PYG@tok@kc}{\let\PYG@bf=\textbf\def\PYG@tc##1{\textcolor[rgb]{0.00,0.50,0.00}{##1}}}
\@namedef{PYG@tok@kd}{\let\PYG@bf=\textbf\def\PYG@tc##1{\textcolor[rgb]{0.00,0.50,0.00}{##1}}}
\@namedef{PYG@tok@kn}{\let\PYG@bf=\textbf\def\PYG@tc##1{\textcolor[rgb]{0.00,0.50,0.00}{##1}}}
\@namedef{PYG@tok@kr}{\let\PYG@bf=\textbf\def\PYG@tc##1{\textcolor[rgb]{0.00,0.50,0.00}{##1}}}
\@namedef{PYG@tok@bp}{\def\PYG@tc##1{\textcolor[rgb]{0.00,0.50,0.00}{##1}}}
\@namedef{PYG@tok@fm}{\def\PYG@tc##1{\textcolor[rgb]{0.00,0.00,1.00}{##1}}}
\@namedef{PYG@tok@vc}{\def\PYG@tc##1{\textcolor[rgb]{0.10,0.09,0.49}{##1}}}
\@namedef{PYG@tok@vg}{\def\PYG@tc##1{\textcolor[rgb]{0.10,0.09,0.49}{##1}}}
\@namedef{PYG@tok@vi}{\def\PYG@tc##1{\textcolor[rgb]{0.10,0.09,0.49}{##1}}}
\@namedef{PYG@tok@vm}{\def\PYG@tc##1{\textcolor[rgb]{0.10,0.09,0.49}{##1}}}
\@namedef{PYG@tok@sa}{\def\PYG@tc##1{\textcolor[rgb]{0.73,0.13,0.13}{##1}}}
\@namedef{PYG@tok@sb}{\def\PYG@tc##1{\textcolor[rgb]{0.73,0.13,0.13}{##1}}}
\@namedef{PYG@tok@sc}{\def\PYG@tc##1{\textcolor[rgb]{0.73,0.13,0.13}{##1}}}
\@namedef{PYG@tok@dl}{\def\PYG@tc##1{\textcolor[rgb]{0.73,0.13,0.13}{##1}}}
\@namedef{PYG@tok@s2}{\def\PYG@tc##1{\textcolor[rgb]{0.73,0.13,0.13}{##1}}}
\@namedef{PYG@tok@sh}{\def\PYG@tc##1{\textcolor[rgb]{0.73,0.13,0.13}{##1}}}
\@namedef{PYG@tok@s1}{\def\PYG@tc##1{\textcolor[rgb]{0.73,0.13,0.13}{##1}}}
\@namedef{PYG@tok@mb}{\def\PYG@tc##1{\textcolor[rgb]{0.40,0.40,0.40}{##1}}}
\@namedef{PYG@tok@mf}{\def\PYG@tc##1{\textcolor[rgb]{0.40,0.40,0.40}{##1}}}
\@namedef{PYG@tok@mh}{\def\PYG@tc##1{\textcolor[rgb]{0.40,0.40,0.40}{##1}}}
\@namedef{PYG@tok@mi}{\def\PYG@tc##1{\textcolor[rgb]{0.40,0.40,0.40}{##1}}}
\@namedef{PYG@tok@il}{\def\PYG@tc##1{\textcolor[rgb]{0.40,0.40,0.40}{##1}}}
\@namedef{PYG@tok@mo}{\def\PYG@tc##1{\textcolor[rgb]{0.40,0.40,0.40}{##1}}}
\@namedef{PYG@tok@ch}{\let\PYG@it=\textit\def\PYG@tc##1{\textcolor[rgb]{0.24,0.48,0.48}{##1}}}
\@namedef{PYG@tok@cm}{\let\PYG@it=\textit\def\PYG@tc##1{\textcolor[rgb]{0.24,0.48,0.48}{##1}}}
\@namedef{PYG@tok@cpf}{\let\PYG@it=\textit\def\PYG@tc##1{\textcolor[rgb]{0.24,0.48,0.48}{##1}}}
\@namedef{PYG@tok@c1}{\let\PYG@it=\textit\def\PYG@tc##1{\textcolor[rgb]{0.24,0.48,0.48}{##1}}}
\@namedef{PYG@tok@cs}{\let\PYG@it=\textit\def\PYG@tc##1{\textcolor[rgb]{0.24,0.48,0.48}{##1}}}

\def\PYGZbs{\char`\\}
\def\PYGZus{\char`\_}
\def\PYGZob{\char`\{}
\def\PYGZcb{\char`\}}
\def\PYGZca{\char`\^}

\def\PYGZlt{\char`\<}
\def\PYGZgt{\char`\>}
\def\PYGZsh{\char`\#}

\def\PYGZdl{\char`\$}
\def\PYGZhy{\char`\-}
\def\PYGZsq{\char`\'}

\makeatother

\makeatletter
\def\PYGdefault@reset{\let\PYGdefault@it=\relax \let\PYGdefault@bf=\relax%
    \let\PYGdefault@ul=\relax \let\PYGdefault@tc=\relax%
    \let\PYGdefault@bc=\relax \let\PYGdefault@ff=\relax}
\def\PYGdefault@tok#1{\csname PYGdefault@tok@#1\endcsname}
\def\PYGdefault@toks#1+{\ifx\relax#1\empty\else%
    \PYGdefault@tok{#1}\expandafter\PYGdefault@toks\fi}
\def\PYGdefault@do#1{\PYGdefault@bc{\PYGdefault@tc{\PYGdefault@ul{%
    \PYGdefault@it{\PYGdefault@bf{\PYGdefault@ff{#1}}}}}}}
\def\PYGdefault#1#2{\PYGdefault@reset\PYGdefault@toks#1+\relax+\PYGdefault@do{#2}}

\@namedef{PYGdefault@tok@w}{\def\PYGdefault@tc##1{\textcolor[rgb]{0.73,0.73,0.73}{##1}}}
\@namedef{PYGdefault@tok@c}{\let\PYGdefault@it=\textit\def\PYGdefault@tc##1{\textcolor[rgb]{0.24,0.48,0.48}{##1}}}
\@namedef{PYGdefault@tok@cp}{\def\PYGdefault@tc##1{\textcolor[rgb]{0.61,0.40,0.00}{##1}}}
\@namedef{PYGdefault@tok@k}{\let\PYGdefault@bf=\textbf\def\PYGdefault@tc##1{\textcolor[rgb]{0.00,0.50,0.00}{##1}}}
\@namedef{PYGdefault@tok@kp}{\def\PYGdefault@tc##1{\textcolor[rgb]{0.00,0.50,0.00}{##1}}}
\@namedef{PYGdefault@tok@kt}{\def\PYGdefault@tc##1{\textcolor[rgb]{0.69,0.00,0.25}{##1}}}
\@namedef{PYGdefault@tok@o}{\def\PYGdefault@tc##1{\textcolor[rgb]{0.40,0.40,0.40}{##1}}}
\@namedef{PYGdefault@tok@ow}{\let\PYGdefault@bf=\textbf\def\PYGdefault@tc##1{\textcolor[rgb]{0.67,0.13,1.00}{##1}}}
\@namedef{PYGdefault@tok@nb}{\def\PYGdefault@tc##1{\textcolor[rgb]{0.00,0.50,0.00}{##1}}}
\@namedef{PYGdefault@tok@nf}{\def\PYGdefault@tc##1{\textcolor[rgb]{0.00,0.00,1.00}{##1}}}
\@namedef{PYGdefault@tok@nc}{\let\PYGdefault@bf=\textbf\def\PYGdefault@tc##1{\textcolor[rgb]{0.00,0.00,1.00}{##1}}}
\@namedef{PYGdefault@tok@nn}{\let\PYGdefault@bf=\textbf\def\PYGdefault@tc##1{\textcolor[rgb]{0.00,0.00,1.00}{##1}}}
\@namedef{PYGdefault@tok@ne}{\let\PYGdefault@bf=\textbf\def\PYGdefault@tc##1{\textcolor[rgb]{0.80,0.25,0.22}{##1}}}
\@namedef{PYGdefault@tok@nv}{\def\PYGdefault@tc##1{\textcolor[rgb]{0.10,0.09,0.49}{##1}}}
\@namedef{PYGdefault@tok@no}{\def\PYGdefault@tc##1{\textcolor[rgb]{0.53,0.00,0.00}{##1}}}
\@namedef{PYGdefault@tok@nl}{\def\PYGdefault@tc##1{\textcolor[rgb]{0.46,0.46,0.00}{##1}}}
\@namedef{PYGdefault@tok@ni}{\let\PYGdefault@bf=\textbf\def\PYGdefault@tc##1{\textcolor[rgb]{0.44,0.44,0.44}{##1}}}
\@namedef{PYGdefault@tok@na}{\def\PYGdefault@tc##1{\textcolor[rgb]{0.41,0.47,0.13}{##1}}}
\@namedef{PYGdefault@tok@nt}{\let\PYGdefault@bf=\textbf\def\PYGdefault@tc##1{\textcolor[rgb]{0.00,0.50,0.00}{##1}}}
\@namedef{PYGdefault@tok@nd}{\def\PYGdefault@tc##1{\textcolor[rgb]{0.67,0.13,1.00}{##1}}}
\@namedef{PYGdefault@tok@s}{\def\PYGdefault@tc##1{\textcolor[rgb]{0.73,0.13,0.13}{##1}}}
\@namedef{PYGdefault@tok@sd}{\let\PYGdefault@it=\textit\def\PYGdefault@tc##1{\textcolor[rgb]{0.73,0.13,0.13}{##1}}}
\@namedef{PYGdefault@tok@si}{\let\PYGdefault@bf=\textbf\def\PYGdefault@tc##1{\textcolor[rgb]{0.64,0.35,0.47}{##1}}}
\@namedef{PYGdefault@tok@se}{\let\PYGdefault@bf=\textbf\def\PYGdefault@tc##1{\textcolor[rgb]{0.67,0.36,0.12}{##1}}}
\@namedef{PYGdefault@tok@sr}{\def\PYGdefault@tc##1{\textcolor[rgb]{0.64,0.35,0.47}{##1}}}
\@namedef{PYGdefault@tok@ss}{\def\PYGdefault@tc##1{\textcolor[rgb]{0.10,0.09,0.49}{##1}}}
\@namedef{PYGdefault@tok@sx}{\def\PYGdefault@tc##1{\textcolor[rgb]{0.00,0.50,0.00}{##1}}}
\@namedef{PYGdefault@tok@m}{\def\PYGdefault@tc##1{\textcolor[rgb]{0.40,0.40,0.40}{##1}}}
\@namedef{PYGdefault@tok@gh}{\let\PYGdefault@bf=\textbf\def\PYGdefault@tc##1{\textcolor[rgb]{0.00,0.00,0.50}{##1}}}
\@namedef{PYGdefault@tok@gu}{\let\PYGdefault@bf=\textbf\def\PYGdefault@tc##1{\textcolor[rgb]{0.50,0.00,0.50}{##1}}}
\@namedef{PYGdefault@tok@gd}{\def\PYGdefault@tc##1{\textcolor[rgb]{0.63,0.00,0.00}{##1}}}
\@namedef{PYGdefault@tok@gi}{\def\PYGdefault@tc##1{\textcolor[rgb]{0.00,0.52,0.00}{##1}}}
\@namedef{PYGdefault@tok@gr}{\def\PYGdefault@tc##1{\textcolor[rgb]{0.89,0.00,0.00}{##1}}}
\@namedef{PYGdefault@tok@ge}{\let\PYGdefault@it=\textit}
\@namedef{PYGdefault@tok@gs}{\let\PYGdefault@bf=\textbf}
\@namedef{PYGdefault@tok@gp}{\let\PYGdefault@bf=\textbf\def\PYGdefault@tc##1{\textcolor[rgb]{0.00,0.00,0.50}{##1}}}
\@namedef{PYGdefault@tok@go}{\def\PYGdefault@tc##1{\textcolor[rgb]{0.44,0.44,0.44}{##1}}}
\@namedef{PYGdefault@tok@gt}{\def\PYGdefault@tc##1{\textcolor[rgb]{0.00,0.27,0.87}{##1}}}
\@namedef{PYGdefault@tok@err}{\def\PYGdefault@bc##1{{\setlength{\fboxsep}{\string -\fboxrule}\fcolorbox[rgb]{1.00,0.00,0.00}{1,1,1}{\strut ##1}}}}
\@namedef{PYGdefault@tok@kc}{\let\PYGdefault@bf=\textbf\def\PYGdefault@tc##1{\textcolor[rgb]{0.00,0.50,0.00}{##1}}}
\@namedef{PYGdefault@tok@kd}{\let\PYGdefault@bf=\textbf\def\PYGdefault@tc##1{\textcolor[rgb]{0.00,0.50,0.00}{##1}}}
\@namedef{PYGdefault@tok@kn}{\let\PYGdefault@bf=\textbf\def\PYGdefault@tc##1{\textcolor[rgb]{0.00,0.50,0.00}{##1}}}
\@namedef{PYGdefault@tok@kr}{\let\PYGdefault@bf=\textbf\def\PYGdefault@tc##1{\textcolor[rgb]{0.00,0.50,0.00}{##1}}}
\@namedef{PYGdefault@tok@bp}{\def\PYGdefault@tc##1{\textcolor[rgb]{0.00,0.50,0.00}{##1}}}
\@namedef{PYGdefault@tok@fm}{\def\PYGdefault@tc##1{\textcolor[rgb]{0.00,0.00,1.00}{##1}}}
\@namedef{PYGdefault@tok@vc}{\def\PYGdefault@tc##1{\textcolor[rgb]{0.10,0.09,0.49}{##1}}}
\@namedef{PYGdefault@tok@vg}{\def\PYGdefault@tc##1{\textcolor[rgb]{0.10,0.09,0.49}{##1}}}
\@namedef{PYGdefault@tok@vi}{\def\PYGdefault@tc##1{\textcolor[rgb]{0.10,0.09,0.49}{##1}}}
\@namedef{PYGdefault@tok@vm}{\def\PYGdefault@tc##1{\textcolor[rgb]{0.10,0.09,0.49}{##1}}}
\@namedef{PYGdefault@tok@sa}{\def\PYGdefault@tc##1{\textcolor[rgb]{0.73,0.13,0.13}{##1}}}
\@namedef{PYGdefault@tok@sb}{\def\PYGdefault@tc##1{\textcolor[rgb]{0.73,0.13,0.13}{##1}}}
\@namedef{PYGdefault@tok@sc}{\def\PYGdefault@tc##1{\textcolor[rgb]{0.73,0.13,0.13}{##1}}}
\@namedef{PYGdefault@tok@dl}{\def\PYGdefault@tc##1{\textcolor[rgb]{0.73,0.13,0.13}{##1}}}
\@namedef{PYGdefault@tok@s2}{\def\PYGdefault@tc##1{\textcolor[rgb]{0.73,0.13,0.13}{##1}}}
\@namedef{PYGdefault@tok@sh}{\def\PYGdefault@tc##1{\textcolor[rgb]{0.73,0.13,0.13}{##1}}}
\@namedef{PYGdefault@tok@s1}{\def\PYGdefault@tc##1{\textcolor[rgb]{0.73,0.13,0.13}{##1}}}
\@namedef{PYGdefault@tok@mb}{\def\PYGdefault@tc##1{\textcolor[rgb]{0.40,0.40,0.40}{##1}}}
\@namedef{PYGdefault@tok@mf}{\def\PYGdefault@tc##1{\textcolor[rgb]{0.40,0.40,0.40}{##1}}}
\@namedef{PYGdefault@tok@mh}{\def\PYGdefault@tc##1{\textcolor[rgb]{0.40,0.40,0.40}{##1}}}
\@namedef{PYGdefault@tok@mi}{\def\PYGdefault@tc##1{\textcolor[rgb]{0.40,0.40,0.40}{##1}}}
\@namedef{PYGdefault@tok@il}{\def\PYGdefault@tc##1{\textcolor[rgb]{0.40,0.40,0.40}{##1}}}
\@namedef{PYGdefault@tok@mo}{\def\PYGdefault@tc##1{\textcolor[rgb]{0.40,0.40,0.40}{##1}}}
\@namedef{PYGdefault@tok@ch}{\let\PYGdefault@it=\textit\def\PYGdefault@tc##1{\textcolor[rgb]{0.24,0.48,0.48}{##1}}}
\@namedef{PYGdefault@tok@cm}{\let\PYGdefault@it=\textit\def\PYGdefault@tc##1{\textcolor[rgb]{0.24,0.48,0.48}{##1}}}
\@namedef{PYGdefault@tok@cpf}{\let\PYGdefault@it=\textit\def\PYGdefault@tc##1{\textcolor[rgb]{0.24,0.48,0.48}{##1}}}
\@namedef{PYGdefault@tok@c1}{\let\PYGdefault@it=\textit\def\PYGdefault@tc##1{\textcolor[rgb]{0.24,0.48,0.48}{##1}}}
\@namedef{PYGdefault@tok@cs}{\let\PYGdefault@it=\textit\def\PYGdefault@tc##1{\textcolor[rgb]{0.24,0.48,0.48}{##1}}}

\makeatother

%% file: 00abstract.tex
\begin{abstract}

Optimizing the expected values of probabilistic processes is a central problem in computer science and its applications, arising in fields ranging from artificial intelligence to operations research to statistical computing. Unfortunately, automatic differentiation techniques developed for deterministic programs do not in general compute the correct gradients needed for widely used solutions based on gradient-based optimization.

In this paper, we present ADEV, an extension to forward-mode AD that correctly differentiates the expectations of probabilistic processes  represented as programs that make random choices. Our algorithm is a source-to-source program transformation on an expressive, higher-order language for probabilistic computation, with both discrete and continuous probability distributions. The result of our transformation is a new probabilistic program, whose expected return value is the derivative of the original program's expectation. This output program can be run to generate unbiased Monte Carlo estimates of the desired gradient, which can then be used within the inner loop of stochastic gradient descent. We prove ADEV correct using logical relations over the denotations of the source and target probabilistic programs. Because it modularly extends forward-mode AD, our algorithm lends itself to a concise implementation strategy, which we exploit to develop a prototype in just a few dozen lines of Haskell (\href{https://github.com/probcomp/adev}{https://github.com/probcomp/adev}).


\end{abstract}

%% file: 01intro.tex
\section{Introduction}

Specifying and solving optimization problems has never been easier, thanks in large part to the maturation of programming languages and libraries that support \textit{automatic differentiation} (AD). With AD, users can specify objective functions as programs, then automate the construction of programs for computing their derivatives. These derivatives can be fed into optimization algorithms, such as gradient descent or ADAM, to find local minima or maxima of the original objective function. 

\noindent Unfortunately, there is an important class of functions that today's AD systems {\it cannot} differentiate correctly: those defined as \textit{expected values} of probabilistic processes. Consider, for example, the reinforcement learning problem of optimizing the parameters of a robot's algorithm, based on simulations of its behavior in random environments. The practitioner hopes maximize the \textit{expected} (i.e., average) reward across all possible runs of the simulator. But obtaining gradients of this objective is not straight-forward; naively applying AD to the stochastic reward simulator will in general give incorrect results. Instead, practitioners often resort to hand-derived gradient estimators that they must manually prove correct. And this dilemma is hardly unique to robotics: the optimization of expected values is a ubiquitous problem, as the diverse examples in
Table~\ref{tab:example-applications} 
attest.

\input{figures/01figures/intro_commutative_diagram}

\input{figures/01figures/exampleAD}

In this paper, we present ADEV, a new AD algorithm that {correctly} computes derivatives of the expected values of probabilistic programs. Our general approach is sketched in Figure~\ref{fig:ad-diagram}:
\begin{itemize}[leftmargin=*]
    \item The user provides a program $t$ encoding a probabilistic process dependent on a parameter $\theta$.
    
    \item The user's goal is to find $\theta^* = \text{argmin}_\theta\, \loss(\theta)$, where the \textit{loss function} $\loss$ maps a parameter value $\theta$ to the \textit{expected return value} of $t$, run on input $\theta$.
    
    \item Applying ADEV to $t$ yields a new probabilistic program $s$. Our algorithm is \textit{correct} in that the expected return value of $s$ at input $\theta$ is exactly the derivative $\loss'(\theta)$ of the loss.
    
    \item Even if $\loss'(\theta)$ cannot be evaluated exactly, users can {\it run} the probabilistic program $s$ to simulate provably unbiased estimates of $\loss'(\theta)$, which can be used for stochastic optimization.
\end{itemize}

Figure~\ref{fig:flip-example} illustrates our method on a toy example
. The loss function $\mathcal{L}$ is defined as the expectation of a program that flips a biased coin, with probability-of-heads $\theta$. Depending on the outcome, we receive either 0 loss (the `heads' case), or a \textit{negative} loss of $-\frac{\theta}{2}$ (indicating a positive reward). The problem is to find the $\theta$ that minimizes expected loss. Intuitively, the optimal strategy must trade off the \textit{benefits} of increasing $\theta$ (higher payoff in the `tails' case) with its \textit{drawbacks} (lower probability of entering the `tails' case in the first place). The expected loss $\loss(\theta) = \frac{\theta^2-\theta}{2}$ is minimized at $\theta = 0.5$. 

Applying AD to only the deterministic parts of $f$ fails to account for the effect of increasing $\theta$ on the \textit{probability} of entering the high-reward branch. The resulting (incorrect) gradient is negative for all $\theta \in (0, 1)$; optimizing with it significantly overshoots the optimal value of $0.5$. By contrast, ADEV automatically introduces additional terms to account for the dependence of $b$ on $\theta$, leading to a gradient that can be soundly used to optimize the loss. 
\input{figures/01figures/table}

Our translation of $\loss$ into $\loss'$ may appear complex and non-local, but in fact, we arrived at our algorithm by modularly extending a standard `dual-number' forward-mode AD macro (e.g., as presented by~\citet{huot2020correctness}) to handle probabilistic types and terms. As in standard forward-mode AD, our translation is mostly structure-preserving, with almost all the action happening in the translation of primitives, like $\texttt{flip}$ in this example. (The term we display for $\loss'$ in Figure~\ref{fig:flip-example} has been further simplified for clarity, via monad laws and $\beta$-reductions; see Figure~\ref{fig:04reduction-example}.)
\\

\noindent\textbf{\textit{Contributions.}} We present ADEV, a new AD algorithm for correctly automating the derivatives of the expectations of expressive probabilistic programs. It has the following desirable properties:

\begin{enumerate}[leftmargin=*]
    \item {\bf Provably correct:} It comes with guarantees relating the output program's expectation to the input program's expectation's derivative (Theorem~\ref{thm:full-correctness}).
    \item {\bf Modular:} ADEV is a modular extension to traditional forward-mode AD, and can be modularly extended to support new gradient estimators and probabilistic primitives (Table~\ref{tab:extensions}).
    \item {\bf Compositional:} ADEV's translation is local, in that all the action happens in the translation of primitives (as in standard forward-mode AD).
    \item {\bf Flexible:} ADEV provides levers for navigating trade-offs between the variance and computational cost of the output  program, viewed as an unbiased gradient estimator.
    \item {\bf Easy to implement:} It is easy to modify existing forward-mode implementations to support ADEV — our Haskell prototype is just a few dozen lines (Appx.~\ref{appx:impl}, \href{http://github.com/probcomp/adev}{github.com/probcomp/adev}).
\end{enumerate}

\input{figures/01figures/extensions_table}
\\\vspace{-2mm}

\noindent\textbf{\textit{Key challenges.}} To develop our algorithm, we had to overcome four key technical challenges:
\begin{enumerate}[leftmargin=*]
    \item \textbf{Challenge: Reasoning about correctness compositionally.} Our correctness criterion makes sense for the main program, but not necessarily for subterms, hindering compositional reasoning. 
    
    \textbf{Solution: Logical relations.} We adapt the \textit{logical relations} technique of~\citet{huot2020correctness} (Sec.~\ref{sec:background}) to define extended correctness criteria that apply to any type in our language.
    
    \item \textbf{Challenge: Compositional differentiation of probability kernels.} ML researchers often build gradient estimators for whole models~\citep{mohamed2020monte}. But to differentiate \textit{compositionally} we need a notion of `probability kernel derivative,' and rules for composition. 
    
    \textbf{Solution: Higher-order semantics of probabilistic programs and AD.} Recent PPL semantics view probability as a submonad of the continuation monad~\citep{vakar2019domain}. In this light, probabilistic primitives are really higher-order primitives, averaging a continuation's value over all possible sampled inputs. This gives a blueprint for a notion of derivative at probabilistic types, based on existing theory of higher-order AD~\citep{huot2020correctness} (Sec.~\ref{sec:discrete}). 
    
    \item \textbf{Challenge: Commuting limits.} Differentiating expectations requires swapping integrals and derivatives, which may not be sound. The dominated convergence theorem gives sufficient regularity conditions, but they are hard to formulate or enforce compositionally.
    
    \textbf{Solution: Lightweight static analysis to surface regularity conditions.} Our macro optionally outputs a \textit{verification condition} (presented to the user as syntax) making explicit every regularity assumption that ADEV makes while translating a program (Sec.~\ref{sec:continuous}). These regularity assumptions are often ignored (i.e., not even stated) in the ML literature on gradient estimation.
    
    \item \textbf{Challenge: Safely exposing non-differentiable primitives.} Probabilistic programs that use non-differentiable primitives, like $\leq$ or $ReLU$, may have differentiable expectations. But dominated convergence requires integrands to be continuously differentiable w.r.t. the parameter.
    
    \textbf{Solution: Static typing for fine-grained differentiability tracking.} To ensure we only swap integrals and derivatives when it is sound to do so, we use static typing to track the smoothness of deterministic subterms with respect to each of their free variables (Sec.~\ref{sec:general}).
\end{enumerate}

%% file: figures/01figures/intro_commutative_diagram.tex
\begin{figure}
    \centering
    \includegraphics[width=0.95\linewidth]{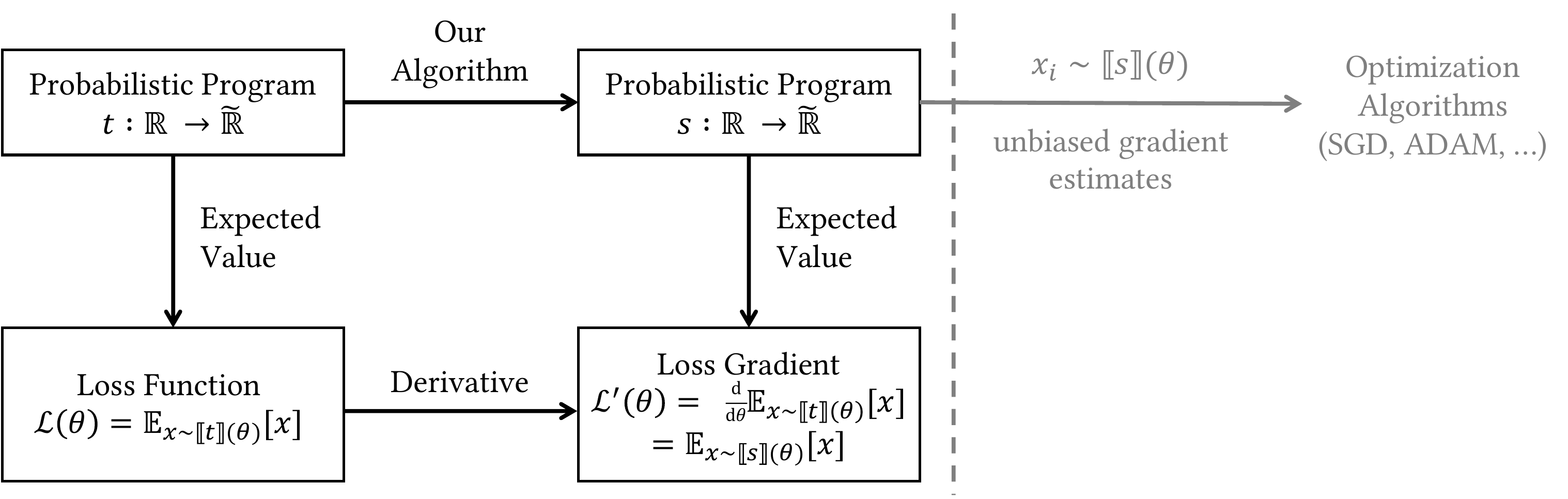}
        \vspace{-4mm}
    \caption{Our approach to differentiating loss functions defined as expected values. Our algorithm takes as input a probabilistic program $t$, which, given a parameter of type $\RR$ (or a subtype), outputs a value of type $\eRR$, which represents \textit{probabilistic estimators} of losses (Def.~\ref{def:eRRdef}). We translate $t$ to a new probabilistic program $s$, whose expected return value is the \textit{derivative} of $t$'s expected return value. Running $s$ yields provably unbiased estimates $x_i$ of the loss's derivative, which can be used to guide optimization.}\vspace{-2mm}
    \label{fig:ad-diagram}
\end{figure}

%% file: figures/01figures/exampleAD.tex
\begin{figure}
\centering
\begin{minipage}[t]{0.7\linewidth}
\footnotesize{
    \begin{tabular}{|l|l|l|}
         \multicolumn{1}{c}{\bf Input Loss as a} & 
         \multicolumn{1}{c}{\bf AD on deterministic} & 
         \multicolumn{1}{c}{\bf ADEV} \\
         \multicolumn{1}{c}{\bf Probabilistic Program} & 
         \multicolumn{1}{c}{\bf parts only (incorrect)} & 
         \multicolumn{1}{c}{\bf (correct derivative)} \\\hline
    \hspace{-1mm}\begin{tabular}{l}
    $\loss = \lambda \theta:\II.\,  \mbe (\haskdo~\{$\\
    $\quad b \gets \texttt{flip}\,\theta$\\
    $\quad \iif~b~\then$\\
    $\quad\quad \return~0$\\
    $\quad\eelse$\\
    $\quad\quad \return~-(\theta \div {2})$\\ 
    $\})$\\
     \\
     \\
    \end{tabular}
    & 
    \hspace{-1mm}\begin{tabular}{l}
    $\loss' = \lambda \theta:\II. \, \mbe(\haskdo~\{$\\
    $\quad b \gets \texttt{flip}\,\theta$\\
    $\quad \iif~b~\then$\\
    $\quad\quad \return~0$\\
    $\quad\eelse$\\
    $\quad\quad \return~-1 \div {2}$\\ 
    $\})$\\
    \\
    \\
    \end{tabular}
    & 
    \hspace{-1mm}\begin{tabular}{l}
    $\loss' = \lambda \theta:\II. \mbe(\haskdo~\{$\\
    $\quad b \gets \texttt{flip}\,\theta$\\
    $\quad \iif~b~\then$\\
    $\quad\quad \return~0$\\
    $\quad\eelse$\\
    $\quad\quad\llet~\delta p = 1 \div (\theta-1)$\\
    $\quad\quad\llet~\delta l = -1 \div 2$\\
    $\quad\quad\llet~l = -\theta \div 2$\\
    $\quad\quad \return~\delta l + l \times \delta p\})$
    \end{tabular}\\\hline
    \multicolumn{3}{c}{}   \vspace{-3mm}\\
    \multicolumn{1}{c}{$\loss(\theta) = \frac{\theta^2 - \theta}{2}$} &
    \multicolumn{1}{c}{$\loss_{\textit{naive}}'(\theta) = \frac{\theta-1}{2}$}
    &
    \multicolumn{1}{c}{$\loss'_{\textit{correct}}(\theta) = \theta - \frac{1}{2}$}
    \end{tabular}
      \vspace{-2mm}
    }
\end{minipage}%
\begin{minipage}{0.3\linewidth}
\includegraphics[width=0.9\linewidth]{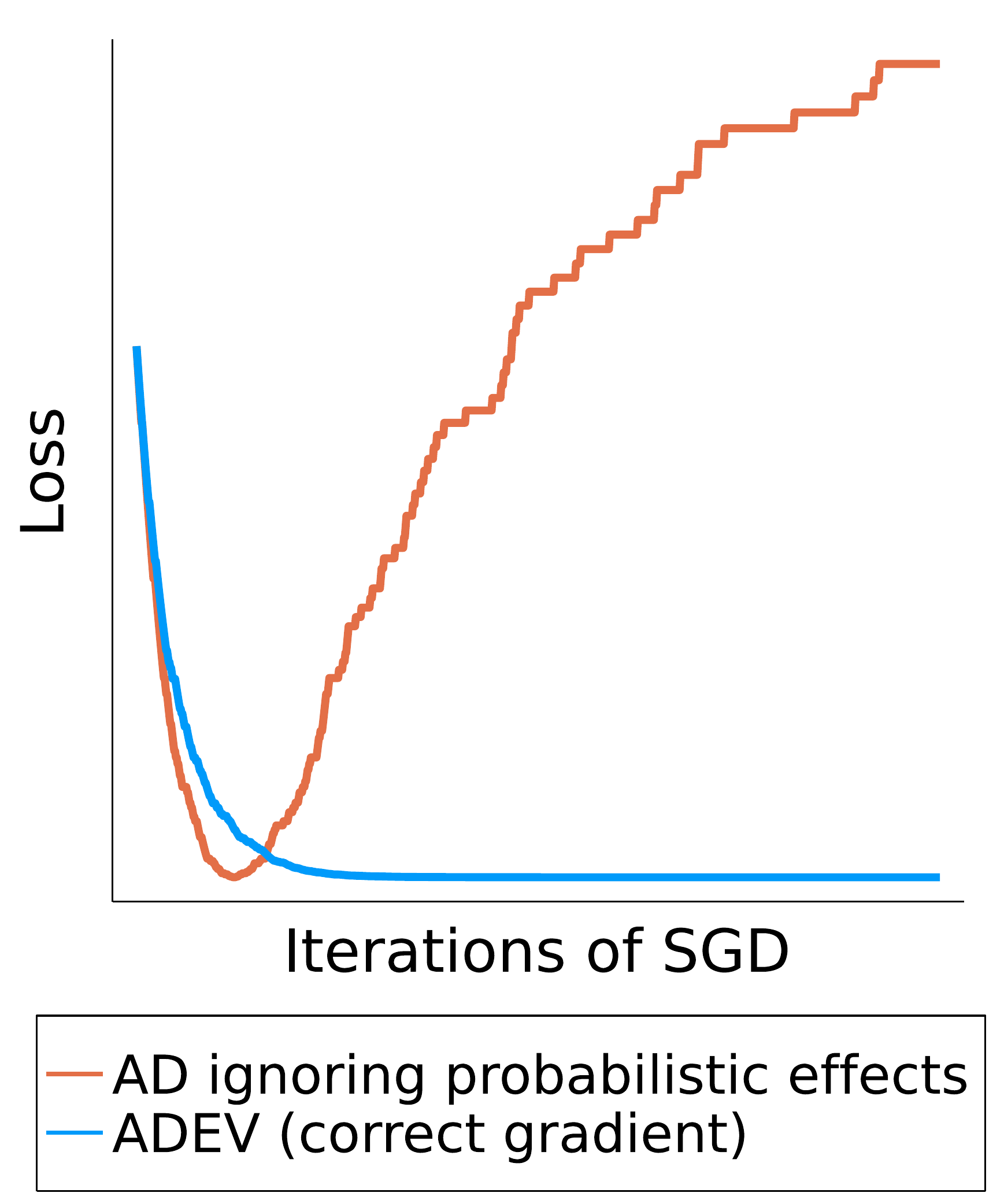}
\end{minipage}
    \vspace{-2mm}
\caption{If probabilistic constructs are ignored, AD may produce incorrect results. In this case, standard AD fails to account for $\theta$'s 
effect on the {\it probability} of entering each branch. ADEV, by contrast, correctly accounts for the probabilistic effects, generating similar code to what a practitioner might hand-derive. \textit{Right:} Correct gradients are often crucial for downstream applications, e.g. optimization via stochastic gradient descent.}\vspace{-6mm}
\label{fig:flip-example}
\end{figure} 

%% file: figures/01figures/table.tex
\begin{table} 
\footnotesize{
\caption{The need to differentiate expected values of probabilistic processes is ubiquitous in many fields, including machine learning, operations research, and finance~\citep{mohamed2020monte}.}\label{tab:example-applications}
\vspace{-3mm}
\begin{tabular}{p{0.2\linewidth}p{0.35\linewidth}p{0.15\linewidth}p{0.19\linewidth}}
\toprule
                                                    \textbf{Application}                                             & \textbf{Probabilistic Process}                                                                                                                                                        & \textbf{Expected Value}                             & \textbf{Use of Gradients}                                                                            \\\midrule
{Supervised learning}                                                                     & Evaluate loss on random minibatch                                                                                                                            & Loss on all data                      & Minimize total loss                                                                                        \\
{Reinforcement learning}                                                             & Measure reward in simulated environment                                                                                & Average reward                & Maximize reward                                                                                    \\
{Variational Bayes}                                                          & Sample variational family, estimate ELBO                                                                                                                               & ELBO objective                             & Minimize $KL(q || p)$                                                                                              \\
{Train on synthetic data}                                                            & Generate synthetic data and evaluate loss           & Expected loss under simulator & Minimize average loss                                                                                      \\
{Sensitivity analysis in computational finance} & Simulate future option prices, to assess investment risk                                                                                                            & Expected risk                            & Analyze risk assessment's sensitivity to pricing assumptions   \\
{Operations research} & Simulate efficiency of a customer queue                                                     & Average efficiency        & Maximize efficiency                                                                                        \\
{Bayesian optimization}                                                                   & Sample current belief distribution about a function's value at a candidate point, and evaluate whether the point would be a new `best parameter value' & Probability of improvement over current best parameter value                        &Choose next sample point to maximize probability of improvement\\\bottomrule
\end{tabular}
}
\vspace{-3mm}
\end{table}

%% file: figures/01figures/extensions_table.tex
\noindent
\begin{minipage}{0.2\textwidth}
\footnotesize{
\vspace{0mm}
\begin{center}
    \textit{\quad Recipe for New\\\quad ADEV Modules}
\end{center}
\vspace{3mm}
\begin{tabular}{p{\textwidth}}
\framebox{\parbox{\dimexpr\linewidth-2\fboxsep-2\fboxrule}{%
  Add new types, constructs, or primitives}}\\[5mm]
\framebox{\parbox{\dimexpr\linewidth-2\fboxsep-2\fboxrule}{%
Extend macro $\ad{\cdot}$ to new constructs}}\\[5mm]
\framebox{\parbox{\dimexpr\linewidth-2\fboxsep-2\fboxrule}{%
For new types $\tau$, define specification $\rel_\tau$}}\\[5mm]
\framebox{\parbox{\dimexpr\linewidth-2\fboxsep-2\fboxrule}{%
Prove new constructs preserve correctness}}
\end{tabular}
}
\end{minipage}\hfill
\begin{minipage}{0.8\textwidth}
\vspace{-4mm}
\begin{table}[H]
\footnotesize{
\caption{ADEV is implemented modularly and admits modular extensions.}\label{tab:extensions}
\vspace{-3mm}
\begin{tabular}{p{0.75\textwidth}p{0.1\textwidth}}
\toprule
                                                    \textbf{Modular language extension}                                             & \textbf{Reference}                                                                                                    \\\midrule
 
Real-valued probabilistic primitives + combinators & Sec.~\ref{sec:combinator}\\
 Discrete prob. prog. + enumeration + REINFORCE~\citep{ranganath2014black} & Sec.~\ref{sec:discrete}\\
 Continuous prob. prog. + REPARAM~\citep{kingma2013auto} & Sec.~\ref{sec:continuous}\\
 Discontinuous operations (e.g. $\leq$) & Sec.~\ref{sec:general}\\
 Control variates (baselines) for variance reduction~\citep{mnih2014neural} & Appx.~\ref{sub:baseline}\\
Variance reduction via dependency tracking~\citep{schulman2015gradient} & Appx.~\ref{sub:scg}\\
 Storchastic~\citep{krieken2021storchastic} multi-sample estimators & Appx.~\ref{sub:storchastic}\\
 Higher-order primitive for differentiable particle filter~\citep{scibior2021differentiable} & Appx.~\ref{sub:particle-filter}\\
 Implicit reparameterization gradients~\citep{figurnov2018implicit} & Appx.~\ref{sub:implicit-diff}\\
 Weak or measure-valued derivatives~\citep{heidergott2000measure} & Appx.~\ref{sub:weak-deriv}\\
 Reparameterized rejection gradients~\citep{naesseth2017reparameterization} & Appx.~\ref{sub:reject}\\\bottomrule
\end{tabular}
}
\end{table}
\end{minipage}

%% file: 02background.tex
\section{Background: Forward-Mode AD for Deterministic Programs}
\label{sec:background}

In this section, we review standard forward-mode AD, a well-established technique for automating the derivatives of \textit{deterministic} programs~\citep{rall1981automatic,director1969automated,griewank2008evaluating}.
Our presentation is based on~\citet{huot2020correctness}'s formalization of forward-mode AD in a pure, higher-order functional language. 
The simplicity of the algorithm, and the modularity of~\citet{huot2020correctness}'s correctness argument via logical relations, makes it well-suited to extensions, like those we introduce in Sections~\ref{sec:combinator}-\ref{sec:general} and in Appendix~\ref{sec:extending} (see Table~\ref{tab:extensions}).


\subsection{Source Language for AD}
The grammar of our starting language is given in Figure~\ref{fig:02syntax}. 
Our types, terms, typing rules, and semantics are standard, but we recall them here to fix notation:

\textit{Types and terms} Our language includes numeric types\footnote{In our implementation, reals are represented by floating-point numbers, but we note that our correctness results do not account for any error introduced by floating-point approximations.} ($\RR$, $\RR_{>0}$, $\II = (0, 1)$, $\NN$), tuples $A \times B$, and function types $A \to B$. 
For terms, it features the standard constructs for building and accessing tuples, creating abstractions, and applying functions. We also provide  primitives for smooth numerical operations, like $\texttt{log} : \RR_{>0} \to \RR$.
Technically, we need multiple versions of each primitive ($+_\NN, +_\RR$), but we will suppress these subscripts when clear from context. 

\textit{Judgments.} A \textit{context} $\Gamma$ is a list associating variable names with their types (e.g., $\Gamma = x : \tau, y : \sigma$). 
The typing judgment $\Gamma \vdash t : \tau$ indicates that, in context $\Gamma$, $\ter$ is a well-typed term of type $\tau$. 
If $\vdash \ter : \tau$ (i.e., if $\ter$ is well-typed in an empty context), we call $t$ a closed term. 
The typing rules are standard.

\input{figures/02figures/syntax}

\textit{Semantics.} To each type $\tau$ we assign a set of values $\sem{\tau}$. To numeric types, we assign the corresponding sets of numbers. We interpret product and function types as products and functions on the interpretations of their arguments: $\sem{A \times B} = \sem{A} \times \sem{B}$, and $\sem{A \to B} = \sem{A} \to \sem{B}$. Then, for any term in context $\Gamma \vdash t : \tau$, we assign a meaning $\sem{\Gamma \vdash t : \tau} \in \sem{\Gamma} \to \sem{\tau}$, where $\sem{\Gamma}$ is the space of \textit{environments} mapping the variable names in $\Gamma$ to values of their corresponding types. For example, the meaning of a variable is the function that looks up that variable in the environment: $\sem{\Gamma \vdash x : \tau}(\rho) = \rho[x]$. When the context or the type is clear, we may omit them, writing $\sem{t}$ or $\sem{t : \tau}$. Using this shorthand, we give some more examples of term interpretations: $$\sem{t_1\,t_2}(\rho) = \sem{t_1}(\rho)(\sem{t_2}(\rho)) \,\,\, \sem{(t_1, t_2)}(\rho) = (\sem{t_1}(\rho), \sem{t_2}(\rho)) \,\,\, 
\sem{\lambda x. t}(\rho) = \lambda v. \sem{t}(\rho[x \mapsto v])$$

\begin{notation}
For closed terms $t$, we write $\sem{\ter}$ instead of $\sem{\ter}(\rho)$, where $\rho$ is the empty environment.
\end{notation}

\subsection{Forward-mode AD} We assume the user has written a program $\vdash t : \RR \to \RR$, and wishes to automate the
\revision{construction}
of a program $\vdash s : \RR \to \RR$ computing its (denotation's) derivative $\theta \mapsto \sem{t}'(\theta)$. Forward-mode AD translates the source program into a program representing the derivative in two steps:

\begin{itemize}[leftmargin=*]
    \item First, we apply a macro, \ad{\cdot}, to the user's program, yielding a new program $\vdash \ad{t} : \RR \times \RR \to \RR \times \RR$. This new program operates on \textit{dual numbers}, pairs of numbers representing the \textit{value} and \textit{derivative} of a computation. For any differentiable $h$, applying $\sem{\ad{t}}$ to the dual number $(h(\theta), h'(\theta))$ should yield $((\sem{t} \circ h)(\theta), (\sem{t} \circ h)'(\theta))$.
    \item Second, we output the program $s = \lambda \theta:\RR. \snd (\ad{t} \, (\theta, 1))$. Since $(\theta, 1) = (id(\theta), id'(\theta))$, we know $\sem{\ad{t}}(\theta, 1)$ returns a dual number representing $((\sem{t} \circ id)(\theta), (\sem{t} \circ id)'(\theta))  = (\sem{t}(\theta), \sem{t}'(\theta))$, whose second component we extract to return $\sem{t}$'s derivative. 
\end{itemize}
If the first step is done correctly, the correctness of the second step should be clear. Therefore, the content of the forward-mode AD algorithm mostly lives in the definition of the $\ad{\cdot}$ macro, and in the proof of its correctness. To emphasize this, we restate the property we need $\sem{\ad{t}}$ to satisfy if we want the second step above to follow:

\begin{definition}[correct dual-number derivative at $\RR$]
\label{def:correct-at-r}
Let $f : \RR \to \RR$ be a differentiable function. Then $f_D : \RR \times \RR \to \RR \times \RR$ is a correct dual-number derivative of $f$ if for all differentiable $h : \RR \to \RR$, $f_D(h(\theta), h'(\theta)) = ((f \circ h)(\theta), (f \circ h)'(\theta))$.
\end{definition}

Then the AD macro is correct if it computes these dual-number derivatives:

\begin{definition}[correctness of $\ad{\cdot}$]
\label{def:correct-macro}
The AD macro $\ad{\cdot}$ is correct if, for all closed terms $\vdash t : \RR \to \RR$, $\sem{\vdash \ad{t} : \RR \times \RR \to \RR \times \RR}$ is a correct dual-number derivative of $\sem{\vdash t : \RR \to \RR}$.
\end{definition}

\noindent{\textbf{Defining the AD macro.}} The AD macro $\ad{\cdot}$ itself is given in Figure~\ref{fig:02ad}. In every place that real numbers (of type $\RR$) appeared in the original program, they are now replaced by dual numbers (of type $\RR \times \RR$).
This affects the type of every term in the program, and the assumed types of any free variables in the context; we write $\ad{\tau}$ for the type that terms of type $\tau$ have after translation to use dual-numbers. Since reals are replaced by pairs of reals, we have $\ad{\RR} = \RR \times \RR$. Because functions into $\NN$ have no derivative information to track, $\ad{\NN} = \NN$. The derivative of a function into $\RR_{>0}$ or $\mathbb{I}$ may still be negative, so we set $\ad{\mathbb{K}} = \mathbb{K} \times \RR$. On product and function types, $\ad{\cdot}$ is defined recursively: $\ad{A \times B} = \ad{A} \times \ad{B}$ and $\ad{A \to B} = \ad{A} \to \ad{B}$.

\input{figures/02figures/ad}

When applied to a term $x_1 : \tau_1, \dots, x_n : \tau_n \vdash t : \tau$, AD produces a new term $x_1 : \ad{\tau_1}, \dots, x_n :  \ad{\tau_n} \vdash \ad{t} : \ad{\tau}$. The new term is mostly the same as the old term---only two things change:
\begin{itemize}[leftmargin=*]
    \item Constant real numbers $r$ are replaced with constant \textit{dual} numbers $(r, 0)$ with $0$ derivative.
    
    \item Primitives $c : \tau \to \sigma$ are translated into new, target-language primitives $c_D : \ad{\tau} \to \ad{\sigma}$. For each $c$, $c_D$ is a built-in dual-number derivative for the primitive $c$ (though we have yet to make this precise, except when $\tau = \sigma = \RR$).
\end{itemize}

\begin{notation}
In examples, we use the variable naming convention $dx = (x, \delta x)$ for dual numbers. 
\end{notation}

\noindent The semantics of the new primitives $c_D$ are given in Figure~\ref{fig:pure_constant_translation}. 
When $c : \RR \to \RR$, $\sem{c_D}$ has the form $$\sem{c_D} = \lambda (x, 
\delta x). (\sem{c}(x), \sem{c}'(x) \cdot \delta x),$$ which ensures that when run on a pair $(h(\theta), h'(\theta))$, it computes $(\sem{c}(h(\theta)), \sem{c}'(h(\theta)) \cdot h'(\theta))$, the second component of which uses the familiar chain rule from calculus to compute $(\sem{c} \circ h)'(\theta)$.

\input{figures/pure_constant_translation}

\subsection{Proof Technique: Reasoning about Correctness with Logical Relations}
We now review a powerful proof technique for reasoning about AD and showing that it is correct, based on logical relations~\citep{ahmed2006step,katsumata2013relating,krawiec2022provably,barthe2020versatility,huot2020correctness}. Although it may seem like overkill for such a simple algorithm, the technique will really shine when we try to make sense of highly non-standard extensions to forward-mode AD in Sections~\ref{sec:combinator}-\ref{sec:general}.
\\

\noindent \textbf{The challenge with simple proof by induction.} We might hope we could establish AD's correctness with a simple proof by induction: if AD is correct for each subterm in a program, it is correct for the whole program.\footnote{Technically, the induction is usually over the typing derivation of the term, and what we call ``subterms'' are really subtrees of the typing derivations corresponding to the premises of the bottom-most inference rule in the typing derivation.}  The challenge is that the notion of \textit{correctness} we gave in Definition~\ref{def:correct-at-r} applies only to translations of closed $\RR \to \RR$ programs. The meaning of an open subterm, $\sem{\Gamma \vdash t : \tau}$, will in general be a function from environments to values of type $\tau$ (which may not be $\RR$). The meaning of its translation, $\sem{\ad{\Gamma} \vdash \ad{t} : \ad{\tau}}$, will also be a function, from dual-number environments to values of type $\ad{\tau}$. Our simple correctness criterion about differentiable $\RR \to \RR$ functions cannot be applied here, and so it is unclear what inductive hypothesis a proof by induction would use. 
\\

\noindent \textbf{The logical relations approach.} The logical relations proof technique circumvents this issue by defining a {\it different} inductive hypothesis for each type. In proofs about AD, what this means is that we ultimately define a different notion of \textit{correct dual-number derivative} for functions between {\it any} two types in our language. Once we've done this, we can then proceed with an ordinary proof by induction: if the translation of each subterm $\Gamma \vdash t : \tau$ yields a correct dual-number derivative of the $\sem{\Gamma} \to \sem{\tau}$ function that $t$ denotes, then the translation of the enclosing term is also correct.

But how can we define correct dual-number derivatives for functions $f : \sem{\tau_1} \to \sem{\tau_2}$ between arbitrary types? Looking more closely at Definition~\ref{def:correct-at-r}, we can see that it phrases correctness for $f : \RR \to \RR$ functions in a slightly non-standard way: $f_D$ is a correct derivative if, when composed with a function $h$'s derivative, it \textit{preserves the relationship} that $h$ and $h'$ enjoyed, of ``being a derivative.'' This motivates a more general approach to defining correctness based on the idea of preserving the relationship between a function and its derivative. We proceed in two stages: 
\begin{itemize}[leftmargin=*]
    \item First, for each type $\tau$, we define a notion of derivative for $\RR \to \sem{\tau}$ functions: a relation between an $\RR \to \sem{\tau}$ function and an $\RR \to \sem{\ad{\tau}}$ function encoding what it means to be a derivative.
    
    \item Then, for arbitrary functions $f : \sem{\tau_1} \to \sem{\tau_2}$, we define correctness as the \textit{preservation} of this relationship: we look at what happens when $f$ (and its translation) are composed with $\RR \to \sem{\tau_1}$ functions (and their $\RR \to \sem{\ad{\tau_1}}$ derivatives), and check that what we get out are related $\RR \to \sem{\tau_2}$ and $\RR \to \sem{\ad{\tau_2}}$ functions. 
\end{itemize}

More precisely, in step 1, we define for each type $\tau$ a \textit{dual-number relation} $\rel_\tau$ encoding what it means to be a derivative of an $\RR \to \sem{\tau}$ function:
\begin{definition}[dual-number relation]
For a type $\tau$, a dual-number relation for $\tau$ is a relation $\rel_\tau$ over the sets $\RR \to \sem{\tau}$ and $\RR \to \sem{\ad{\tau}}$, that is, a subset $\rel_\tau \subseteq (\RR \to \sem{\tau}) \times (\RR \to \sem{\ad{\tau}})$. 
\end{definition}

For each $\tau$, we choose $\rel_\tau$ so that it relates continuously-parameterized $\sem{\tau}$ values (i.e., curves $\RR \to \sem{\tau}$) with continuously-parameterized dual-number values (curves $\RR \to \sem{\ad{\tau}}$) that use their dual-number storage to correctly track a local linear approximation to how the $\sem{\tau}$ value is changing when the real-valued parameter changes. This allows us to define correct dual-number derivatives at any type automatically:

\begin{definition}[correct dual-number derivative (general)]
\label{def:correct-dual-number}
Suppose that dual-number relations $\rel_\tau$ have been chosen for every $\tau$, and let $f : \sem{\tau_1} \to \sem{\tau_2}$. We say $f_D : \sem{\ad{\tau_1}} \to \sem{\ad{\tau_2}}$ is a correct dual-number derivative of $f$ if, for all $(g, g') \in \rel_{\tau_1}$, the functions $(f \circ g, f_D \circ g') \in \rel_{\tau_2}$. 
\end{definition}

This last definition is the one we will use as our inductive hypothesis in proving the AD macro correct overall. Although the proof by induction ultimately needs to cover all terms, the power of the technique is that the inductive steps are largely covered by existing, well-studied machinery, and so most of the AD-specific action happens at base types and primitives.

\subsection{Proof of Correctness for Forward-mode AD}

We now apply the logical relations technique from the last section to prove our AD macro correct. 
To do so, we define dual-number relations $\rel_\tau$ which encode a notion of \textit{derivative} at each type (Figure~\ref{fig:02log_rel}). For the reals (and continuous subsets of the reals), this notion coincides with the usual derivative: the relation $\rel_{\RR}$, for example, relates a differentiable function $f$ with the function $g(\theta) = (f(\theta), f'(\theta))$. 
For discrete types, such as $\NN$ , because the derivative will necessarily be zero, we can avoid storing a dual number.

\input{figures/02figures/logrel}


We then need to prove what is often called the {\it fundamental lemma} of a logical relations argument:
\begin{lemma}[Fundamental lemma]
\label{lem:fund1}
For every term-in-context $\Gamma \vdash t : \tau$, $\sem{\ad{t}}$ is a correct dual-number derivative of $\sem{t}$, with respect to the relations $\rel_\tau$ given in Figure~\ref{fig:02log_rel}. 
\end{lemma}

This is proved by induction on the typing derivation of $t$, and as our definitions of $\rel_\tau$ for product types and function types are completely standard, the inductive cases can be handled by standard logical relations machinery~\citep{huot2020correctness}. 
The only interesting cases are the base cases, where we must show that the interpretation of every primitive function $c$ has a correct dual-number derivative given by the interpretation of its translation $c_D$ (Definition~\ref{def:correct-at-r}).


The last step is to use our proof of the fundamental lemma to establish the more basic correctness criterion we outlined in Definition~\ref{def:correct-macro}:

\begin{theorem}[correctness of forward-mode AD~\citep{huot2020correctness}]
\label{thm:pure-correctness}
For all closed terms $\vdash t : \RR \to \RR$, $\sem{\lambda \theta : \RR. \snd (\ad{t}\,(\theta, 1))}$ is the derivative of $\sem{t}$. 
\end{theorem}
\revision{
\begin{proof}
By the fundamental lemma (\ref{lem:fund1}), $\sem{\ad{t}}\in \rel_{\RR \to \RR}$, so for functions $(f, g) \in \rel_\RR$, we have $(\sem{t} \circ f, \sem{\ad{t}} \circ g) \in \rel_\RR$. Take $f=id$ and  $g=\lambda \theta. (\theta, 1)$, and note that  $(f,g)\in\rel_\RR$, because $g(\theta) = (f(\theta), f'(\theta))$. 
Then $(\sem{t} \circ f,\sem{\ad{t}} \circ g)=(\sem{t},\lambda \theta. \sem{\ad{t}}(\theta, 1))\in \rel_\RR$, and so by the definition of $\rel_\RR$, for all $\theta\in\RR$, $\sem{\ad{t}}(\theta, 1) = (\sem{t}(\theta), \sem{t}'(\theta))$. 
Applying $\pi_2$ to extract just the second component, we have $\pi_2 \sem{\ad{\ter}}(\theta, 1)= \sem{\ter}'(\theta)$. 
As a function of $\theta$, the left-hand side is precisely $\sem{\lambda \theta:\RR.\snd (\ad{\ter}(\theta,1))}$, and the right-hand side is the derivative of $\sem{\ter}$.
\end{proof}
}



%% file: figures/02figures/syntax.tex
\begin{figure}[t]
\fbox{
  \parbox{0.97\textwidth}{
    \centering
    \vspace{-2mm}
    \begin{align*}
    \text{Smooth base types }\KK ::=\,& 
     \RR \mid \posreal \mid \II \\
        \text{Types }\type ::=\,& 
          \mathbf{1} \mid \NN \mid
         \KK \mid
          \type_1 \times \type_2 \mid
          \type_1 \to \type_2 \\
        \text{Terms } \ter ::=\,&
          () \mid
          \revision{r~(\in\KK)} \mid
          c \mid
          \colorbox{gray!15}{$c_\mathcal{D}$} \mid
          \var \mid
          (\ter_1, \ter_2) \mid
          \lambda \var:\type. \ter \mid
          \llet~\var = \ter_1 \,\texttt{in}\, \ter_2 \\
          & \mid
          \fst\,\ter \mid
          \snd\,\ter \mid
          \ter_1\,\ter_2 \\
        \text{Primitives } c ::=\,&
          + \mid - \mid \times \mid \div \mid
          \exps \mid \logs \mid \sins \mid \coss \mid \pows
    \end{align*}
    \revision{We write $\llet ~(x,y)=\ter_1~\iin~\ter_2$ as sugar for $\llet~x=\fst~\ter_1~\iin~ \llet~y=\snd ~\ter_1~\iin~\ter_2 $}
    }}
    \vspace{-3mm}
    \caption{Syntax of the \textit{deterministic} simply-typed $\lambda$-calculus for Sec.~\ref{sec:background}. $r$ ranges over real numeric constants, and $c$ over source-language primitive functions, each of which has an associated target-language dual-number derivative $c_\mathcal{D}$ (Fig.~\ref{fig:pure_constant_translation}).
    Gray highlights indicate syntax only present in the target language of the AD macro.
    }
    \label{fig:02syntax}
\vspace{-3mm}
\end{figure}

%% file: figures/02figures/ad.tex
\begin{figure}[tb]
\fbox{
  \parbox{0.97\textwidth}{
  

  
  AD on contexts:\quad  
     \begin{tabular}{ll}
          $\ad{\bullet}$ &= $\bullet$
     \end{tabular}
     \quad
    \begin{tabular}{ll}
        $\ad{\Gamma,\var:\type}$ &= $\ad{\Gamma},\var:\ad{\type}$
    \end{tabular}
\vspace{.1cm}

AD on types:\quad
\begin{tabular}{cc}
\begin{tabular}{ll}
  \ad{\mathbb{K}} &= $\mathbb{K}\times \RR$ \\
   \ad{\NN} &= $\NN$ 
\end{tabular}
&
\begin{tabular}{ll}
    \ad{\mathbf{1}} &= $\mathbf{1}$ \\
        \ad{\type_1\times\type_2} &= $\ad{\type_1}\times \ad{\type_2}$ \\
        \ad{\type_1\to\type_2} &= $\ad{\type_1}\to \ad{\type_2}$ 
    \end{tabular}

\end{tabular}

AD on pure expressions: 
 \vspace{-.1cm}
\begin{center}
\begin{tabular}{c|c|c}
\hspace{-.5cm}
\begin{tabular}{ll}
    \ad{\lambda\var:\type.\ter} &= $\lambda \var:\ad{\type}.\ad{\ter}$ \\
    \ad{\ter_1\ter_2} &= \ad{\ter_1}\ad{\ter_2} \\
    \ad{\llet ~\var=\ter_1~\iin~\ter_2} &= $\llet~ \var=\ad{\ter_1}~\iin~\ad{\ter_2}$ \\
    \ad{(\ter_1,\ter_2)} &= $(\ad{\ter_1},\ad{\ter_2})$\hspace{-2mm} \\
\end{tabular}
 \hspace{-.3cm}
       &
       \hspace{-.3cm}
      \begin{tabular}{ll}   
    \ad{\fst~\ter} &= $\fst~\ad{\ter}$ \\
      \ad{\snd~\ter} &= $\snd~\ad{\ter}$ \\ 
          \ad{r:\mathbb{K}} &= $(r,0)$ \\
     \ad{r:\NN} &= $r$ 
\end{tabular}
&
\hspace{-.3cm}
\begin{tabular}{ll}
    \ad{\var} &= $\var$ \\
    \ad{()} &= () \\
    \ad{c} &= $c_\mathcal{D}$ \\
    &
\end{tabular}
\end{tabular}
\end{center}
}}
    \vspace{-3mm}
\caption{The standard forward-mode AD  translation as a whole program transformation.
Note that the types of variables $x:\type$ (both free and bound) are changed to $x:\ad{\type}$.
For every primitive $c:\type$ of the source language, $c_\mathcal{D}:\ad{\type}$ is its built-in derivative.
$\ad{-}$ is a typed-translation: if $\Gamma\vdash \ter:\type$, then $\ad{\Gamma}\vdash \ad{\ter}:\ad{\type}$.
}
\label{fig:02ad}
\end{figure}

%% file: figures/pure_constant_translation.tex
\begin{figure}[t]
\fbox{
  \parbox{.97\textwidth}{
  \hspace{-.3cm}
  \begin{tabular}{c|c}
         \begin{tabular}{ll}
    $\sem{\exps_\der}(x,\delta x)$ &= $\llet~y=\exp~x$\\
    &$\quad\iin~(y, y\times \delta x)$ \\
    $\sem{\pows_\der}(dx,0)$ &= $(0,0)$
    \end{tabular}
       & 
     \begin{tabular}{ll}
        $\sem{(+_\RR)_\der}((x,\delta x),(y,\delta y))$ &=  $(x+y,\delta x+\delta y)$ \\
    $\sem{(\times_\RR)_\der}((x,\delta x),(y,\delta y))$ &= $(x\times y,\delta x\times y+x\times \delta y)$ \\
    $\sem{\pows_\der}((x,\delta x),n+1)$ &= $\llet~y= \pows(x,n)~\iin$\\
    &\quad$(x\times y,(n+1)\times y\times \delta x)$
     \end{tabular}
  \end{tabular}
  }}
      \vspace{-3mm}
\caption{Dual number interpretation of deterministic primitives. On the right we use syntax for simplicity, but it should be understood as metalanguage syntax. \revision{We write $dx=(x,\delta x)$ for dual numbers. Note that each primitive implements the chain rule, multiplying the derivative with respect to $x$ by $\delta x$.}}
\label{fig:pure_constant_translation}
\end{figure}

%% file: figures/02figures/logrel.tex
\begin{figure}[tb]
\fbox{
  \parbox{.97\textwidth}{
  \begin{center}
      Dual-Number Logical Relations $\rel_\tau$ for the Deterministic Language (\S\ref{sec:background})
  \end{center}
  \small{
  \begin{align*}
      \rel_\RR &= \big\{(f : \RR \to \RR, g : \RR \to \RR \times \RR) \mid f \text{ differentiable} \wedge \forall \theta \in \RR. g(\theta) = (f(\theta), f'(\theta))\big\} \\
      \rel_{\posreal} &= \big\{(f : \RR \to \RR_{>0}, g : \RR \to \posreal \times \RR) \mid (\iota_{\posreal} \circ f, \langle \iota_{\posreal}, id\rangle \circ g) \in \rel_\RR\big\} \\
      \rel_{\II} &= \big\{(f : \RR \to \II, g : \RR \to \II \times \RR) \mid (\iota_\II \circ f, \langle \iota_\II, id\rangle \circ g) \in \rel_\RR\big\} \\
      \rel_\NN &= \big\{(f : \RR \to \NN, g : \RR \to \NN) \mid f \text{ is constant } \wedge f = g\big\} \\
      \rel_{\tau_1 \times \tau_2} &= \big\{(f : \RR \to \sem{\tau_1 \times \tau_2}, g : \RR \to \sem{\ad{\tau_1} \times \ad{\tau_2}}) \mid \\
      &\qquad\qquad\qquad (\pi_1 \circ f, \pi_1 \circ g) \in \rel_{\tau_1} \wedge (\pi_2 \circ f, \pi_2 \circ g) \in \rel_{\tau_2}\big\} \\
      \rel_{\tau_1 \to \tau_2} &= \big\{(f : \RR \to \sem{\tau_1 \to \tau_2}, g : \RR \to \sem{\ad{\tau_1 \to \tau_2}}) \mid \\
      &\qquad\qquad\qquad\forall (j, k) \in \rel_{\tau_1}. (\lambda r. f(r)(j(r)), \lambda r. g(r)(k(r))) \in \rel_{\tau_2}\big\}
      \end{align*}
          \vspace{-4mm}
  }}}
      \vspace{-3mm}
\caption{Definition of the dual-number logical relation at each type. \revision{ $\iota_\KK:\KK\to\RR$ is the canonical injection, for every smooth base type $\mathbb{K}$.}}
\label{fig:02log_rel}
\end{figure}

%% file: 03combinator.tex
\section{Warm-up: differentiating a probabilistic combinator DSL}
\label{sec:combinator}

Now that we have set the stage, we can begin to introduce our main characters: new types and terms for probabilistic programming. In this section, we tackle only a small warm-up extension: we study the most basic setting where probability arises, a simple and restrictive DSL for composing probability distributions with combinators. Unlike general probabilistic programming languages, which we study in Sections~\ref{sec:discrete}-\ref{sec:general}, the DSL in this section does not allow for arbitrary sequencing of probabilistic computations. Despite the simplicity of this setting, our development here provides an important foundation for the fancier extensions we will add next. 


    
\subsection{The Type $\eRR$ of Random Real Numbers}
\label{sec:eRR_intro}

Typically, in differentiable programming languages, users aim to construct a closed expression $\vdash t : \RR \to \RR$,
implementing a differentiable function whose derivative they wish to compute.
But for the rest of this paper, we consider a different workflow: instead of constructing
a program of type $\RR \to \RR$, the user constructs a program of type $\RR \to \eRR$, where $\eRR$ is a new type of \textit{random} real numbers, whose {\it expected values} are the quantities of interest. We call the values of $\eRR$ \textit{unbiased real-valued estimators}: sampling them yields unbiased estimates of the true values we care about.

\begin{definition}[real-valued estimator]
\label{def:eRRdef}
We denote by $\eRR$ the set of \textit{unbiased real-valued estimators}: probability measures $\mu$ on the measurable space $(\RR, \mathcal{B}(\RR))$. If $\mathbb{E}_{x \sim \mu}[x] = \int x \mu(dx)$ exists, i.e. is finite and equal to some number $r \in \RR$, we say $\mu$ estimates (or is an unbiased estimator of) $r$.
\end{definition}

\begin{remark}
A distribution $\mu \in \eRR$ need not have a density function, and may be supported on a finite, countable, or uncountable set of reals.
For example, the Dirac distribution, $\delta_{r}$, which assigns all its mass to the number $r$, is an unbiased estimator of $r$, as is the Gaussian distribution $\mathcal{N}(r, 1)$.
\end{remark}

Although the user's program $\vdash \widetilde{t} : \RR \to \eRR$ denotes a map into the space of probability distributions, the function they wish to differentiate is the loss function $\loss : \RR \to \RR = \lambda \theta. \mathbb{E}_{x \sim \sem{\widetilde{t}}(\theta)}[x]$, if $\loss$ is well-defined (i.e., if the expectation always exists). As illustrated in Figure~\ref{fig:ad-diagram}, applying ADEV to the program $\widetilde{t}$, we get a new program $\vdash \widetilde{s} : \RR \to \eRR$ that estimates $\loss'$, the derivative of the loss. It will be useful to have a word for the relationship between $\sem{\widetilde{t}}$ and $\sem{\widetilde{s}}$; we coin \textit{unbiased derivative}:

\begin{definition}[unbiased derivative]
\label{def:unbiased-deriv}
Given a function \revision{$\widetilde{f} : U \to \eRR$, for $U \subseteq \RR$}, suppose $\loss : \revision{U} \to \RR = \lambda \theta.  \mathbb{E}_{x \sim \widetilde{f}(\theta)}[x]$ is well-defined and differentiable. We say that a function $\widetilde{g} : \revision{U} \to \eRR$ is an unbiased derivative of $\widetilde{f}$ if for all $\theta \in \revision{U}$, $\widetilde{g}(\theta)$ estimates $\loss'(\theta)$, that is,  
$$\mathbb{E}_{x \sim \widetilde{g}(\theta)}[x] = \loss'(\theta) = \frac{\text{d}}{\text{d}\theta}\mathbb{E}_{x \sim \widetilde{f}(\theta)}[x].$$
Note that estimator-valued functions may have many unbiased derivatives. 
\end{definition}


\subsection{Syntax and Semantics of the Combinator DSL}

We now present our combinator language, a toy DSL for constructing values of type $\eRR$. The syntax is given in Fig.~\ref{fig:03syntax}, and the semantics in Fig.~\ref{fig:03someprimitives} (some primitives deferred to Fig.~\ref{fig:03primitives}).


\input{figures/03figures/syntax}

\input{figures/03figures/some_primitives}

Beyond the new base type $\eRR$, our extended source language exposes a small collection of combinators for implementing stochastic loss functions. 
The $\minibatch$ primitive constructs a probabilistic estimator of a large sum $\sum_{i=1}^M f(i)$, that works by subsampling $m \ll M$ indices $(i_1, \dots, i_m)$ uniformly at random, and evaluating $f$ only at those indices, \revision{returning $\frac{M}{m} \sum_{j=1}^m f(i_j)$. The $\exact$ primitive constructs the trivial deterministic estimator of a real value that returns the value with probability 1.} Our other new primitives \textit{transform} existing estimators, creating a new estimator with expected value equal to some function (e.g., a sum, product, or exponentiation) of the inputs' expected values. We give one example ($\etimes$) in Figure~\ref{fig:03someprimitives}; the full semantics can be found in Appendix~\ref{appx:figures}, Figure~\ref{fig:03primitives}.

\subsection{ADEV for the Combinator DSL: Differentiating through $\eRR$}
\label{sec:estimator_ad}
\input{figures/03figures/ad}

Suppose a user has written a program $\vdash \widetilde{t} : \RR \to \eRR$, representing a stochastic estimator $\sem{\widetilde{t}}$ of a loss function $\loss(\theta) = \mathbb{E}_{x \sim \sem{\widetilde{t}}(\theta)}[x]$. Our algorithm, ADEV, differentiates $\loss$ by constructing a program $\vdash \widetilde{s} : \RR \to \eRR$ that implements an \textit{unbiased derivative} of $\sem{\widetilde{t}}$ (Definition~\ref{def:unbiased-deriv}), in two steps:

\begin{itemize}[leftmargin=*]
    \item First, we will apply an extended version of the AD macro $\ad{\cdot}$ to the user's program, yielding a new program $\ad{\widetilde{t}} : \RR \times \RR \to \deRR$. It accepts a \textit{dual number} as input, and instead of estimating a single real value, estimates a dual number value: the type $\deRR$ denotes the set of probability distributions over {\it pairs} of reals. The key correctness property in this extended setting is that if we are given as an input dual number $(h(\theta), h'(\theta))$ for some differentiable function $h$, then $\sem{\ad{\widetilde{t}}}$ should send it to an estimator of the dual number $((\loss \circ h)(\theta), (\loss \circ h)'(\theta))$. 
    
    \item Second, we output the term $\widetilde{s} = \lambda \theta:\RR. \snd_* (\ad{\widetilde{t}}(\theta, 1))$. As before, since $(\theta, 1) = (id(\theta), id'(\theta))$, we know ${\pi_2}_* \circ \sem{\ad{\widetilde{t}}}$ returns a new estimator that estimates $(\loss \circ id)'(\theta)$, as desired. 
\end{itemize}

As in the standard AD algorithm from Section~\ref{sec:background}, the correctness of the second step follows directly if we can prove the first step works correctly, so we now turn to extending $\ad{\cdot}$.
\\

\noindent\textbf{Defining the ADEV macro at the type level.} Our AD macro from Section~\ref{sec:background} had one key job: replacing every real number flowing through the program with a dual number, and all real number operations with dual number operations. In doing so, it translated terms of type $\tau$ to terms of type $\ad{\tau}$\textemdash the dual-number version of the type $\tau$. We now have a type of \textit{estimated} real numbers, $\eRR$, and
the dual-number version of an estimator should be an \textit{estimator} of a dual number:

\begin{definition}[unbiased dual-number estimator]
\label{def:unbiased_estimator}
We denote by $\deRR$ the set of unbiased dual-number estimators: probability measures $\mu$ on $(\RR \times \RR, \mathcal{B}(\RR \times \RR))$. If $\mathbb{E}_{(x, \delta x) \sim \mu}[x] = r$ and $\mathbb{E}_{(x, \delta x) \sim \mu}[\delta x] = \delta r$ for finite real numbers $r$ and $\delta r$, we say that $\mu$ estimates the dual number $(r, \delta r)$. 
\end{definition}
This type, which appears in the target language in Figure~\ref{fig:03syntax} but not in our source language, represents random processes for estimating a dual number. Note that it allows for  the two components of the estimate to depend on the same random choices: it is a distribution over pairs, not a pair of distributions. We set $\ad{\eRR} := \deRR$.
\\

\noindent\textbf{Defining the ADEV macro at the term level.} To extend the macro $\ad{\cdot}$ from Section~\ref{sec:background} to handle our extended language, we need to say what it does on each new term. But the only new terms we have added to our source language are the new primitives, like $\minibatch$ and $exp^{\eRR}$. Thus, the only new behavior we need to specify is how to translate each primitive\textemdash in other words, we need to attach to each new primitive $c : \tau \to \sigma$ a custom built-in derivative $c_\mathcal{D}$. We give one example in Figure~\ref{fig:03someprimitives}, with the full list in Appendix~\ref{appx:figures}, Figure~\ref{fig:03primitives}.

If a primitive $c : \tau \to \eRR$ builds an estimator of some loss, the goal of $c_\mathcal{D} : \ad{\tau} \to \deRR$ is to estimate both the loss and the derivative of the loss. In many cases, this is quite straightforward. For example, the primitive $\times^{\eRR}$ in Fig.~\ref{fig:03someprimitives} estimates the product of the two numbers $x$ and $y$ that its input arguments $\widetilde{x}$ and $\widetilde{y}$ estimate. Its built-in derivative $\times^{\eRR}_\der$ does the same but with dual numbers: it independently generates estimates $dr = (r, \delta r)$ of $(x, \delta x)$ and $ds = (s, \delta s)$ of $(y, \delta y)$, then returns $(r, \delta r) \times_\der (s, \delta s) = (rs, s\delta r + r\delta s)$. Because $r$ and $s$ are independent random variables, their product's expectation is the product of their expectations, $\mathbb{E}[rs] = xy$. And by linearity of expectation, $\mathbb{E}[s\delta r + r \delta s] = \mathbb{E}[s\delta r] + \mathbb{E}[r \delta s]$; exploiting again the fact that $s$ and $\delta r$ are independent (and likewise for $r$ and $\delta s$), we obtain the desired result $y\delta x + x \delta y$. 
With such built-in derivatives for all the primitives (Figure~\ref{fig:03primitives}), the ADEV macro now covers the new source language. 
\\

\subsection{Correctness Criterion for ADEV}
Before trying to prove ADEV correct, let's formulate a definition of correctness for the programs it produces (an updated version of Definition~\ref{def:correct-at-r}):

\begin{definition}[correct dual-number derivative at $\eRR$]
\label{def:correctness-at-err}
Let $\widetilde{f} : \RR \to \eRR$ be an estimator-valued function, and suppose that the map $\loss : \RR \to \RR$ that sends $\theta$ to $\mathbb{E}_{x \sim \widetilde{f}(\theta)}[x]$ is well-defined (i.e., the expectation exists for all $\theta$) and differentiable. Then $\widetilde{f_D} : \RR \times \RR \to \deRR$ is a correct dual-number derivative of $\widetilde{f}$ if for all differentiable $h : \RR \to \RR$, the dual number estimator $\widetilde{f_D}(h(\theta), h'(\theta))$ estimates the dual number $((\mathcal{L} \circ h)(\theta), (\loss \circ h)'(\theta))$.
\end{definition}

Then our macro should compute these correct dual-number derivatives:

\begin{definition}[correctness of $\ad{\cdot}$ (ADEV)]
\label{def:correctness-macro-extended}
The ADEV macro $\ad{\cdot}$ is correct if for all closed terms $\vdash \widetilde{t} : \RR \to \eRR$, $\sem{\vdash \ad{\widetilde{t}} : \RR \times \RR \to \deRR}$ is a correct dual-number derivative of $\sem{\vdash \widetilde{t} : \RR \to \eRR}$. 
\end{definition}


This notion of correctness is \textit{intensional}~\citep{lee2020correctness}, in that there is more than one correct dual-number derivative of a function $\widetilde{f} : \RR \to \eRR$. 
In practice, which derivative a user gets will depend on which primitives a user's program invokes, intuitively because ``all the action'' in forward-mode AD happens at the primitives (each primitive is equipped with a built-in derivative, and these are composed to implement a program's derivative). By providing users with a library of primitives, some of which have the same meaning but different built-in derivatives, we give users a compositional way to explore the space of gradient estimation strategies \revision{(see Section~\ref{sub:syntax_disc_ppl}).}

\subsection{Proving the ADEV Algorithm Correct} To extend Section~\ref{sec:background}'s proof to cover the ADEV algorithm, we need to:
\begin{enumerate}[leftmargin=*]
    \item Define a dual-number relation $\rel_{\eRR} \subseteq (\RR \to \eRR) \times (\RR \to \deRR)$, characterizing when a function that estimates dual numbers is a correct derivative of a function that estimates reals. \textit{\textbf{Intuition:}} this step is about defining a notion of derivative for $\RR \to \eRR$ functions; we will base our choice on our earlier Definition~\ref{def:correctness-at-err}.
    
    \item Prove an updated version of the fundamental lemma (Lemma~\ref{lem:fund1}) for the extended language, with respect to all the old relations $\rel_\tau$, but also the new relation $\rel_{\eRR}$. \textit{\textbf{ Intuition:}} this step updates an inductive proof now that we have more base cases (new primitives). Luckily, the inductive steps don't change at all, and it suffices to check the base cases. Concretely, this means showing that each of our new primitives has a correct built-in derivative, using the definition of `correctness' arising from our choice in step (1) together with Definition~\ref{def:correct-dual-number}. 
    
    \item Prove an updated version of Theorem~\ref{thm:pure-correctness}, to show how correctness of ADEV follows from the updated fundamental lemma. \textbf{\textit{Intuition:}} This step shows that if our program estimates dual numbers correctly (implied by step 2), then our ``wrapper'' that extracts the second component of the dual number to return as the derivative is correct. As in Section~\ref{sec:background}, this step is straightforward.
\end{enumerate}

\input{figures/03figures/logrel}

\noindent For the first step, we extend our definitions of dual-number relations $\rel_\tau$ from Section~\ref{sec:background} to cover our new type, $\eRR$. The new relation $\rel_{\eRR}$ is presented in Figure~\ref{fig:03logrel}, and captures what it means to be a correct dual-number derivative estimator. Since our logical relations have changed, we need to reprove the fundamental lemma:

\begin{lemma}[fundamental lemma (revised with $\eRR$)]
\label{lem:fund2}
For every term $\Gamma \vdash t : \tau$, $\sem{\ad{t}}$ is a correct dual-number derivative of $\sem{t}$, with respect to the relations $\rel_\tau$ defined at each type (incl. $\eRR$). 
\end{lemma}
\begin{proof}
The proof is the same inductive proof we used for Lemma~\ref{lem:fund1}, except that there are now new base cases: we must show that the interpretation of every new primitive function $c$ has a correct dual-number derivative given by the interpretation of its translation $c_D$.
\begin{itemize}[leftmargin=*]
    \item For $\exact : \RR \to \eRR$, we must check that $\exact_\mathcal{D} (h(\theta), h'(\theta))$ estimates $(h(\theta), h'(\theta))$ for differentiable $h$ (which it clearly does: it returns its input dual number exactly).
    
    
    \item For $exp^{\eRR}$, $\times^{\eRR}$ and $+^{\eRR}$, implementing $n$-ary operations $op$ on estimators, we must check that for all $n$-tuples differentiable functions $(r_1,\dots,r_n) : \RR \to \RR$, if $\widetilde{dr}_i$ estimates $(r_i(\theta), r'_i(\theta))$, then $op(\widetilde{dr}_1, \dots, \widetilde{dr}_n)$ estimates $(op(r_1(\theta), \dots, r_n(\theta)), \frac{d}{d\theta} op(r_1(\theta), \dots, r_n(\theta)))$. 
    
    \item For the primitive $\minibatch$, we must check that if $df : \NN \to \RR \times \RR$ maps each natural number $i$ to the dual number $(f_i(\theta), f_i'(\theta))$ for some differentiable function $f_i$, then $\minibatch_\mathcal{D}\,M\,m\,df$ estimates the dual number $(\sum_{i=1}^M f_i(\theta), \sum_{i=1}^M f'_i(\theta))$.

\end{itemize}
Once we check all these primitives (given in Figure~\ref{fig:03primitives}), the proof is done.
\end{proof}

Finally, we can conclude correctness of the ADEV algorithm on the extended language:
\begin{theorem}[correctness of ADEV on the combinator DSL]
\label{thm:err-correctness}
For all closed terms $\vdash \widetilde{t} : \RR \to \eRR$, $\sem{\lambda \theta : \RR. \snd_* (\ad{\widetilde{t}}\,(\theta, 1))}$ is an unbiased derivative of $\sem{\widetilde{t}}$. 
\end{theorem}

%% file: figures/03figures/syntax.tex
\begin{figure}[t]
\fbox{
  \parbox{.97\textwidth}{
    \centering
        \vspace{-2mm}
    \begin{align*}
        \text{Types }\type ::=\,& \ldots 
            \mid \eRR
            \mid \colorbox{gray!15}{$\deRR$}
            \\
        \text{Primitives } c ::= \,& \ldots 
        \mid \minibatch : \NN\to\NN\to(\NN\to\RR)\to\eRR
        \mid \colorbox{gray!15}{$\fst_*, \snd_*:\eRR_\der\to\eRR$} 
        \\
        &\mid \eplus, \etimes : \eRR\times\eRR\to\eRR 
        \mid \eexp : \eRR\to\eRR
        \mid \exact:\RR \to \eRR
    \end{align*}
        \vspace{-4mm}
    }}
        \vspace{-3mm}
    \caption{Syntax for the Probabilistic Combinator DSL (\S\ref{sec:combinator}), as an extension to Fig.~\ref{fig:02syntax}. \revision{ Gray highlights indicate syntax only present in the target language of the AD macro.}
    }
    \label{fig:03syntax}
\vspace{-3mm}
\end{figure}

%% file: figures/03figures/some_primitives.tex
\begin{figure}[t]
\fbox{
  \parbox{.97\textwidth}{
  \begin{center}
    \vspace{-1mm}
\begin{tabular}{c|c}
    Semantics of types:  & 
     \hspace{-5mm}
    Example primitive and its built-in derivative: \\
      \begin{tabular}{ll}
      $\sem{\eRR}$ &= $\{\mu~|~\mu$ a probability measure\\
      &\quad  on $(\RR,\mathcal{B}(\RR))\}$ \\
      $\sem{\eRR_\der}$ &= $\{\mu~|~\mu$
      a probability measure \\
      &\quad on $(\RR\times\RR,\mathcal{B}(\RR\times \RR))\}$
  \end{tabular}
  &
   \hspace{-3mm}
  \begin{tabular}{ll}
  \begin{minipage}{.25\linewidth}
  \begin{algorithm}[H]
\DontPrintSemicolon
\SetKwProg{Fn}{}{:}{end}
\Fn{$\etimes(\widetilde{x}:\eRR,\widetilde{y}:\eRR)$}{
$r\sim\widetilde{x}$\;
$s\sim\widetilde{y},$\; 
\Return $r\times s$
}
\end{algorithm}  
  \end{minipage}
       & 
       \hspace{-7mm}
       \begin{minipage}{.31\linewidth}
       \begin{algorithm}[H]
\DontPrintSemicolon
\SetKwProg{Fn}{}{:}{end}
\Fn{$\etimes_\der(\widetilde{dx}:\deRR,\widetilde{dy}:\deRR)$}{
$dr\sim \widetilde{dx}$\;
$ds\sim \widetilde{dy},$\; 
\Return $dr\times_\der ds$
}
\end{algorithm}  
  \end{minipage}
      \vspace{-2mm}
  \end{tabular}
\end{tabular}      
   \end{center} 
 }}
     \vspace{-3mm}
\caption{Semantics of the new types for the Combinator DSL and an example of a new primitive. }
\label{fig:03someprimitives}
\end{figure}

%% file: figures/03figures/ad.tex
\begin{figure}[tb]
\fbox{
  \parbox{.97\textwidth}{
  \centering
\begin{tabular}{c|c|c}
     \begin{tabular}{ll}
        \ad{\eRR} &= $\eRR_\mathcal{D}$ \\
        \ad{exp^{\eRR}} &= $exp^{\eRR}_\mathcal{D}$
 \end{tabular} 
     & 
     \begin{tabular}{ll}
    \ad{\minibatch} &= $\minibatch_\mathcal{D}$ \\
    \ad{\exact} &= $\exact_\mathcal{D}$
    \end{tabular}
    &
    \begin{tabular}{ll}
         \ad{+^{\eRR}} &= $+^{\eRR}_\mathcal{D}$ \\
    \ad{\times^{\eRR}} &= $\times^{\eRR}_\mathcal{D}$
    \end{tabular}
\end{tabular}
}}
    \vspace{-3mm}
\caption{ADEV macro for the Combinator DSL (\S\ref{sec:combinator}), extending Fig.~\ref{fig:02ad}
}
\label{fig:03ad}
\end{figure}

%% file: figures/03figures/logrel.tex
\begin{figure}[tb]
\fbox{
  \parbox{.97\textwidth}{
      \vspace{-2mm}
      
      \begin{align*}
      \hspace{-4mm}\rel_{\eRR} &= \{(\widetilde{f} : \RR \to \eRR, \widetilde{g} : \RR \to \deRR) \mid \loss := \theta \mapsto \mathbb{E}_{x \sim \widetilde{f}(\theta)}[x] \text{ is well-defined and differentiable}\\ 
      \hspace{-4mm}&\quad\wedge\text{ for all } \theta\in\RR,~  \widetilde{g}(\theta) \text{ estimates } (\loss(\theta), \loss'(\theta)) \text{ (Def.~\ref{def:unbiased_estimator})}\}
        \end{align*}
            \vspace{-4mm}
  }}
      \vspace{-2mm}
\caption{Logical relation for the Probabilistic Combinator DSL}
\label{fig:03logrel}
\end{figure}

%% file: 04discrete.tex
\section{Differentiating Expected Values of Discrete Probabilistic Programs}
\label{sec:discrete}

In this section, we develop one of the most important ideas in the paper: how to differentiate \textit{expressive} probabilistic programs, with sequencing and branching, compositionally. 
For now we make the simplifying assumption that primitive probability distributions have finite support (e.g., a coin flip, which can take only two possible values). But this is only to simplify the proofs; when we add continuous distributions in Section~\ref{sec:continuous}, the ADEV algorithm itself won't change, only the theory. 

 

\subsection{Syntax and Semantics of the Discrete Probabilistic Programming Language}
\label{sub:syntax_disc_ppl}

Figure~\ref{fig:04syntax} gives the syntax of this section's language. It is an extension of the language from Section~\ref{sec:combinator} with two new features: \textit{sequencing} of probabilistic computations, and \textit{branching}. This greatly increases the expressiveness of the language; e.g., even without the advances of Sections~\ref{sec:continuous} and~\ref{sec:general}, we can already express and differentiate the motivating example program in Figure~\ref{fig:flip-example}.
\input{figures/04figures/syntax}

\input{figures/discrete_semantics}

\textit{Types:} We introduce a type $\mathbb{B}$ of Booleans, and for each  type $\tau$, a new monadic type $\pmonad \tau$, of (finitely supported) probability distributions over $\sem{\tau}$. In our semantics, we need to fix a way of representing these distributions, and we choose $\sem{\pmonad \tau}$ to be the set of probability mass functions $\sem{\tau} \to [0, \infty)$ with finitely many non-zero values, which form a monad over \textbf{Set}. For $\mu \in \sem{\pmonad\tau}$, we write $\text{supp}(\mu) \subseteq \sem{\tau}$ for the finite subset of inputs at which it is non-zero.\footnote{We emphasize that the choice to represent probability distributions as mass functions in our semantics does \textit{not} mean that the \textit{operational} meaning of a $\pmonad\tau$ term is a mass function evaluator: we think of $\pmonad\tau$ terms as probabilistic programs, which are tractable to \textit{run} (i.e., to draw samples from), but for which it may be extremely expensive to evaluate probabilities.}

\textit{Terms:} Since we now have Booleans, we can introduce an $\iif$ statement. We will not introduce discontinuous comparators like $\leq$ until Section~\ref{sec:general}, so for now the $\iif$ statement is primarily useful for branching on the outcomes of random coin flips. The new term $\texttt{flip} : \mathbb{I} \to \pmonad\mathbb{B}$ (which comes in two flavors, $\flipenum$ and $\flipreinforce$, for reasons we defer to Section~\ref{sec:ad-discrete}) is the key primitive probability distribution. It is parameterized by a number $\theta \in (0, 1)$, and returns $\True$ with probability $\theta$ and $\False$ with probability $1-\theta$. More complex probability distributions can be constructed using the Haskell-inspired $\haskdo\,  \{ x \gets t; m\}$ syntax: it builds a new probabilistic program that first samples $x$ from $\sem{t}$, then runs $\sem{m}$ in an environment extended with the sampled $x$. This allows us to, for example, sequence two coin flips, where the second flip's probability depends on the outcome of the first:
$\haskdo \, \{ b_1 \gets \flipreinforce \, 0.5; \, b_2 \gets \flipreinforce (\iif~b_1~\then~0.2~\eelse~0.4);\, \return~(b_1 \wedge b_2)\} : \pmonad\mathbb{B}.$ 
\\

\noindent\revision{\textbf{The expectation operator.} We now have two types denoting probability distributions over reals:
\begin{itemize}[leftmargin=*]
    \item $\eRR$, the type of \textit{estimators}\textemdash arbitrary probability distributions over $\RR$, that can be composed using the combinator DSL from Section~\ref{sec:combinator}.
    
    \item $\pmonad \RR$, the type of \textit{monadic probabilistic programs returning reals}, which may be composed arbitrarily with downstream probabilistic computation. 
\end{itemize}
The \textit{expectation operator} $\mbe : \pmonad \RR \to \eRR$ casts a probabilistic program returning random real numbers into an \textit{unbiased estimator} of the original program's \textit{expectation}. 
ADEV can then be applied to the resulting esitmator to construct an estimator of its expectation's derivative.} 
Returning to the example from Fig.~\ref{fig:flip-example}, the term $\loss$ has type $\RR \to \eRR$, and is thus a suitable `main function' for ADEV to differentiate\textemdash but it is \textit{constructed} by applying $\mbe$ to a term of type $\pmonad \RR$. Note that because our language has primitives that transform and combine $\eRR$ terms, the user's main function need not be a simple expectation of a probabilistic program\textemdash it can also be the $exp^{\eRR}$ of an expectation, for example, or the $+^{\eRR}$ of two expectations.


\subsection{Differentiating the Probabilistic Language: Three False Starts}

We now face the challenge of extending our ADEV macro $\ad{\cdot}$ to handle this much more expressive probabilistic language. The first step
is to define the macro's action on each new \textit{type} in our language. The Booleans are simple enough ($\ad{\mathbb{B}} = \mathbb{B}$, since they do not track derivative information), but the monadic types $\pmonad\tau$ pose a real hurdle. If a source language term has type $\pmonad\tau$, what type should its translation have? In this section, we first explore three superficially appealing but ultimately problematic answers, before introducing our solution in Section~\ref{sec:ad-discrete}. 
\\

\noindent \textbf{False Start 1: Probabilistic dual number programs.} In Section~\ref{sec:combinator}, we saw how for simple probabilistic computations over $\RR$, it sufficed to translate them to simple probabilistic computations over $\RR \times \RR$. 
Naively, we might wonder whether this approach works at all types: can we define $\ad{\pmonad\tau} = \pmonad(\ad{\tau})$? 
Unfortunately, this simple, structure-preserving choice doesn't work. 
Values of type $\ad{\tau}$ must track both a primal value of type $\tau$, and the way that value depends continuously on an external parameter. At $\RR$, for instance, this is done explicitly using a pair of reals.
But now consider a program of type $P\,\mathbb{B}$, for example $\texttt{flip} \, \theta$. Even though $\mathbb{B}$ is discrete, probability distributions over Booleans \textit{may} depend continuously on parameters, and so $\ad{\pmonad\mathbb{B}}$ values must somehow track both the primal value (a distribution over $\mathbb{B}$) and a dual value (how that distribution changes when $\theta$ changes). But if we choose $\ad{\pmonad\tau} := \pmonad \ad{\tau}$, then we get that $\ad{\pmonad\mathbb{B}} = \pmonad\mathbb{B}$, which can only track the primal value. This loss of information is one of the key reasons why Standard AD can fail when naively applied to probabilistic programs, as depicted in Figure~\ref{fig:flip-example}.
\\

\noindent \textbf{False Start 2: Differentiating the mass function semantics.} Our semantics interprets a term of type $\pmonad\tau$ as a mass function, mapping values of $\sem{\tau}$ to non-negative real probability values.
Viewed in this light, a primitive like $\texttt{flip}$ is actually a real-valued function, in this case from $\mathbb{I} \to \mathbb{B} \to [0, \infty)$.
We already know how to make AD work compositionally with functions of this type; would it work to set $\ad{\pmonad\tau} := \ad{\tau \to \RR} = \ad{\tau} \to \RR \times \RR$? The idea would be that applying AD to a probabilistic program $\vdash p : \RR \to \pmonad \RR$ would give us the derivative of its mass function, $\lambda (\theta, x). \frac{d}{d\theta} \sem{p}(\theta)(x)$. To get derivatives of an expectation, $\frac{d}{d\theta} \sum_{x \in \text{supp}(\sem{p}(\theta))} \sem{p}(x, \theta) \cdot x$, we would then differentiate term-by-term, using the automatically computed derivative. Unfortunately, it is not clear how to handle the fact that $\sem{p}$'s \textit{support} can depend on $\theta$, and relatedly, that the mass functions of probabilistic programs are not always differentiable. Consider, for example, the program $t = \lambda \theta : \RR. \haskdo \,\{b \gets \texttt{flip}\,0.3; \iif~b~\then~\return\,\theta~\eelse~\return~(2\theta)\}$, whose support $\{\theta, 2\theta\}$ depends on $\theta$ and whose mass function, $\sem{t}(\theta)(r) = 0.3[\theta=r] + 0.7[2\theta = r]$, is not differentiable with respect to $\theta$. 
\\

\noindent \textbf{False Start 3: Differentiating the expectation directly.} Ultimately, we only need to differentiate terms of type $\pmonad\tau$ because we care about how they affect the \textit{expectation} of the program they are used within.
This suggests that when we translate a term $p$ of type $\pmonad\tau$, we 
might wish to produce a term that tells us not how $\sem{p}$ itself depends on a parameter $\theta$, but how \textit{expectations} with respect to the distribution $\sem{p}$ depend on the parameter $\theta$. That is, 
can we differentiate the expectation $\mathbb{E}_{x \sim \sem{p}}[f(x)] = \sum_{x \in \text{supp}(\sem{p})} \sem{p}(x) \cdot f(x)$, for a formal expectand $f : \tau \to \RR$?

One way to make good on this intuition is to set $\ad{\pmonad\tau} := \ad{(\tau \to \RR) \to \RR}$. Here, we understand a probability distribution $\mu$ to be a \textit{higher-order function}, taking in an expectand $f : \tau \to \RR$, and outputting the expectation $\sum_{x \in \text{supp}(\mu)} \mu(x) \cdot f(x)$. If we know how to differentiate $\mu$ \textit{as an expectation operator}, then we will know how to differentiate expectations with respect to $\mu$. 

What would this look like in practice? For the primitive $\texttt{flip}$, we would need to implement a built-in derivative $\texttt{flip}_\der$, of type $\mathbb{I} \times \RR \to (\mathbb{B} \to \RR \times \RR) \to \RR \times \RR$. Intuitively, it takes in a dual number $(\theta, \delta \theta) : \mathbb{I} \times \RR$ representing the probability of heads, and dual-number \textit{expectand} $df : \mathbb{B} \to \RR \times \RR$, and returns a dual number with the value and derivative of $$\mathbb{E}_{x \sim \texttt{flip}(\theta)}[\pi_1(df(x))] = \theta \cdot \pi_1(df(\True)) + (1-\theta) \cdot \pi_1(df(\False)).$$ This looks reasonable, and is not hard to implement in practice, using the dual number operators $\times_\der$ and $+_\der$. 
Indeed, this choice turns out to be quite nice. Expectation operators of probability distributions form a submonad of the continuation monad~\citep{vakar2019domain}, so we would be translating one term of monadic type ($t : \pmonad \tau$) to a new term of monadic type ($\ad{t} : (\ad{\tau} \to \ad{\RR}) \to \ad{\RR}$, whose type is equivalent to $\textbf{Cont}_{\ad{\RR}}\, (\ad{\tau})$). 
\revision{Furthermore, if we translate $\haskdo \, \{x \gets t; m\}$ into $\haskdo_{\textbf{Cont}}\, \{x \gets \ad{t}; \ad{m}\}$ and $\return~t$ to $\return_{\textbf{Cont}}~\ad{t}$, we obtain correct (exact) derivatives of compound probabilistic programs' expectations. This is nice in that $\ad{\cdot}$ still preserves even a monadic program's structure, with ``all the action'' happening at the primitives.}

But there is one fatal flaw with this otherwise appealing approach: it computes \textit{exact} derivatives of expectations, by summing over all possible random paths through a program, and in practice this will generally be completely intractable. 

\subsection{ADEV for the Probabilistic Language, Correctly}
\label{sec:ad-discrete}
\input{figures/04figures/ad}
\input{figures/04figures/primitives}

We present our approach to extending the ADEV macro in Figures~\ref{fig:04ad}-\ref{fig:discrete_builtins}. Our strategy reaps all the benefits of False Start 3 from the previous section, but avoids the fatal flaw:
everywhere that False Start 3 must compute exact expectations of type $\RR$, we permit estimated expectations of type $\eRR$. For example, instead of requiring each primitive to compute intractable exact derivatives of expectations of arbitrary expectands, we allow primitives $p$ to expose procedures for \textit{estimating} the derivatives of expectations $\mathbb{E}_{x \sim p}[f(x)]$, given as input a procedure $\widetilde{df} : \sem{\ad{\type}} \to \deRR$ for \textit{estimating} the value and derivative of the expectand $f$. 
We describe the intuition behind the translation:

\begin{itemize}[leftmargin=*]
    \item \textbf{Understanding the macro at the type level:} Our macro translates terms of probabilistic program type $\pmonad\,\tau$ into terms of \textit{dual-number expectation estimator} type $$\dpmonad\,\ad{\tau} := \textbf{Cont}_{\deRR}\,\ad{\tau} = (\ad{\tau} \to \deRR) \to \deRR.$$
    Given a probabilistic program $\vdash p : \pmonad\tau$, the translation {produces} an {\it algorithm} $\vdash \ad{p} : (\ad{\tau} \to \deRR) \to \deRR$ for estimating the value and derivative of a $\sem{p}$-\textit{expectation}. 
    The generated procedure $\ad{p}$ takes as input a function $\widetilde{df} : \sem{\ad{\tau}} \to \deRR$, which, on input $(x, \delta x)$, estimates some true (dual-number) expectand $(f(x), \delta f(x, \delta x)) : \RR \times \RR$. The goal of the procedure $\sem{\ad{p}}$ is then to estimate $\mathbb{E}_{x \sim \sem{p}}[f(x)]$ and its tangent value.

    \item \textbf{Understanding the macro on $\return$:} One of the simplest probabilistic programs is $\return\, x : \pmonad\tau$, which implements the Dirac delta distribution that returns $x$ with probability 1. The expectation of a function $f$ with respect to this distribution is just $f(x)$. Indeed, our macro translates this term to $\dreturn \, dx : \dpmonad \ad{\tau}$, which, unfolding the definitions, is equivalent to $\lambda \widetilde{df}:\ad{\tau} \to \deRR. \widetilde{df}(dx) : (\ad{\tau} \to \deRR) \to \deRR$. Intuitively, if we know how to estimate the dual number $df(dx)$ for any $dx$, we can also estimate its expectation under the Dirac delta\textemdash just plug in $dx$.
    
    \item \textbf{Understanding the macro on \texttt{flip}:} For the primitive distribution $\texttt{flip}$, we must attach a built-in derivative $\texttt{flip}_\der$, capable of estimating expectations with respect to the Bernoulli distribution, as well as their derivatives. It turns out there are multiple sensible choices, which strike different trade-offs between computational cost and variance. To afford the user maximum flexibility in navigating these trade-offs, we expose two \textit{versions} of \texttt{flip}, $\flipenum$ and $\flipreinforce$, which have the same semantics, but different built-in derivatives. Our implementations of these built-in derivatives are given in Figure~\ref{fig:discrete_builtins}. The estimator $\flipenum_\der$ is the costlier but lower-variance option: to estimate an expected loss, it estimates the expectand on both possible sample values, $\True$ and $\False$, and computes a (dual-number) weighted average. By contrast, $\flipreinforce_\der$ \textit{samples} a value $b$, and only estimates the expectand for that sample value. It then uses the REINFORCE or score-function estimator to estimate the derivative of the expectation. In both cases, we emphasize how the logic of a particular derivative estimation strategy is encapsulated inside a procedure attached to the $\flipenum$ or $\flipreinforce$ primitive. This modular design supports future extensions with new gradient estimation strategies, or with new primitive distributions.
    
    \item \textbf{Understanding the macro on $\haskdo$:} To translate a term that \textit{sequences} probabilistic computations, $\haskdo \, \{x \gets t; m\}$, the macro outputs $\dhaskdo \, \{dx \gets \ad{t}; \ad{m}\}$. This is sugar for the continuation monad; desugaring, if $t : \pmonad\sigma$ and $x : \sigma \vdash \haskdo \{m\} : \tau$, we get that the produced term is equivalent to $\lambda \widetilde{df} : (\ad{\tau} \to \deRR). \, \ad{t}(\lambda dx : \ad{\sigma}. \haskdo_\der \{\ad{m}\}(\widetilde{df}))$. How should we understand this term? In order to estimate an expectation with respect to the \textit{sequence} of computations, we apply the law of iterated expectation ($\mathbb{E}_{(x, y) \sim p}[f(y)]=\mathbb{E}_{x\sim p}[\mathbb{E}_{y \sim p(\cdot \mid x)}[f(y)]]$): we estimate an \textit{expected} [expectation with respect to $\sem{\haskdo \{m\}}$] with respect to $\sem{t}$. The inner expectation is estimated using the translation of $m$, and the outer one is estimated using the translation of $t$. 
    
    \item \textbf{Understanding the macro on $\mbe$:} Once the user has constructed a term $t : \pmonad \RR$, they can construct a term $\mbe \, t : \eRR$, of estimator type. We think of $\mbe \, t$ as an estimator of $\sem{t}$'s expectation, i.e., the expectation of the identity function under the distribution $\sem{t}$. When our macro is applied to $\mbe \, t$,
    we get the term $\mbe_\der \ad{t} : \deRR$, which, as can be seen from Fig.~\ref{fig:discrete_builtins}, is equivalent to $\ad{t}(\lambda dr:\RR\times \RR. \exact_\der dr) : \deRR$. The idea is that $\ad{t}$ is a procedure for estimating expectations (and their derivatives) of \textit{any} dual-number function with respect to $\sem{t}$; we want the expectation of the identity, so we pass in $\lambda dr. \exact_\der \, dr$, (a zero-variance estimator of) the dual-number $id$ function.
\end{itemize}

\input{figures/04figures/example}

\subsection{Correctness of ADEV on the Probabilistic Language} 
The overall ADEV workflow, and the statement of the overall correctness theorem for ADEV, will not change from Section~\ref{sec:combinator}: the user still ultimately constructs a program of type $\RR \to \eRR$, and 
it is still our job to differentiate the expectation of that program. The novelty in this section is
that now, the user has more tools for \textit{constructing} the final $\RR \to \eRR$ function, most notably
the ability to construct probabilistic programs and pass them to $\mbe$.
To establish correctness, we need to define new dual-number logical relations, defining appropriate notions of correct derivative for each new type. Then, we will need to reprove the fundamental lemma, by adding cases to our inductive proof for every new term constructor we added in this section.

\noindent\textbf{Defining the new logical relations.} The new dual number logical relations for $\mathbb{B}$ and $\pmonad~\tau$ are given in Fig.~\ref{fig:04logrel}. The relation for $\mathbb{B}$ is essentially the same one we had for $\NN$, another discrete type, reflecting that the only differentiable functions from $\RR$ into the Booleans are the constant functions. 

The relation at the probabilistic program type $\pmonad\tau$ is more interesting.
Its goal is to relate a parameterized probabilistic program $f : \RR \to \sem{\pmonad \tau}$
to the algorithm $\widetilde{g} : \RR \to (\sem{\ad{\tau}} \to \deRR) \to \deRR$ for 
estimating derivatives of expectations with respect to it.
As an intermediate step for understanding the correctness relationship that must
hold between $f$ and $\widetilde{g}$, let's first consider a simpler algorithm than $\widetilde{g}$,
that simply estimates expectations of $f$, not their derivatives. Such
an algorithm $\widetilde{a}$ would have a similar type to $\widetilde{g}$, but would have no need for dual numbers:
the program $\widetilde{a} : \RR \to (\sem{\tau} \to \eRR) \to \eRR$ would take as input a parameter 
$\theta$ and an estimated expectand $\widetilde{l} : \sem{\tau} \to \eRR$, and return an estimator 
of $\mathbb{E}_{x \sim f(\theta)}[l(x)]$, where $l(x) = \mathbb{E}_{y \sim \widetilde{l}(x)}[y]$.
One way of implementing such an $\widetilde{a}$ would be to have it sample $x \sim f(\theta)$, 
then sample $y \sim \widetilde{l}(x)$, then return $y$. That is, $\widetilde{a}$ is just the
\textit{monadic bind}, in the underlying semantic space of probability measures, of $f(\theta)$ (a probability
distribution over $\sem{\tau}$) with the continuation $\widetilde{l}$ (a probability kernel from $\sem{\tau}$ to $\RR$). 
Mathematically, we can write 
$$\widetilde{a} = \lambda \theta. \lambda \widetilde{l}. \sqint \widetilde{l}(x) f(\theta, \text{d}x),$$
where we have borrowed \textit{Kock integral} notation $\sqint \nu(y) \mu(x, dy)$ for binding a kernel $\mu : X \to \pmonad Y$ to a continuation kernel $\nu : Y \to \pmonad Z$ from synthetic measure theory~\citep{kock2011commutative,scibior2017denotational}.

\input{figures/04figures/logrel}

Now, what we want from the algorithm $\widetilde{g}$, which estimates \textit{dual-number derivatives} 
of expectations, is that it be a correct dual-number derivative of \textit{the expectation
estimation algorithm $\widetilde{a}$}. This is exactly what our logical relation $\rel_{\pmonad\tau}$
says (Figure~\ref{fig:04logrel}, inlining the definition of $\widetilde{a}$ we gave above).
\\

\noindent\textbf{Correctness of primitives.} Using this definition, we can work out a precise statement of the specification that a custom built-in
derivative for a primitive $\RR \to \pmonad \tau$ (such as $\flipenum$ and $\flipreinforce$) must meet.
It arises as a special case of Definition~\ref{def:correct-dual-number}, for the type $\RR \to \pmonad\tau$:

\begin{definition}[correct dual-number expectation estimator]
\label{def:correct-exp-est}
Let $p \in \RR \to \sem{\pmonad\tau}$ be a probability kernel from $\RR$ to $\sem{\tau}$. Then $p_D : \sem{\ad{\RR}} \to  \sem{\dpmonad \,\ad{\tau}}$ is a correct dual-number expectation estimator for $p$ if for all $(\widetilde{f} : \RR \to \sem{\tau} \to \eRR, \widetilde{g} : \RR \to \sem{\ad{\tau}} \to \deRR) \in \rel_{\tau \to \eRR}$, and all differentiable functions $h : \RR \to \RR$, $p_D(h(\theta), h'(\theta))(\widetilde{g}(\theta))$ estimates the dual number $(\mathbb{E}_{x \sim p(h(\theta))}[\mathbb{E}_{y \sim \widetilde{f}(\theta)(x)}[y]], \frac{d}{d\theta} \mathbb{E}_{x \sim p(h(\theta))}[\mathbb{E}_{y \sim \widetilde{f}(\theta)(x)}[y]])$. 
\end{definition}

The idea is that a built-in derivative for a probabilistic primitive (e.g. $\flipreinforce_\der$ for the primitive $\flipreinforce$) receives two inputs: (1) a dual-number parameter, $(h(\theta), h'(\theta))$, that is already tracking its own derivative with respect to some underlying parameter $\theta$, and 
(2) the expectand-estimator $\widetilde{g}_\theta : \sem{\ad{\tau}} \to \deRR$, which is in general a \textit{closure} that may have captured the underlying parameter $\theta$. When returning an estimated expectation and derivative of the expectation, $\flipreinforce_\der$ must account for both the way that $\theta$ influences the sampling distribution of $x \sim \flipreinforce(\theta)$, and also how it influences the closure whose expectation
is being estimated. It is instructive to go through the exercise of showing \textit{why} (for example) $\flipreinforce$'s built-in derivative satisfies this specification: the argument combines the standard REINFORCE estimator with the use of dual numbers to propagate derivatives.

\begin{lemma}
\label{lem:flipreinforce-correct}
$\sem{\flipreinforce_\der}$ is a correct dual-number expectation estimator for $\sem{\flipreinforce}$.
\end{lemma}
\begin{proof}
\allowdisplaybreaks
Let $h : \RR \to \RR$ differentiable, and $(\widetilde{f} : \RR \to \mathbb{B} \to \eRR, \widetilde{g} : \RR \to \mathbb{B} \to \deRR) \in \rel_{\BB \to \eRR}$. Then by the definition of correct dual-number derivative, $\sem{\flipreinforce_\der}$ should estimate
\begin{align}
    \frac{d}{d\theta} &\mathbb{E}_{x \sim Bern(h(\theta))}[\mathbb{E}_{y \sim \widetilde{f}(\theta)(x)}[y]]\\
    & \textit{(standard REINFORCE estimator, based on log derivative trick)}\notag\\
    &= \mathbb{E}_{x \sim Bern(h(\theta))}\left[\left(\frac{d}{d\theta} \log Bern(x; h(\theta))\right)\mathbb{E}_{y \sim \widetilde{f}(\theta)(x)}[y] + \frac{d}{d\theta}\mathbb{E}_{y \sim \widetilde{f}(\theta)(x)}[y]\right]\\
    & \textit{(use the fact that $(\widetilde{f}, \widetilde{g}) \in \rel_{\tau \to \eRR}$ to rewrite both terms)}\notag\\
    &= \mathbb{E}_{x \sim Bern(h(\theta))}\left[\left(\frac{d}{d\theta} \log Bern(x; h(\theta))\right)\mathbb{E}_{(y, \delta y) \sim \widetilde{g}(\theta)(x)}[y] + \mathbb{E}_{(y, \delta y) \sim \widetilde{g}(\theta)(x)}[\delta y]\right]\\
    & \textit{(push log density term inside expectation, then combine expectations)}\notag\\
    &= \mathbb{E}_{x \sim Bern(h(\theta))}\left[\mathbb{E}_{(y, \delta y) \sim \widetilde{g}(\theta)(x)}\left[y\cdot\left(\frac{d}{d\theta} \log Bern(x; h(\theta))\right) + \delta y\right]\right]\\
    & \textit{(evaluating the derivative)}\notag\\
    &= \mathbb{E}_{x \sim Bern(h(\theta))}\left[\mathbb{E}_{(y, \delta y) \sim \widetilde{g}(\theta)(x)}\left[y\cdot\left(\frac{-1^{1-x} \cdot h'(\theta)}{Bern(x; h(\theta))}\right) + \delta y\right]\right].
\end{align}
This final expression can be estimated using just the dual number $(h(\theta), h'(\theta))$, and the function $\widetilde{g}$, by generating $x \sim Bern(x; h(\theta))$, then $(y, \delta y) \sim \widetilde{g}(\theta)(x)$, and then returning the value from line (5) above. This is exactly what $\sem{\flipreinforce_\der}$ does, to compute the tangent value it returns.
\end{proof}

Proving every primitive correct in a similar manner, following Section~\ref{sec:ad-discrete}'s logic, we can derive:

\begin{lemma}[Fundamental lemma (revised with $\mathbb{B}$ and $\pmonad\tau$)]
For every term $\Gamma \vdash t : \tau$, $\sem{\ad{t}}$ is a correct dual-number derivative of $\sem{t}$, w.r.t. the relations $\rel_\tau$ defined at each type (incl. $\mathbb{B}$ and $\pmonad\tau$).
\end{lemma}

Having reproved the fundamental lemma, the proof of Theorem~\ref{thm:err-correctness} goes through unchanged:

\begin{theorem}[correctness of ADEV for the discrete probabilistic language]
\label{thm:correctness_discrete}
For all closed terms $\vdash t : \RR \to \eRR$, $\sem{\lambda \theta : \RR. \snd_* (\ad{t} \, (\theta, 1))}$ is an unbiased derivative of $\sem{t}$.
\end{theorem}

\revision{
\begin{corollary}\label{cor:all_smooth_types}
For all continuous numeric types $\KK \in \{\RR, \posreal, \II\}$, and all closed terms $\vdash t : \mathbb{K} \to \eRR$,
 $\llbracket \lambda \theta : \mathbb{K}. \text{snd}_*\, (\ad{\ter}(\theta, 1))\rrbracket$ is an unbiased derivative of $\llbracket t\rrbracket$.
\end{corollary}

\begin{proof}
We have the result for $\KK=\RR$ from Thm~\ref{thm:correctness_discrete}, so first consider $\KK=\II$. Let $\vdash \ter:\II\to\eRR$.
By the fundamental lemma, $\sem{\ad{\ter}}$ is a correct dual-number derivative of $\sem{\ter}$, so for any $(h,h_\der)\in \rel_\II$, we have $(\sem{\ter}\circ h, \sem{\ad{\ter}}\circ h_\der)\in \rel_{\eRR}$. 
In particular, this means that for any $r\in\RR$, $\expect_{(x,\delta x)\sim \sem{\ad{\ter}}(h_\der(d))}[\delta x]=(\lambda r.\expect_{x\sim \sem{\ter}(h(r))}[x])'(r)$.
Now let $\theta\in\II$ and consider $h:= \lambda r.\frac{\theta}{\theta+(1-\theta)e^{-r/(\theta-\theta^2)}}, h_\der:= \lambda r.(h(r),h'(r))$.
Because $h_\der$ computes $h$'s derivative, $(h,h_\der)\in\rel_\II$. The important property of this function $h$ is that $h(0)=\theta$ and $h'(0)=1$.
Plugging this $h$ into the equation from above, and setting $r$ to $0$, we get that $\expect_{(x,\delta x)\sim \sem{\ad{\ter}(\theta,1)}}[\delta x]= (\lambda r.\expect_{x\sim \sem{\ter}(h(r))}[x])'(0)$.
The left-hand side is the expected value of $\sem{\lambda \theta:\II.\snd_* \ad{\ter}(\theta,1)}(\theta)$. The right-hand side can be rewritten, using the chain rule, to yield $h'(0) \cdot (\lambda z.\expect_{x\sim \sem{\ter}(z)}[x])'(h(0)) =1 \cdot (\lambda z.\expect_{x\sim \sem{\ter}(z)}[x])'(\theta)$.
The fact that the LHS and RHS are equal implies that $\sem{\lambda \theta:\II.\snd_* \ad{\ter}(\theta,1)}$ is an unbiased derivative of $\sem{\ter}$.
For the type $\posreal$, the argument is the same, expect that we define $h:=\lambda r. \theta.e^{r/\theta}$, which also has the property that $h(0)=\theta$ and $h'(0)=1$ but has codomain $\posreal$ instead of $\II$.
\end{proof}
}

%% file: figures/04figures/syntax.tex
\begin{figure}[t]
\fbox{
  \parbox{.97\textwidth}{
      \vspace{-2mm}
    \centering
    \begin{align*}
        \text{Types }\type ::=\,& \ldots
        \mid \BB
        \mid \pmonad\type
        \mid \colorbox{gray!15}{$\dpmonad\type $}
        \\
        \text{Terms } \ter ::=\,& \ldots 
        \mid \True \mid \False
        \mid \iifthenelse{\ter}{\ter_1}{\ter_2}
        \mid 
          \haskdo\, \{ \, m \, \}
          \mid  \colorbox{gray!15}{$\dhaskdo\, \{ \, m \, \} $}
          \\
          & \revision{\mid \return\, \ter \mid
        \colorbox{gray!15}{$\dreturn\, \ter $}}\\
        \text{Do notation } m ::=\,&
        \revision{\ter \mid\,}
          \var \gets \ter; \, m\\
        \text{Primitives } c ::=\,& \ldots 
        \mid  \flipreinforce, \flipenum: \II\to \pmonad\BB 
        \mid \mbe : \pmonad\RR\to \eRR
    \end{align*}
    \begin{tabular}{c}
   $\Gamma\vdash \ter:\type$ \\ \hline
   $\Gamma \vdash \return~\ter:\pmonad\type$
\end{tabular}
\quad
\begin{tabular}{c}
    $\Gamma\vdash \ter :\pmonad\type$  \\  \hline
     $\Gamma\vdash \haskdo \{\ter\}:\pmonad\type$
\end{tabular}
\quad
\begin{tabular}{c}
     $\Gamma \vdash\ter:\pmonad\type_1$ \quad 
     $\Gamma,\var:\type_1\vdash \haskdo\{m\} :\pmonad\type$  \\  \hline
     $\Gamma\vdash \haskdo \{\var\gets \ter;m\}:\pmonad\type$
\end{tabular}
\vspace{.1cm}

 \begin{tabular}{c}
   $\Gamma\vdash \ter:\type$ \\ \hline
   $\Gamma \vdash \dreturn~\ter:\dpmonad\type$
\end{tabular}
\quad
\begin{tabular}{c}
    $\Gamma\vdash \ter :\dpmonad\type$  \\  \hline
     $\Gamma\vdash \dhaskdo \{\ter\}:\dpmonad\type$
\end{tabular}
\quad
\begin{tabular}{c}
     $\Gamma \vdash\ter:\dpmonad\type_1$ \quad 
     $\Gamma,\var:\type_1\vdash \dhaskdo\{m\} :\dpmonad\type$  \\  \hline
     $\Gamma\vdash \dhaskdo \{\var\gets \ter;m\}:\dpmonad\type$
\end{tabular}
\vspace{.1cm}

\revision{
\begin{tabular}{c}
     $\Gamma \vdash \ter:\BB$ 
     \quad 
     $\Gamma \vdash \ter_1:\type$  
     \quad 
     $\Gamma \vdash \ter_2:\type$ \\ 
     \hline
     $\Gamma\vdash \iifthenelse{\ter}{\ter_1}{\ter_2}:\type$ 
\end{tabular}}
    \vspace{-1mm}
    \[\llet~\var =\ter;m \text{ is sugar for }\var\gets \return~\ter;m\text{ and }
    \ter;m\text{ for }\_ \gets \ter;m\]
    \vspace{-4mm}
    }}
        \vspace{-2mm}
    \caption{Syntax of the discrete probabilistic language, as an extension to Figs.~\ref{fig:02syntax} and~\ref{fig:03syntax}. \revision{ Gray highlights indicate syntax only present in the target language of the AD macro.}
    }
    \label{fig:04syntax}
\end{figure}

%% file: figures/discrete_semantics.tex
\begin{figure}[tb]
\fbox{
  \parbox{.97\textwidth}{
  \centering
      
  Semantics of types:
  
  \vspace{1mm}
  \begin{tabular}{ll}
      $\sem{\pmonad \type}$ &= $\{\mu : \sem{\tau} \to [0, \infty) \mid \mu$ a probability distribution on $\sem{\tau}$ with finite support $\}$ \\
      $\sem{\dpmonad \type}$ &=  $(\sem{\ad{\type}}\to \eRR_\der)\to \eRR_\der$ \quad\quad
      $\sem{\mathbb{B}}$ = $\{\textbf{True}, \textbf{False}\}$
  \end{tabular}
  \vspace{.1cm}

Semantics of terms:

\vspace{1mm}
\begin{tabular}{lll}
      $\sem{\return}(x)(y) = [x==y]$ 
      & $\sem{\return_\der}(dr) = \lambda dl. dl(dr)$ 
      & $\sem{\mbe}(\mu) = \mu$  \\
     $\sem{\flipenum}(\theta)(b) = b \, ?\,  \theta \, :\,  (1 - \theta)$ 
     & $\sem{\mbe_\der}(f) = f(\sem{\exact_\der})$ 
     & $\sem{\flipreinforce} = \sem{\flipenum}$
\end{tabular}
\vspace{-1mm}
\begin{tabular}{l}
    $\sem{\iif~t~\then~t_1\eelse~t_2}(\rho) = \sem{t}(\rho) \, ? \, \sem{t_1}(\rho)\, : \sem{t_2}(\rho)$ \\
    $ \sem{\haskdo~\{\var\gets\ter;m\}}(\rho)(y) = \sum_{z\in \text{supp}(\sem{\ter}(\rho))}\left(\sem{\ter}(\rho)(z)\times \sem{m}(\rho,z)(y)\right)$ \\
    $ \sem{\haskdo_\der~\{\var\gets\ter;m\}}(\rho)(dl) = \sem{\ter}(\rho)(\lambda v. \sem{m}(\rho[x := v])(dl))$
\end{tabular}
\vspace{1mm}
  }}
      \vspace{-2mm}
\caption{Semantics of the discrete probabilistic language}
\label{fig:disc_semantics}
\end{figure}

%% file: figures/04figures/ad.tex
\begin{figure}[tb]
\fbox{
  \parbox{.97\textwidth}{
  \small{
  \centering
\begin{tabular}{c|c}
    \hspace{-4mm}
     \begin{tabular}{ll}
\ad{\BB} &= $\BB$ \\
        \ad{\pmonad\type} &= $\dpmonad\ad{\type}$\\
        &= $(\ad{\tau} \to \deRR) \to \deRR$
 \end{tabular} 
     & \hspace{-2mm}
     \begin{tabular}{ll}
    
    \ad{\iif~\ter~\then~\ter_1~\eelse~\ter_2} &= $\iif~\ad{\ter}~\then~\ad{\ter_1}~\eelse~\ad{\ter_2}$ \\
    
    \ad{\return~\ter} &= $\dreturn~\ad{\ter}$ \\
    
    \ad{\haskdo\{m\}} &= $\dhaskdo~\{\ad{m}\}$ \\
       
    \ad{\var\gets \ter;m} &= $\var\gets \ad{\ter};\ad{m}$ 
     
    \end{tabular}
\end{tabular}

+ new primitives for the built-in derivatives of $\flipenum$, $\flipreinforce$, and $\mbe$.
}}}
    \vspace{-2mm}
\caption{Extended ADEV macro for the Discrete Probabilistic Programming Language (\S\ref{sec:discrete})
}
\label{fig:04ad}
\end{figure}

%% file: figures/04figures/primitives.tex


\begin{figure}[t]
\fbox{
\parbox{.97\textwidth}{
  \footnotesize{
   \begin{tabular}{ccc}
   \hspace{-.7cm}
\begin{minipage}{.45\linewidth}
     \begin{algorithm}[H]
\DontPrintSemicolon
\SetKwProg{Fn}{}{:}{end}
\Fn{$\flipenum_\der(dp : \II\times \RR, \widetilde{dl}: \BB\to\deRR)$}{
$dl_1\sim \widetilde{dl}~\True$\;
$dl_2\sim \widetilde{dl}~\False$\;
$dr1 \gets (dp ~\times_\der ~dl_1)$ \;
$dr2 \gets ((1,0)-_\der dp)~ \times_\der~ dl_2)$\;
\Return $dr1 +_\der dr2$
}
\end{algorithm}  
\end{minipage}
&
\hspace{-2.1cm}
\begin{minipage}{.3\linewidth}
\begin{algorithm}[H]
\DontPrintSemicolon
\SetKwProg{Fn}{}{:}{end}
\Fn{$\mbe_\der(\widetilde{dl} : \dpmonad (\RR \times \RR))$}{
$\widetilde{dr} \sim \widetilde{dl} (\sem{\exact_\der})$\;
\Return $\widetilde{dr}$
}
\end{algorithm}
\end{minipage}
     & 
     \hspace{-1.8cm}
\begin{minipage}{.55\linewidth}
\begin{algorithm}[H]
\DontPrintSemicolon
\SetKwProg{Fn}{}{:}{end}
\Fn{$\flipreinforce_\der(dp : \II\times \RR, \widetilde{dl}: \BB\to\deRR)$}{
$b\sim \text{Bernouilli}(\fst~dp)$\;
$(l_1,l_2)\sim \widetilde{dl}~b$ \;
$dlp \gets$ \textbf{if} $b$ \textbf{then} $\logs_\der ~dp$ \textbf{else} $\logs_\der ~((1,0)-_\der dp)$\;
$\delta logpdf \gets \snd~dlp$\;
\Return $(l_1,l_2+l_1\times \delta logpdf)$
}
\end{algorithm}
\end{minipage}
\end{tabular}
\vspace{-2mm}
}}}
    \vspace{-2mm}
\caption{Built-in derivatives for our new probabilistic primitives.}
\label{fig:discrete_builtins}
\end{figure}

%% file: figures/04figures/example.tex
\begin{figure}
\tikzstyle{mybox} = [rectangle, 
minimum width=2cm, minimum height=1cm,
text centered, 
draw=black, dashed
]

\tikzstyle{arrow} = [thick,->,>=stealth]

\resizebox{.9\textwidth}{!}{
\begin{tikzpicture}[node distance=2cm]
\node (start1) [mybox] {
\small{
\begin{tabular}{l} 
    $\loss = \lambda \theta:\RR.\,  \mbe (\haskdo~\{$\\
    $\quad b \gets \flipreinforce\,\theta$\\
    $\quad \iif~b~\then$\\
    $\quad\quad \return~0$\\
    $\quad\eelse$\\
    $\quad\quad \return~(\theta \div {-2})\})$
\end{tabular}}
};
\node (start2) [mybox, right of=start1, xshift=3cm] {\small{
\begin{tabular}{l}
    $d\loss = \lambda\, d\theta:\RR\times\RR.\,  \mbe_\der (\dhaskdo~\{$\\
    $\quad b \gets \flipreinforce_\der\,d\theta$\\
    $\quad \iif~b~\then$\\
    $\quad\quad \return_\der~(0,0)$\\
    $\quad\eelse$\\
    $\quad\quad \return_\der~(d\theta \div_\der ({-2}, 0))\})$
\end{tabular}}
};
\draw [arrow] (start1) -- node[anchor=south] {
\small{
\begin{tabular}{c}
   $\ad{\cdot}$
\end{tabular}}
} (start2);

\node (start3) [mybox, right of=start2, xshift=3.7cm] {\small{
\begin{tabular}{l} 
     $d\loss = \lambda\, d\theta:\RR\times\RR.\,  \mbe_\der (\lambda \widetilde{dl}.$\\   $\quad\flipreinforce_\der\,d\theta\,(\lambda b.$\\
    $\quad\quad\iif~b~\then$\\
    $\quad\quad\quad \widetilde{dl}(0,0)$\\
    $\quad\quad\eelse$\\
    $\quad\quad\quad \widetilde{dl}(d\theta \div_\der ({-2}, 0))))$
\end{tabular}}
};
\draw [arrow] (start2) -- node[anchor=south] {\small{
\begin{tabular}{c}
    desugar
    \\ 
    $\dhaskdo$
\end{tabular}}
} (start3);

\node (start4) [mybox, below of=start3, xshift=-0.5cm, yshift=-1.5cm] {\small{
\begin{tabular}{l} 
    $d\loss = \lambda\, d\theta:\RR\times \RR.\,$\\
    $\quad\flipreinforce_\der\,d\theta\,(\lambda b.$\\
    $\quad\quad\iif~b~\then$\\
    $\quad\quad\quad \exact_\der(0,0)$\\
    $\quad\quad\eelse$\\
    $\quad\quad\quad \exact_\der(d\theta \div_\der ({-2}, 0)))$
\end{tabular}}
};

\draw [arrow] (start3) -- node[anchor=east] {\small{apply $\mbe_\der$}} (start4.70);

\node (start5) [mybox, left of=start4, xshift=-7.5cm,yshift=-1.7cm] {\small{
\begin{tabular}{l} 
    $d\loss = \lambda\, d\theta:\RR\times \RR.\,$\\
    $\quad(\lambda (\theta, \delta \theta). \lambda \widetilde{dl}. $\\
    $\quad\quad\mbe(\haskdo \{$ \\
    $\quad\quad\quad b \gets \flipreinforce \theta$\\
    $\quad\quad\quad (l, \delta l) \gets \widetilde{dl}$\\
    $\quad\quad\quad \llet~\delta logpdf = \iif~b~\then~$\\
    $\quad\quad\quad\quad\quad \delta \theta \div \theta~$\\
    $\quad\quad\quad\quad\eelse$\\
    $\quad\quad\quad\quad\quad \delta\theta \div (\theta - 1)$\\
    $\quad\quad\quad \return (l, \delta l + l \times \delta logpdf)$\\
    $\quad\quad\}))(d\theta)(\lambda b.$\\
    $\quad\quad\quad\iif~b~\then$\\
    $\quad\quad\quad\quad \exact_\der(0,0)$\\
    $\quad\quad\quad\eelse$\\
    $\quad\quad\quad\quad \exact_\der(d\theta \div_\der ({-2}, 0)))$
\end{tabular}}
};

\draw [arrow] (start4.180) -- node[anchor=south] {\small{
\begin{tabular}{c}
     inline  \\
     $\flipreinforce_\der$
\end{tabular}}
} (start5.34);


\node (start7) [mybox, below of=start4, yshift=-1.5cm, xshift=-0.2cm] {\small{
\begin{tabular}{l} 
     $d\loss = \lambda\, (\theta, \delta \theta):\RR\times\RR.\,\mbe(\haskdo\{$\\
    $\quad b \gets \flipreinforce \theta$\\
    $\quad \iif~b~\then$\\
    $\quad\quad\return~(0,0)$\\
    $\quad\eelse$\\
    $\quad\quad\llet~(l, \delta l) =(\theta, \delta\theta) \div_\der ({-2}, 0)$\\
    $\quad\quad \llet~\delta logpdf = \delta\theta \div (\theta - 1)$\\
    $\quad\quad \return (l, \delta l + l \times \delta logpdf)\})$
\end{tabular}}
};

\draw [arrow] (start5.-35) -- node[anchor=south] {\small{
$\beta$-reduce}
} (start7.180);


\node (start8) [mybox, below of=start5, yshift=-3cm] {\small{
\begin{tabular}{l} 
      $\loss' = \lambda\, \theta:\RR.\,\mbe(\haskdo\{$\\
    $\quad b \gets \flipreinforce \theta$\\
    $\quad \iif~b~\then$\\
    $\quad\quad\return~0$\\
    $\quad\eelse$\\
    $\quad\quad\llet~(l, \delta l) = (\theta, 1) \div_\der ({-2}, 0)$\\
    $\quad\quad \llet~\delta logpdf = 1 \div (\theta - 1)$\\
    $\quad\quad \return (\delta l + l \times \delta logpdf)\})$
\end{tabular}}
};

\draw [arrow] (start7.211) -- node[anchor=north] {\small{
\begin{tabular}{c}
   apply to $(\theta, 1)$, \\
   extract dual component
\end{tabular}}
} (start8.33);

\node (start9) [mybox, right of=start8, xshift=7.6cm,yshift=-1cm] {\small{
\begin{tabular}{l} 
      $\loss' = \lambda\, \theta:\RR.\,\mbe(\haskdo\{$\\
    $\quad b \gets \flipreinforce \theta$\\
    $\quad \iif~b~\then$\\
    $\quad\quad\return~0$\\
    $\quad\eelse$\\
    $\quad\quad \llet~l = -\theta \div 2$\\
    $\quad\quad \llet~\delta l = -1 \div 2$\\
    $\quad\quad \llet~\delta logpdf = 1 \div (\theta - 1)$\\
    $\quad\quad \return (\delta l + l \times \delta logpdf)\})$
\end{tabular}}
};

\draw [arrow] (start8.-22) -- node[anchor=north] {\small{
\begin{tabular}{c}
    perform \\
    dual-number arithmetic
\end{tabular}}
} (start9);

\end{tikzpicture}}
    \vspace{-3mm}
\caption{How ADEV, applied to the example program from Fig.~\ref{fig:flip-example}, derives the term on the bottom right.
The ADEV macro $\ad{\cdot}$ is itself very simple, changing only
constants and primitives, just as in forward-mode AD.
After applying it, we partially evaluate the resulting term
for clarity, but these are \textit{not} new transformations.
(NB: We overload $\mbe : \pmonad\RR \to \eRR$ to also work on inputs of $\pmonad (\RR \times \RR)$ type, yielding output of type $\deRR$.)}

\label{fig:04reduction-example}
\end{figure}

%% file: figures/04figures/logrel.tex
\begin{figure}[tb]
\fbox{
  \parbox{.97\textwidth}{
    \vspace{-2mm}
      \begin{align*}
      \rel_\BB &= \{(f : \RR \to \BB, g : \RR \to \BB) \mid f \text{ is constant } \wedge f = g\} \\
        \rel_{\pmonad ~\type} &= \{(f:\RR\to \pmonad~\type,\widetilde{g}:\RR\to (\ad{\type}\to \deRR)\to \deRR) \mid \\
        &\quad (\lambda \theta.\lambda \widetilde{l}:\tau\to\eRR.\sqint \widetilde{l}(x)f(\theta)(dx), \widetilde{g})\in \rel_{(\tau\to\eRR)\to\eRR}\}
        \end{align*}
        \vspace{-.3cm}
  }}
      \vspace{-2mm}
\caption{Definition of the dual-number logical relation  for our Discrete Probabilistic language}
\label{fig:04logrel}
\end{figure}

%% file: 05continuous.tex
\section{Extending ADEV to Continuous Probabilistic Programs}
\label{sec:continuous}


We now lift the key restriction from Section~\ref{sec:discrete}: we add to our language new primitives for sampling from continuous distributions (Fig.~\ref{fig:05syntax}).
Perhaps surprisingly, nothing about the ADEV macro or the user's workflow
changes with this extension. Adding continuous
primitives is no different from adding discrete primitives: just as
in Section~\ref{sec:discrete}, the key task is to design 
built-in derivative-of-expectation estimators of type $\dpmonad \tau$
for every $\pmonad \tau$ primitive we add. What \textit{does} change
is the correctness proof: as we will see in Section~\ref{sec:correctness-continuous}, the introduction of continuous probability adds several 
wrinkles to our semantics and logical relations.

\subsection{Syntax and Algorithm}
Our extended language (Fig.~\ref{fig:05syntax}) has four new primitives: $\sample:\pmonad \mathbb{I}$ (which samples on the unit interval), $\normalreparam,\,\normalreinforce:\RR\times \posreal\to \pmonad \RR$ (which sample a normal distribution with user-specified parameters), and $\geometricreinforce:\II\to \pmonad\NN$ (which samples a geometric).\footnote{The geometric distribution is not continuous, but violates a different restriction from Section~\ref{sec:discrete}\textemdash finite support.} 

Following our development in Section~\ref{sec:discrete}, we equip each primitive with a built-in derivative estimation procedure, based on the REINFORCE and reparameterization-trick gradient estimators, well-studied in the machine learning literature~\citep{kingma2013auto}. (See Fig.~\ref{fig:probabilistic_semantics} in Appendix~\ref{appx:figures}.) As in Section~\ref{sec:discrete}, the novelty here is not in the estimators themselves but in the modularity, with gradient estimators exposed to the user in the form of composable primitives. Beyond the translation of these new primitives, the ADEV macro requires no further extensions. 

\input{figures/05figures/syntax}

\subsection{Correctness of ADEV in the Continuous Language}
\label{sec:correctness-continuous}
\textbf{New challenges.} All we have done is add a few new primitives, but formally justifying the extension raises two significant technical difficulties:

\begin{itemize}[leftmargin=*]
\item \textbf{Measurability issues.} In Section~\ref{sec:discrete} we chose $\sem{\pmonad\tau}$ to be the monad of finitely supported mass functions. This choice was nice, because (1) the set of finitely supported mass functions on $\sem{\tau}$ is well-defined for \textit{any} set $\sem{\tau}$, and (2) the expectation of \textit{any} function $f : \sem{\tau} \to \RR$ with respect to a finitely-supported distribution is well-defined (it is just a finite sum). 
We exploited property (2) in defining our logical relations $\rel_{\pmonad\tau}$, which talk about expectations of arbitrary functions. In our newly extended language, we must revise our choice of $\sem{\pmonad\tau}$, setting it to (something like) the set of probability measures on $\sem{\tau}$. But this breaks both of the nice properties above: (1) there is no nice way to define the set of probability measures over $\sem{\tau}$ when $\tau$ is higher-order (e.g. $\tau = \RR\to\RR$)~\citep{heunen2017convenient}, and (2) in general expectations can only be taken of measurable functions.
The challenge, then, is to find a way of defining semantics for the extended language, and updating our logical relations, that doesn't break anything we've done so far.

\item \textbf{Edge cases where primitive gradient estimators are incorrect.} The standard proofs that the REINFORCE and reparameterization trick estimators are correct come with regularity conditions on the function $f$ whose expectation's derivative is being estimated. As such, our new primitives and their ADEV translations $(c, c_\der)$ do not 
technically satisfy the correctness criterion implied by our
logical relations (Definition~\ref{def:correct-exp-est}), which
quantifies over all possible expectands. An updated
correctness theorem will need to somehow account for
these regularity conditions.
\end{itemize}

\noindent\textbf{Resolving the first challenge: quasi-Borel semantics.} 
A long line of research has recently culminated in a new setting for measure theory where function spaces are well-behaved: the quasi-Borel spaces~\citep{heunen2017convenient}. Like a measurable space, a \textit{quasi-Borel space} $X$ pairs an underlying set $|X|$ with additional structure for reasoning precisely about probability; see~\citet{scibior2017denotational} for an overview. Here we just summarize our application of the theory:
\begin{itemize}[leftmargin=*]
    \item For every quasi-Borel space $X$, there is a quasi-Borel space $\pmonad X$ of probability measures on $X$, and these form a strong commutative monad. Using it, we were able to reformulate our language's semantics in terms of quasi-Borel spaces: every type $\tau$ is interpreted by a space $\sem{\tau}$, and terms $t$ are interpreted as \textit{quasi-Borel morphisms} $\sem{t}$, which are just functions satisfying a generalized measurability property ensuring they work nicely with quasi-Borel measures. Our interpretations are standard, matching those of~\citet{scibior2017denotational}. We note that, unlike in Section~\ref{sec:discrete}, we now interpret  $\eRR$ and $\pmonad\RR$ the same way: as the quasi-Borel probability measures on $\RR$.
    
    \item Now that we have changed our semantics, how should we think about the definitions, appearing throughout our paper, of logical relations $\rel_\tau$? We can read the definitions exactly as they are written, but interpreting them as relations over sets of quasi-Borel morphisms $\RR \to \sem{\tau}$ (or $\RR \to \sem{\ad{\tau}}$), rather than over sets of arbitrary functions. Any expectations and integrals appearing in our definitions should now be understood as expectations and integrals of quasi-Borel morphisms with respect to quasi-Borel measures.\footnote{Existing work on quasi-Borel spaces usually defines integration of $X \to [0, \infty]$ functions, and not of $X \to \RR$ functions. But just as in standard measure theory, we can extend the definition for non-negative functions to one for arbitrary real functions: we split the integrand into a positive part and negative part, separately integrate each, and then subtract the results. Because each result can be either finite or infinite, their difference can be either finite, infinite, \textit{or undefined} (if both the positive part and negative part are infinite). In this paper, when we say that an expectation exists or is well-defined, we mean that the result of the integral is \textit{finite}.}
    
    \item Having re-interpreted our language and our definitions of logical relations, we should do a sanity check that our proofs from Section~\ref{sec:discrete} still go through (using the new semantics, but not yet adding our new primitives for continuous sampling). It turns out they do: 
    
    \begin{lemma}[fundamental lemma for the old language, new semantics]
    For every term $\Gamma \vdash t : \tau$ in the discrete  probabilistic language of Section~\ref{sec:discrete}, $\sem{\ad{t}}_\textbf{Qbs}$ is a correct dual-number derivative of $\sem{t}_\textbf{Qbs}$, with respect to the $\rel_\tau$ obtained by interpreting our previous definitions as relations of quasi-Borel morphisms.\footnote{For the categorically-minded reader, we provide another presentation of this logical relations argument in Appendix~\ref{appx:qbs_logrel}.}
    \end{lemma}

\end{itemize}

\noindent \textbf{Resolving the second challenge: surfacing regularity conditions with a lightweight static analysis.} Now that we have a clear semantics, we can move onto the second problem: our logical relations $\rel_\tau$ are too strict. In particular, they require that a primitive distribution must be able to estimate the derivative of the expectation of \textit{any} (smooth, quasi-Borel) expectand, when in practice, nearly every gradient estimator needs additional regularity conditions to ensure unbiasedness.
We address this issue in three stages:

\input{figures/05figures/logrel}

\begin{enumerate}[leftmargin=*]
    \item First, we develop a \textit{weaker definition} of \textit{unbiased derivative} (Definition~\ref{def:unbiased-deriv}):
\begin{definition}[weak unbiased derivative]
\label{def:weak-unbiased} 
Let $\widetilde{f} : \RR \to \eRR$, and suppose that the map $\mathcal{L} : \RR \to \RR$ sending $\theta$ to $\mathbb{E}_{x \sim \widetilde{f}(\theta)}[x]$ is well-defined. Then $\widetilde{g} : \RR \to \eRR$ is a \textit{weak unbiased derivative} of $\widetilde{f}$ if there exists a measurable function $h : \RR \times \RR \to \RR$, continuously differentiable in its first argument, such that (1) $\mathcal{L}(\theta) = \int_{\RR} h(\theta, s) ds$, and (2) $g(\theta)$ estimates $\int_{\RR} \frac{\partial}{\partial \theta}h(\theta, s) ds$.
\end{definition}
This definition captures ``unbiased, up to interchange of an integral with a derivative'': instead of requiring that $g$ unbiasedly estimate $\mathcal{L}'(\theta)$, we require that there is some way to write $\mathcal{L}$ as an integral such that, if you could swap the derivative and the integral, $g$ would estimate $\mathcal{L}'(\theta)$.

\item Second, we develop a lightweight static analysis that, given a term $\vdash \widetilde{t} : \RR \to \eRR$, \textit{finds} the measurable function $h : \RR \times \RR \to \RR$ (from Definition~\ref{def:weak-unbiased}) that justifies ADEV's output $\widetilde{s} : \RR \to \eRR$ as a \textit{weak} unbiased derivative estimator. To do so, we create a modified version of $\ad{\cdot}$, where $\ad{\eRR} = \deRR \times (S \to \RR \times \RR)$. Here, $S$ is a new type of \textit{random seeds}: $\sem{S} = \RR$, but $\ad{S} = S$ (no dual numbers), and derivatives of $S$ computations are not tracked. Intuitively, this new $\ad{\cdot}$ translates a term of type $\eRR$ to a \textit{pair}, where the first component is the same dual number estimator (of type $\deRR$) that we produced in Sections~\ref{sec:combinator} and~\ref{sec:discrete}, and the second
is the \textit{justification of the estimator as a weak unbiased derivative}.
A crucial part of this translation is what should happen at the primitives: because $\ad{\cdot}$'s behavior on $\eRR$ has changed, to include an extra component, all our primitives involving $\eRR$ need to be updated, to produce or handle this extra component (see Figure~\ref{fig:verif_constants} in Appendix). A more formal understanding can be gained by examining the new logical relation we define for $\rel_{\eRR}$, given in Figure~\ref{fig:05logrel}. (What it calls $h_1$ is the witness $h$ from Definition~\ref{def:weak-unbiased}, and what it calls $h_2$ is its derivative.) 
By proving the fundamental lemma using this new relation, we obtain a weak correctness result for the language:
\begin{lemma}[Weak correctness of ADEV]
\label{lem:weak-correctness}
Let $\vdash \widetilde{\ter}:\RR\to\eRR$. Letting $$h_i = \lambda (\theta,x).\pi_i(\pi_2(\sem{\ad{\widetilde{\ter}}}(\theta,1))(x)): \RR\times\RR\to \RR,$$
assume that $\forall \theta \in \RR$, $\int_\RR h_1(\theta,x)dx$ and $\int_\RR h_2(\theta,x)dx$ are well-defined.
Then:
\begin{itemize}
    \item  For all $(\theta,x)\in\RR\times \RR,h_2(\theta,x)=\frac{\partial}{\partial \theta}h_1(\theta,x)$ 
    \item For all $\theta\in\RR$, $\sem{\widetilde{\ter}}(\theta)$ is an unbiased estimator of $\int_\RR h_1(\theta,x)dx$.
    \item For all $\theta\in\RR$, $\snd_*(\pi_1(\sem{\ad{\widetilde{\ter}}}(\theta,1)))$ is an unbiased estimator of  $\int_\RR h_2(\theta,x)dx$.
\end{itemize}
Therefore, $\lambda \theta:\RR. \snd_*(\pi_1(\sem{\ad{\widetilde{\ter}}}(\theta,1)))$ is a weak unbiased derivative of $\sem{\widetilde{t}}$. 
\end{lemma}

\item Finally, we state a sufficient condition for a weak unbiased 
derivative to be fully unbiased:

\begin{definition}[Locally Dominated]
We say that a function $f:\RR\times \RR\to \RR$ is locally dominated if, for every $\theta\in\RR$, there is a neighborhood $U(\theta)\subseteq \RR$ of $\theta$ and an integrable function $m_{U(\theta)}:\RR\to [0,+\infty)$ such that $\forall \theta'\in U(\theta),\forall x\in\RR, |f(\theta',x)|\leq m_{U(\theta)}(x)$.
\end{definition}

Combining it with our static analysis that finds $h$ and $h'$, we get our final correctness theorem:

\begin{theorem}[Correctness of ADEV (continuous language)]
\label{thm:cont-correct}
Let $\vdash \widetilde{\ter}:\RR\to\eRR$ be a closed term, and suppose that $\sem{\widetilde{\ter}}(\theta)$ has a well-defined expectation for every $\theta \in \RR$.
If $\sem{\lambda (\theta,x):\RR\times\RR.\snd(\snd(\ad{\widetilde{\ter}}(\theta,1))(x))}$ is locally dominated,
then $\sem{\lambda \theta:\RR.\snd_*(\fst(\ad{\widetilde{\ter}}(\theta,1)))}$ is a correct unbiased derivative of $\sem{\widetilde{\ter}}$.
\end{theorem}

Note that the final conclusion is the same \textit{full} correctness
property we proved in earlier sections; there is now just a single
local domination condition for the user to verify, before the guarantee kicks in. This local domination condition is only \textit{one} of the preconditions for swapping derivatives and integrals to be valid; crucially, the other hypotheses of the Dominated Convergence Theorem are automatically discharged by our proofs. Furthermore, even if the user's program composes different primitives, each using different gradient estimation strategies and making different assumptions, the static analysis performed by our modified $\ad{\cdot}$ macro automatically generates a \textit{single} term $\lambda (\theta,x).\snd(\snd(\ad{\widetilde{\ter}}(\theta,1))(x))$ (the $h_2$ from Lemma~\ref{lem:weak-correctness}) whose local domination should be checked. Because this is an explicit term in our language,
we are optimistic that future, more sophisticated static analyses
could be developed to automatically discharge this local domination
condition in many cases.
\end{enumerate}

%% file: figures/05figures/syntax.tex
\begin{figure}[t]
\fbox{
  \parbox{.97\textwidth}{
    \centering
      \vspace{-3mm}
    \begin{align*}
    \text{Primitives } c ::=\,& \ldots \mid  \sample :\pmonad\II \mid \normalreparam : \RR\times\posreal\to \pmonad\RR \\
    &\mid \normalreinforce : \RR\times\posreal\to \pmonad\RR 
    \mid \geometricreinforce : \II \to \pmonad\NN
    \end{align*} 
    \vspace{-5mm}
    }}
        \vspace{-3mm}
    \caption{Extended syntax for Continuous Probabilistic Programming
    }
    \label{fig:05syntax}
\vspace{-3mm}
\end{figure}

%% file: figures/05figures/logrel.tex
 \begin{figure}[tb]
\fbox{
  \parbox{.97\textwidth}{
  \begin{center}
    \vspace{-3mm}
        \end{center}
        \begin{align*}
              \rel_{\eRR} &= \Big\{(f:\RR\to\eRR, g:\RR\to(\deRR\times (S\to\RR\times \RR)))\mid h_i := \lambda \theta. \lambda s. \pi_i((\pi_2\circ g)(\theta)(s))\\
      &\qquad\wedge \forall \theta.\int_{\RR} h_1(\theta)(s)ds = \mathbb{E}_{x\sim f(\theta)}[x] = \mathbb{E}_{x\sim{\pi_1}_* (\pi_1 \circ g)(\theta)}[x]\\
      &\qquad\wedge \forall \theta.\int_{\RR} h_2(\theta)(s) ds = \mathbb{E}_{x\sim{\pi_2}_* (\pi_1 \circ g)(\theta)}[x] \\
      &\qquad\wedge (\lambda \theta. \lambda s. h_1(\theta)(s), \lambda \theta. \lambda s. (h_1, h_2)(\theta)(s)) \in \rel_{S \to \RR} \Big\} \\
        \rel_S &= \big\{(f : \RR \to S, g : \RR \to S) \mid f \text{ is constant } \wedge f = g \big\}
        \end{align*}
          }}
              \vspace{-3mm}
\caption{Revised logical relation for our type $\eRR$.}
\label{fig:05logrel}
\end{figure}

%% file: 06general.tex
\section{Stronger guarantees with smoothness-tracking types}
\label{sec:general}

\input{figures/06figures/syntax}

\input{figures/06figures/logrel}

The correctness guarantee of Theorem~\ref{thm:cont-correct}
covers an expressive language with discrete and continuous
sampling, higher-order functions, 
monadic probabilistic programming, and conditional
branching. But ADEV sometimes produces correct derivatives
in cases our theory does \textit{not} yet cover, namely when the user's
program uses discontinuous primitives like $\leq$. Consider, for example,
\begin{equation}\lambda \theta : \RR. \haskdo \, \{x \gets \normalreinforce(\theta, 1); \iif~x \leq 3~\then~\return~1~\eelse~\return~0\},
\label{eqn:good-program}\end{equation}
which \textit{uses} $\leq$ but has expectation $\mathbb{P}_{x \sim \mathcal{N}(\theta, 1)}[x \leq 3]$, which is itself differentiable with respect to $\theta$. If we equip $\leq$ with a built-in derivative that ignores the tangent part of any dual-number inputs, we can apply ADEV to this program, and in this case, we \textit{do} get out a correct derivative. Why is this, and can we state a more general theorem about when ADEV is correct?


It turns out that primitives like $\leq$ can be safely added to our language, but only if their use is carefully restricted. This is because the proof that establishes Theorem~\ref{thm:cont-correct} from Lemma~\ref{lem:weak-correctness} relies on the measure-theoretic formulation of the Leibniz integral rule, which requires us to ensure that $h_1(\theta, x)$ is differentiable with respect to $\theta$ for almost all $x$. Importantly, we do \textit{not} need $h_1$ to be differentiable with respect to $x$. Intuitively, we can allow $\leq$ in cases where it introduces discontinuities with respect to the \textit{random seed} $x$, but not to the input parameter $\theta$.
\input{figures/06figures/example}

\noindent\textbf{Types for smoothness tracking.} We can make these intuitions precise by carefully extending our language of study to allow restricted uses of discontinuous primitives, then re-proving our correctness theorem for the extended language. We do this by adding, for each \textit{smooth type} $\mathbb{K}$ ($\RR$, $\II$, and $\RR_{>0}$), a \textit{non-smooth type} $\mathbb{K}^*$. The semantics of a smooth type and its corresponding non-smooth type are the same, but our macro $\ad{\cdot}$ does not attach dual numbers to non-smooth values ($\ad{\KK^*} = \KK^*$), and our logical relations $\rel_{\KK^*}$ (Fig.~\ref{fig:06logrel}) treat them as if they were discrete.

From a user's perspective, values of non-smooth type are \textit{allowed} to be used non-smoothly, whereas values of smooth type \textit{must} be used smoothly. This is reflected in the type of the primitive $\leq : \KK^* \times \KK^* \to \mathbb{B}$. In addition to $\leq$, we introduce a coercion $\forget{\cdot} : \mathbb{K}^* \to \mathbb{K}$ from non-smooth to smooth types, but not in the other direction: you are always allowed to promise (unnecessarily) to use a value smoothly, but not to go back on your promise. In fact, as can be seen from the definition of $\rel_{\RR^*}$, \textit{any} $\RR \to \RR^*$ function expressible in our language must necessarily be constant\textemdash no information can `leak' from the smooth world to the non-smooth world.

Smooth and non-smooth types can be mixed to create functions whose types perform fine-grained tracking of \textit{which} arguments they are differentiable with respect to. For example, a term $\vdash t : \RR \times \RR^* \to \RR$ is guaranteed to have a denotation differentiable with respect to its first argument, but not its second. The key feature unlocked by this fine-grained tracking is that we can now assign more permissive types to some of our primitives from Section~\ref{sec:continuous}: namely, $\sample$ and $\normalreinforce$ now generate samples of non-smooth type, indicating that their correctness proofs do \textit{not} require sampled values to be used smoothly in the rest of the program.  $\normalreparam$, by contrast, still generates samples of type $\RR$.
With these new types, we can accept program~(\ref{eqn:good-program}) from above, while rejecting programs on which ADEV would fail (Figure~\ref{fig:06example}).

With these typing rules, we can import Section~\ref{sec:continuous}'s results with no major hurdles:

\begin{theorem}[Correctness of ADEV (full)]
\label{thm:full-correctness}
Let $\vdash \widetilde{\ter}:\RR\to\eRR$ be a closed term in the extended language of Section~\ref{sec:general}, 
such that $\sem{\widetilde{\ter}}(\theta)$ has a well-defined expectation for every $\theta$.
If $\sem{\lambda (\theta,x):\RR\times \RR.\snd(\snd(\ad{\widetilde{\ter}}(\theta,1))(x))}$ is locally dominated,
then $\sem{\lambda \theta:\RR.\snd_*(\fst(\ad{\widetilde{\ter}}(\theta,1)))}$ is a correct unbiased derivative of $\sem{\widetilde{\ter}}$.
\end{theorem}

%% file: figures/06figures/syntax.tex
\begin{figure}[t]
\fbox{
  \parbox{.97\textwidth}{
    \centering
    \vspace{-.2cm}
    \begin{align*}
    \text{Types } \type ::=\,& \ldots            \mid  \mathbb{K}^* \quad\text{(for every smooth base type $\mathbb{K}$)}\\
    \text{Primitives } c ::=\,& \ldots \mid  \sample :\pmonad\II^* \mid \normalreinforce : \RR\times\posreal\to \pmonad\RR^* \mid \forget{\cdot}_\KK : \KK^* \to \KK\\
    & \quad \; \mid\,  \leq : \KK^* \times \KK^* \to \mathbb{B} 
     \mid\,  = : \KK^* \times \KK^* \to \mathbb{B}
    \end{align*} 
        \vspace{-4mm}
    }}
        \vspace{-3mm}
    \caption{Revised syntax for smooth-tracking types}
    \label{fig:06syntax}
\vspace{-3mm}
\end{figure}

%% file: figures/06figures/logrel.tex
\begin{figure}[tb]
\fbox{
  \parbox{.97\textwidth}{
  \vspace{-.2cm}
      \begin{align*}
      \ad{\KK^*} = \KK^* && \rel_{\KK^*} = \{(f : \RR \to \mathbb{K}^*, g : \RR \to \mathbb{K}^*) \mid f \text{ is constant } \wedge f = g\}
        \end{align*}
        \vspace{-.4cm}
  }}
      \vspace{-3mm}
\caption{Definition of the dual-number type and the dual-number logical relation for smooth-tracking types}
\label{fig:06logrel}
\end{figure}

%% file: figures/06figures/example.tex
\begin{figure}[tb]
        \begin{tabular}{|l|l|l|}
         \multicolumn{2}{c}{\bf Accepted by the Type-checker} 
         & 
         \multicolumn{1}{c}{\bf Rejected} \\
         \hline
    \hspace{-2mm}\begin{tabular}{l}
    $\loss_1 = \lambda \theta:\RR.\,  \mbe (\haskdo~\{$\\
    $\quad \var^* \gets \normalreinforce~\theta~1$\\
    $\quad y\,\,\,\gets \normalreparam~\colorbox{green!15}{$\forget{\var^*}$}~1$\\
    $\quad \iif~\colorbox{green!15}{$\var^*\leq 3$}~\then$\\
    $\quad\quad \return~0$\\
    $\quad\eelse$\\
    $\quad\quad \return~-(\theta \div {2})\})$
    \end{tabular}
    \hspace{-3mm}
    &
    \hspace{-2mm}\begin{tabular}{l}
    $\loss_2 = \lambda \theta:\RR. \, \mbe(\haskdo~\{$\\
    $\quad \var\,\, \gets \normalreparam~\theta~1$\\
    $\quad y^* \gets \normalreinforce~\var~1$\\
    $\quad \iif~\colorbox{green!15}{$y^* \leq 3$}~\then$\\
    $\quad\quad \return~0$\\
    $\quad\eelse$\\
    $\quad\quad \return-(\theta \div {2})\})$
    \end{tabular}
    \hspace{-3mm}
    &
    \hspace{-2mm}\begin{tabular}{l}
    $\loss_3 = \lambda \theta:\RR. \, \mbe(\haskdo~\{$\\
    $\quad \var^* \gets \normalreinforce~\theta~1$\\
    $\quad y\,\,\, \gets \normalreparam~\colorbox{green!15}{$\forget{\var^*}$}~1$\\
    $\quad \iif~\colorbox{red!15}{$y \leq 3$}~\then$\\
    $\quad\quad \return~0$\\
    $\quad\eelse$\\
    $\quad\quad \return-(\theta \div {2})\})$
    \end{tabular}
    \hspace{-3mm}
    \\
    \hline
    \end{tabular}
        \vspace{-2mm}
\caption{Smoothness-tracking types allow us to enforce
preconditions for ADEV's correctness.
In these programs, variables of type $\RR^*$ -- those that can be used non-smoothly -- are indicated with a star.
Type checking will reject the unsound program on the right and accept the two programs on the left.
The error comes from the fact that $y:\RR$ cannot be cast to a variable of type $\RR^*$, for use with $\leq$: $y$ has to be used smoothly for $\normalreparam$'s built-in derivative to be correct.}
\vspace{-2mm}
\label{fig:06example}
\end{figure}

%% file: 09recap.tex
\section{Summary: Full Language and ADEV Macro}
\revision{
Our full language is summarized in Figure~\ref{fig:syntax_recap}, and our full AD macro in Figure~\ref{fig:ad_recap}. By combining Thm.~\ref{thm:full-correctness} with Cor.~\ref{cor:all_smooth_types}, we arrive at the following general correctness result:

\begin{corollary}
\label{cor:full-correctness-extended}
Let $\vdash \widetilde{\ter}:\KK \to\eRR$ be a closed term in the full language, where $\KK \in \{\RR, \RR_{\geq 0}, \mathbb{I}\}$. If $\sem{\widetilde{\ter}}(\theta)$ has a well-defined expectation for every $\theta\in \KK$, 
and $\sem{\lambda (\theta,x):\KK\times \RR.\snd(\snd(\ad{\widetilde{\ter}}(\theta,1))(x))}$ is locally dominated,
then $\sem{\lambda \theta:\KK.\snd_*(\fst(\ad{\widetilde{\ter}}(\theta,1)))}$ is a correct unbiased derivative of $\sem{\widetilde{\ter}}$.
When $\widetilde{\ter}$ samples only from finite discrete distributions, the domination condition is always satisfied. 
\end{corollary}


\input{figures/09figures/language}


\input{figures/09figures/ad}



}

%% file: figures/09figures/language.tex
\begin{figure}[t]
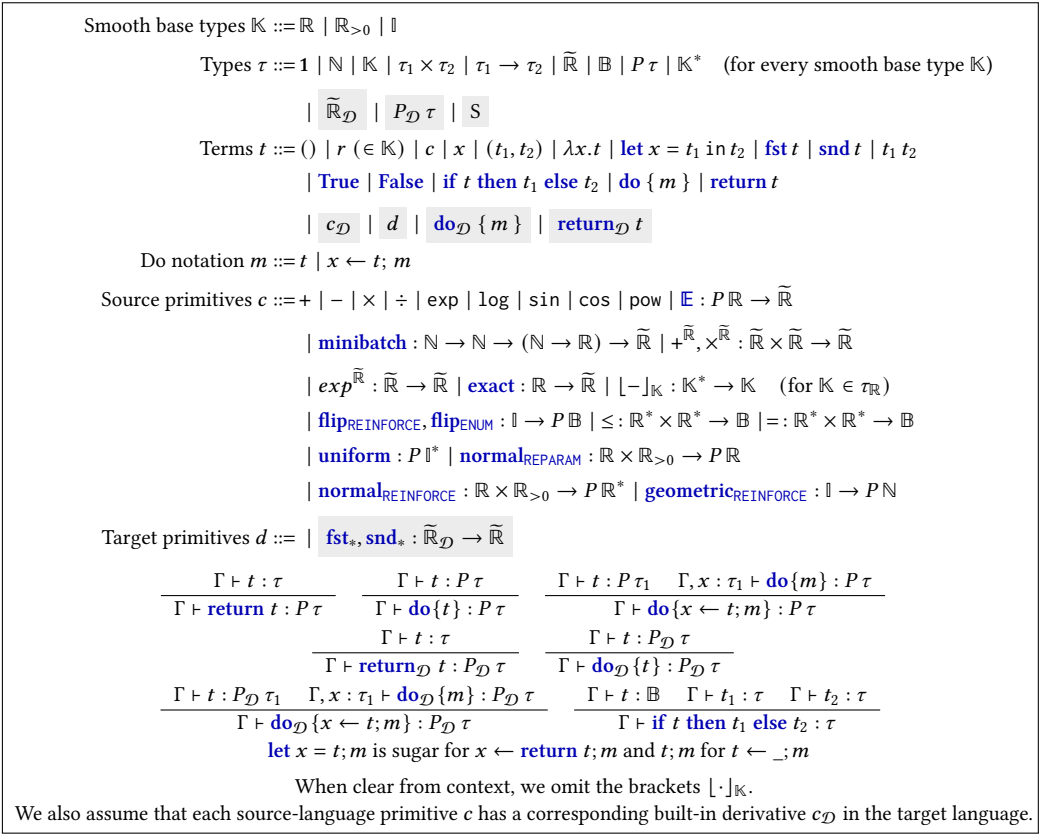

\footnotesize{
\fbox{
  \parbox{.97\textwidth}{
   \vspace{-2mm}
  \revision{
    \centering
    \begin{align*}
    \text{Smooth base types }\KK ::=\,& 
     \RR \mid \posreal \mid \II \\
        \text{Types }\type ::=\,& 
          \mathbf{1} \mid \NN \mid
         \KK \mid
          \type_1 \times \type_2 \mid
          \type_1 \to \type_2  \mid \eRR 
             \mid \BB
        \mid \pmonad\type \mid  \KK^* \quad\text{(for every smooth base type $\KK$)} \\
        & \mid \colorbox{gray!15}{$\deRR$} \mid \colorbox{gray!15}{$\dpmonad\type $} \mid  \colorbox{gray!15}{S} \\
        \text{Terms } \ter ::=\,&
          () \mid r ~(\in \KK) \mid c  \mid
          \var \mid 
          (\ter_1, \ter_2) \mid
          \lambda \var. \ter \mid
          \llet~\var = \ter_1 \,\texttt{in}\, \ter_2\mid
          \fst\,\ter \mid
          \snd\,\ter \mid
          \ter_1\,\ter_2 \\
          & \mid \True \mid \False
        \mid \iifthenelse{\ter}{\ter_1}{\ter_2}
        \mid 
          \haskdo\, \{ \, m \, \}  \mid \return\, \ter
        \\
          & \mid
          \colorbox{gray!15}{$c_\mathcal{D}$} \mid
          \colorbox{gray!15}{$d$} \mid  \colorbox{gray!15}{$\dhaskdo\, \{ \, m \, \} $}\mid
        \colorbox{gray!15}{$\dreturn\, \ter $} \\
         \text{Do notation } m ::=\,&
         \ter \mid
          \var \gets \ter; \, m\\
        \text{Source primitives } c ::=\,& 
          + \mid - \mid \times \mid \div \mid
          \exps \mid \logs \mid \sins \mid \coss \mid \pows  \mid \mbe : \pmonad\RR\to \eRR \\
          & \mid \minibatch : \NN\to\NN\to(\NN\to\RR)\to\eRR \mid \eplus, \etimes : \eRR\times\eRR\to\eRR \\
        & \mid \eexp : \eRR\to\eRR
        \mid \exact:\RR \to \eRR \mid \forget{-}_{\KK} : \KK^*\to\KK\quad(\text{for }\KK\in \type_\RR) \\
        & 
        \mid  \flipreinforce, \flipenum: \II\to \pmonad\BB \mid\, \leq\, : \RR^*\times\RR^*\to\BB   
    \mid\, =\, : \RR^*\times\RR^*\to\BB  \\
        & \mid  \sample :\pmonad\II^* \mid \normalreparam : \RR\times\posreal\to \pmonad\RR \\
    &\mid \normalreinforce : \RR\times\posreal\to \pmonad\RR^* 
    \mid \geometricreinforce : \II \to \pmonad\NN \\
    \text{Target primitives } d ::=\,&  \mid \colorbox{gray!15}{$\fst_*, \snd_*:\eRR_\der\to\eRR$} 
    \end{align*}

    \begin{tabular}{c}
   $\Gamma\vdash \ter:\type$ \\ \hline
   $\Gamma \vdash \return~\ter:\pmonad\type$
\end{tabular}
\quad
\begin{tabular}{c}
    $\Gamma\vdash \ter :\pmonad\type$  \\  \hline
     $\Gamma\vdash \haskdo \{\ter\}:\pmonad\type$
\end{tabular}
\quad
\begin{tabular}{c}
     $\Gamma \vdash\ter:\pmonad\type_1$ \quad 
     $\Gamma,\var:\type_1\vdash \haskdo\{m\} :\pmonad\type$  \\  \hline
     $\Gamma\vdash \haskdo \{\var\gets \ter;m\}:\pmonad\type$
\end{tabular}

 \begin{tabular}{c}
   $\Gamma\vdash \ter:\type$ \\ \hline
   $\Gamma \vdash \dreturn~\ter:\dpmonad\type$
\end{tabular}
\quad
\begin{tabular}{c}
    $\Gamma\vdash \ter :\dpmonad\type$  \\  \hline
     $\Gamma\vdash \dhaskdo \{\ter\}:\dpmonad\type$
\end{tabular}

\begin{tabular}{c}
     $\Gamma \vdash\ter:\dpmonad\type_1$ \quad 
     $\Gamma,\var:\type_1\vdash \dhaskdo\{m\} :\dpmonad\type$  \\  \hline
     $\Gamma\vdash \dhaskdo \{\var\gets \ter;m\}:\dpmonad\type$
\end{tabular}
\quad
\begin{tabular}{c}
     $\Gamma \vdash \ter:\BB$ 
     \quad 
     $\Gamma \vdash \ter_1:\type$  
     \quad 
     $\Gamma \vdash \ter_2:\type$ \\ 
     \hline
     $\Gamma\vdash \iifthenelse{\ter}{\ter_1}{\ter_2}:\type$ 
\end{tabular}
\vspace{-1mm}
    \[\llet~\var =\ter;m \text{ is sugar for }\var\gets \return~\ter;m\text{ and }
    \ter;m\text{ for }\ter\gets \_;m\]

When clear from context, we omit the brackets $\forget{\cdot}_\KK$.

We also assume that each source-language primitive $c$ 
has a corresponding built-in derivative $c_\der$ in the target language. 
    }}}}
      \vspace{-3mm}
    \caption{Full grammar and selected typing rules of the language we study. Gray highlights indicate syntax only present in the target language of the AD macro.}
    \label{fig:syntax_recap}
\vspace{-3mm}
\end{figure}

%% file: figures/09figures/ad.tex
\begin{figure}[tb]
\footnotesize{
\fbox{
  \parbox{.97\textwidth}{
\revision{ 
$\ad{-}$ on contexts
  \begin{center}
     
     \begin{tabular}{ll}
          $\ad{\bullet}$ &= $\bullet$
     \end{tabular}
     \quad
    \begin{tabular}{ll}
        $\ad{\Gamma,\var:\type}$ &= $\ad{\Gamma},\var:\ad{\type}$
    \end{tabular}
  \end{center}

   \vspace{-1mm}
$\ad{-}$ on types
  \begin{center}
\begin{tabular}{ll}
\ad{\eRR} &= $\deRR \times (S\to \RR\times \RR)$ \\
\ad{\pmonad\type} &= $\vpmonad \ad{\type}$ \\
 \ad{\KK} &= $\mathbb{K}\times \RR$ \\
   \ad{\NN} &= $\NN$
\end{tabular}
\begin{tabular}{ll}
    \ad{\type_1\times\type_2} &= $\ad{\type_1}\times \ad{\type_2}$ \\
        \ad{\type_1\to\type_2} &= $\ad{\type_1}\to \ad{\type_2}$ \\
        \ad{\BB} &= $\BB$ \\
        \ad{\mathbb{K}^*} &= $\mathbb{K}^*$
\end{tabular}
\end{center}

 \vspace{-1mm}
$\ad{-}$ on expressions 
\begin{center}
\begin{tabular}{ll}
    \ad{\var} &= $\var$ \\
    \ad{\lambda\var.\ter} &= $\lambda \var.\ad{\ter}$ \\
    \ad{\ter_1\ter_2} &= \ad{\ter_1}\ad{\ter_2} \\
    \ad{\llet ~\var=\ter_1~\iin~\ter_2} &= $\llet~ \var=\ad{\ter_1}~\iin~\ad{\ter_2}$ \\
    \ad{(\ter_1,\ter_2)} &= $(\ad{\ter_1},\ad{\ter_2})$ \\
    \ad{\fst~\ter} &= $\fst~\ad{\ter}$ \\
    \ad{\snd~\ter} &= $\snd~\ad{\ter}$
 \end{tabular} 
 \begin{tabular}{ll}
   \ad{r:\RR} &= $(r,0)$ \\
    \ad{r:\posreal} &= $(r,0)$ \\
    \ad{r:\NN} &= $r$ \\
    \ad{()} &= () \\
    \ad{\return~\ter} &= $\vreturn~\ad{\ter}$ \\
    \ad{\haskdo\{m\}} &= $\vhaskdo~\{\ad{m}\}$ \\
    \ad{\var\gets \ter;m} &= $\var\gets \ad{\ter};\ad{m}$ \\
 \end{tabular} 
 
 We assume built-in primitives $c_\der$ for the derivatives of source primitives $c$, including $\flipenum$, $\flipreinforce$, $\normalreinforce$, $\normalreparam$, $\geometricreinforce$, $\sample$, $\minibatch$, $\exact$, and $\mbe$. For those we have
 \vspace{-1mm}
 \[ \ad{c} = c_\der \]
 \vspace{-2mm}
$\vpmonad \type, \return_\ver,\haskdo_\ver$ are syntactic sugar for the continuation monad given by $\vpmonad \type := (\type\to \ad{\eRR})\to \ad{\eRR}$. 

 \end{center}
}}}}
  \vspace{-3mm}
\caption{Full AD translation $\ad{-}$. We have the following invariant: if $\Gamma\vdash \ter:\type$, then $\ad{\Gamma}\vdash \ad{\ter}:\ad{\type}$. On  terms $\Gamma \vdash \ter : \eRR$, the first projection of $\ad{\ter}$ is the dual-number derivative, and the second is the witness program for the function whose weak domination has to be checked (see Section~\ref{sec:continuous}).
}
\vspace{-3mm}
\label{fig:ad_recap}
\end{figure}
\vspace{-3mm}

%% file: 07relatedwork.tex
\section{Related work}
\label{sec:related_work}

\noindent\textbf{Gradient estimation in machine learning.}
ADEV's primitives compositionally package many gradient estimation strategies
developed in the machine learning community~\citep{mohamed2020monte,kingma2013auto,ranganath2014black,lee2018reparameterization}.
It also extends a growing literature on \textit{stochastic 
computation graphs} (SCGs) 
\citep{schulman2015gradient,weber2019credit,schulman2016optimizing,foerster2018dice},
the goal of which is to help practitioners derive
unbiased gradient estimators for expectations of probabilistic processes
represented as graphs. Recently,~\citet{krieken2021storchastic}
presented Storchastic, a practical system for AD of stochastic computation graphs. 
Storchastic provides reverse-mode AD (often more efficient than the forward-mode AD in our paper); 
and is implemented for PyTorch~\citep{paszke2019pytorch}, a widely used, practical deep learning framework. 
Our work on ADEV is complementary.
We precisely formalize the general problem of automatic differentiation of expected values of probabilistic processes, in a way that applies to broad classes of probabilistic programs (including higher-order)
that cannot easily be represented as computation graphs. Furthermore, our logical relations
allow us to precisely formulate general conditions that new primitives' gradient
estimators must satisfy to be compositionally added to the language. See Appendix~\ref{sec:extending} for further discussion on the consequences of these differences, including: (1) how ADEV can exploit dependency structure that is more explicit in SCGs, (2) how Storchastic gradient estimation methods can be exposed compositionally in ADEV, (3) how ADEV's continuations let it work robustly with multi-sample gradient estimators, whereas Storchastic's broadcasting approach can cause it to fail e.g. in programs with Python \texttt{if} statements, and (4) how higher-order ADEV primitives can encapsulate sophisticated gradient estimation strategies that \textit{don't} decompose into sample-by-sample estimators (as Storchastic's design would require).

Concurrently with our work, \citet{arya2022automatic} developed an intriguing new approach to AD of probabilistic programs, which like ADEV, arises by extending forward-mode AD, but which unlike ADEV, is not based on composing existing, well-understood estimation strategies. It is unclear what source-language features are covered by their algorithm (the authors caution, e.g., that general \texttt{if} statements are unsupported), but the low variance
their estimators appear to achieve may open the door to stable optimization of objectives that have been out of reach using existing estimators. It would be interesting to understand whether their estimators could be exposed compositionally to ADEV users, or even if not, whether the techniques we employ here could be used to prove their 
algorithm sound and extend it to richer source languages.

To our knowledge, among frameworks for deriving unbiased gradient estimators (based on SCGs or \citet{arya2022automatic}'s stochastic triples), ADEV is the only one that handles objectives defined as \textit{functions of} one or more expected values (e.g., $\exp_{\eRR}~(\mbe~p) +_{\eRR} \exp_{\eRR}~(\mbe ~q)$).
\\

\noindent\textbf{Correctness and semantics for probabilistic and differentiable programming.} Partly enabled by
new semantic foundations for probabilistic~\citep{1628343,heunen2017convenient,ehrhard2017measurable} and differentiable~\citep{huot2020correctness,vakar2020denotational,sherman2021} programming, researchers
have recently established a variety of correctness results for both automatic differentiation~\citep{krawiec2022provably,mazza2021automatic,lee2020correctness,abadi-plotkin2020}
and probabilistic program transformations~\citep{scibior2017denotational,lew2019trace,lee2019towards} 
for increasingly expressive languages.
We build most closely on \textit{logical relations} approaches~\citep{katsumata2013relating,ahmed2006step,appel2007very,pientka2019type} for proving properties
of AD algorithms~\citep{huot2020correctness,barthe2020versatility,brunel2019backpropagation,mazza2021automatic}, and on works that use quasi-Borel spaces as a model of
synthetic measure theory~\citep{kock2011commutative,scibior2017denotational,vakar2019domain}. Recent work has begun to formally investigate interactions of differentiability
and probabilistic programming~\citep{mak2021densities,lee2019towards,lew2021towards,sherman2021}, 
but not yet the properties of \textit{AD} in the general probabilistic programming setting.
\\

\noindent\textbf{AD of languages with integration.}
Researchers have recently proposed languages with support
both for integration and AD, including
Teg~\citep{bangaru2021systematically},
a differentiable first-order expression 
language with compact-domain integrals and 
arithmetic, and $\lambda_S$~\citep{sherman2021}, a higher-order language with 
computable integration on $[0, 1]$
as a primitive.  Using compact-domain integration,
it is possible to express some probabilistic program
expectations, but not all (e.g., 
$\lambda_S$ cannot
express probabilistic programs that use
Gaussian distributions). Furthermore, unlike
in Teg and $\lambda_S$, the output of ADEV
is a new probabilistic program, that
can be directly run to produce gradient estimates
for optimization. A unique aspect of Teg is its
support for \textit{parametric discontinuities},
which can sometimes be mimicked in ADEV programs
using discrete random choices like $\texttt{flip}$,
but are in general prohibited by Section~\ref{sec:general}'s 
type system.
\\
\noindent\textbf{AD in PPLs.}
Many practical probabilistic programming languages~\citep{cusumano2019gen,bingham2019pyro,
siddharth2017learning}
support the automated estimation of gradients of a \textit{particular} expectation
with respect to probabilistic programs $q$: the gradient of the ELBO, 
$\nabla_\theta \mathbb{E}_{x \sim q_\theta}[\log p_\theta(x) - \log q_\theta(x)]$.
ADEV formalizes and proves correct a more general algorithm for
arbitrary expected values, giving theory that could help to
understand when these algorithms are correct (as studied in 
a first-order language for independent Gaussians by~\citet{lee2019towards}), and how they can be modularly
extended to support new gradient estimation strategies,
or the estimation of other expectations. 
Many PPLs also rely on AD for reasons \textit{other than} 
differentiating expectations. Typically, these languages 
differentiate \textit{deterministic} programs that are derived from or related to probabilistic ones. 
For example, \citet{NIPS2011_0d7de1ac} differentiate log densities of probabilistic programs, as does the widely-used and highly-optimized Stan~\citep{carpenter2017stan} probabilistic programming system, for use within Hamiltonian Monte Carlo. Venture~\citep{mansinghka2014venture,mansinghka2018probabilistic} and  Gen~\citep{cusumano2019gen} also 
differentiate log densities, for HMC, gradient-based MAP optimization, and Metropolis-Adjusted Langevin Ascent. Gen also computes derivatives of user-defined \textit{involutions} to automatically compute Jacobian 
corrections in reversible-jump MCMC~\citep{cusumano2020automating}. It would be interesting to investigate whether
our semantic setting\textemdash where we can reason about smoothness
via logical relations, and measurability via quasi-Borel semantics\textemdash
could be used to establish the soundness of these PPL applications.

\section{Discussion}
\revision{

\noindent\textbf{Multivariate functions.} To simplify the presentation, we have presented everything in terms of $\RR \to \eRR$ functions with scalar, not vector, inputs and outputs. But the same general strategies used to extend deterministic forward-mode to multivariate functions apply in our case:

Given a term $\vdash t : \RR^n \to P\,\RR^m$, and an input vector $x \in \RR^n$, we can consider the terms $t_{ij} := \lambda \theta : \RR. \mathbb{E} (\haskdo \{ y \gets t(x_1, \dots, x_{i-1}, \theta, x_{i+1}, \dots, x_n); \return~(\pi_j~y)\}$. The translation $\ad{t_{ij}}$ of such a term yields an unbiased estimator of the partial derivative $\frac{\partial y_j}{\partial x_i}$. One (costly) option for estimating the entire Jacobian matrix would be to separately estimate each partial derivative. To reduce the variance of this estimate, the same random seed can be used when generating each term’s estimate, without compromising unbiasedness of the overall estimate.
For a fixed $i$, the computation of $t_{ij}$’s derivative estimate proceeds identically for all $j$; it is only at the end that we extract the $j^{th}$ component of a result vector. There is therefore no need to run the computation $m$ times: we must only run the calculation once for each $i \in \{1, \dots, n\}$, to generate an entire vector of $m$ different $\frac{\partial y_j}{\partial x_i}$ values. This is a well-understood feature of forward-mode AD: it is especially efficient when there are many outputs but few inputs.
As in ordinary forward-mode AD, then, we can compute a Jacobian via $n$
runs of the translated program. Also as in standard forward-mode AD, it is possible to trade memory for time: if instead of dual numbers $\RR \times \RR$ we use dual vectors $\RR \times \RR^n$, we can run the $n$
copies of the computation ‘in parallel.’ 
However, for memory- and time-efficient gradients of functions with high-dimensional inputs, reverse-mode is usually preferred.
\\

\noindent\textbf{Limitations of differentiability analysis.}
Our type system enforces that the user's main program is smooth with respect to the input parameter $\theta$. This limitation has several consequences:
\begin{enumerate}[leftmargin=*]
\item Some ill-typed programs do not have differentiable expectations, so estimating their derivatives is an ill-defined task. We consider rejecting such programs ‘a feature, not a bug.’
\item ADEV’s type system also prevents users from expressing some programs that do have differentiable expectations, but for which efficient gradient estimators are not known or cannot be derived using standard strategies. We would love to differentiate such programs, but we suspect that for expert users hoping to apply ADEV, this limitation would seem natural. (Several recent works~\citep{lee2018reparameterization,bangaru2021systematically} present gradient estimation strategies for restricted classes of discontinuities; these estimators are not yet widely used by practitioners, but we are interested in exploring how they might be incorporated into future versions of ADEV.)
\item Finally, ADEV rejects some programs too eagerly. For example, if a parameter $\theta$ is used non-smoothly but only in a probability-zero set of random executions (i.e., for almost all executions, the function is differentiable for all $\theta$), our type system will reject it, even though our existing gradient estimators would have been correct for the program. More subtly, certain Lipschitz-continuous but non-differentiable uses of a parameter $\theta$ may be permissible, if for any 
$\theta$ the non-differentiability itself is encountered with probability 0 (e.g., $ReLU(x - \theta)$ for $x$ sampled from a Gaussian). A less conservative static analysis could help make ADEV applicable to such programs, which do arise in practice. But we expect this to be a tricky problem. For example, concurrently with our work, \citet{lee2022smoothness}
present a static analysis based on abstract interpretation for careful reasoning about various smoothness properties, including local Lipschitz continuity. Their analysis accepts programs like 
$ReLU(x - \theta)$, 
but it also seems to accept, for example, $ReLU(ReLU(x) - \theta)$,
a term we would want to reject in ADEV (at $\theta = 0$, 
the program is not differentiable for a positive-measure set of x values). We believe that finding more sophisticated static analyses that admit a larger set of programs while still ensuring soundness is an interesting direction for future work, which could broaden the range of applications that AD of probabilistic programs might have.
\end{enumerate}

\noindent\textbf{Haskell prototype.} 
Our Haskell prototype (Appendix~\ref{appx:impl}) is intended as a proof-of-concept illustration of how ADEV integrates with existing libraries for probabilistic and differentiable programming. But with the extensions from Appendix~\ref{sec:extending}, we believe it could be quite usable for practical applications.%
\footnote{Like our theoretical presentation, our implementation extends forward-mode AD, whose cost scales linearly with the number of input parameters. For models with low- to medium-dimensional parameter spaces, forward-mode can be more efficient than reverse-mode, but models containing large neural networks with many parameters, for example, cannot be efficiently differentiated with our current prototype. In the existing literature on AD, improvements to and analyses of reverse-mode algorithms have often built directly on earlier work studying the simpler forward-mode case; our hope is that by showing how standard forward-mode AD algorithms and their proofs can be extended cleanly to handle probabilistic programs, ADEV may lay the groundwork for future research investigating more efficient reverse-mode AD algorithms for probabilistic programs.}
Interestingly, although our analysis does not cover general recursion, our Haskell prototype successfully differentiates many recursive programs. It is also possible, however, to write recursive programs that halt almost surely but whose AD translations do not. For example, consider $\text{geom} = \lambda \theta : \II. \haskdo \{b \gets \flipenum\,\theta; \textbf{if }b\textbf{ then }0\textbf{ else }\haskdo\{ n \gets \text{geom}\, \theta;\, \return\,(n+1)\}\}$. The use of $\flipenum$ causes ADEV’s gradient estimator to attempt an enumeration of program paths, of which there are infinitely many. In this example, the problem could be avoided by using $\flipreinforce$, but it is an open question how to design an AD algorithm and correctness proof that apply to a probabilistic language with general recursion.

}

%% file: 08discussion.tex

%% file: acks.tex
\begin{acks}                            
We have benefited from discussing this work with many friends and colleagues, including Martin Rinard, Tan Zhi-Xuan, Wonyeol Lee, Faustyna Krawiec, Ohad Kammar, Feras Saad, Cathy Wong, McCoy Becker, Cameron Freer, Michele Pagani, Jesse Michel, Ben Sherman, Kevin Mu, Jesse Sigal, Paolo Perrone, Sean Moss, Younesse Kaddar and the Oxford group. We are also grateful to anonymous referees for very helpful feedback.
This material is based on work supported by the NSF Graduate
Research Fellowship under Grant No. 1745302.
Our work is also supported by a Royal Society University Research Fellowship, the ERC BLAST grant, the Air Force Office of Scientific Research (Award No. FA9550–21–1–0038), and the DARPA Machine Common Sense and SAIL-ON projects.
\end{acks}

%% file: appendix/appendix.tex
\appendix
\section*{Appendix}

This appendix is organised as follows. 
We first provide a small Haskell implementation of ADEV. 
Next, we show how to extend ADEV with several primitives and higher-order constructs from the literature. 
Then, we provide some full figures, given as a reference.  
We finish with a more theory-based view and categorical account on our proof strategy, in particular on the shift between Section~\ref{sec:discrete} and Section~\ref{sec:continuous}.


\input{appendix/haskell-impl.tex}
\input{appendix/modularity-primitives-2.tex}
\input{appendix/add-figures.tex}
\input{appendix/qbs.tex}

%% file: appendix/haskell-impl.tex
\section{Haskell Prototype Implementation}
\label{appx:impl}

To back up our claim that ADEV is a modular extension of forward-mode AD,
we developed a prototype Haskell implementation on top of the $\texttt{ad}$~\citep{kmett2021ad} 
and $\texttt{monad-bayes}$~\citep{scibior2017denotational} libraries.
The listing below implements a version of the algorithm that
does not enforce smoothness (Section~\ref{sec:general}) or
output the verification condition (Section~\ref{sec:continuous}).

\small{
\begin{Verbatim}[commandchars=\\\{\}]
\PYG{c+cm}{\PYGZob{}\PYGZhy{}\PYGZsh{} LANGUAGE InstanceSigs, RankNTypes, TypeSynonymInstances, FlexibleInstances,}
\PYG{c+cm}{MultiParamTypeClasses, FunctionalDependencies, ScopedTypeVariables, FlexibleContexts \PYGZsh{}\PYGZhy{}\PYGZcb{}}

\PYG{k+kr}{module}\PYG{+w}{ }\PYG{n+nn}{ADEV}\PYG{+w}{ }\PYG{k+kr}{where}

\PYG{k+kr}{import}\PYG{+w}{ }\PYG{n+nn}{Numeric.Log}\PYG{+w}{  }\PYG{k}{as}\PYG{+w}{ }\PYG{n}{Log}
\PYG{k+kr}{import}\PYG{+w}{ }\PYG{n+nn}{Control.Monad.Bayes.Class}\PYG{+w}{ }\PYG{k}{as}\PYG{+w}{ }\PYG{n}{Bayes}
\PYG{k+kr}{import}\PYG{+w}{ }\PYG{n+nn}{Control.Monad.Cont}
\PYG{k+kr}{import}\PYG{+w}{ }\PYG{n+nn}{Control.Monad}
\PYG{k+kr}{import}\PYG{+w}{ }\PYG{n+nn}{Numeric.AD.Internal.Forward.Double}
\PYG{k+kr}{import}\PYG{+w}{ }\PYG{n+nn}{Control.Monad.Bayes.Sampler.Strict}\PYG{+w}{ }\PYG{p}{(sampleIO)}

\PYG{c+c1}{\PYGZhy{}\PYGZhy{} Typeclass, listing ADEV primitives}
\PYG{k+kr}{class}\PYG{+w}{ }\PYG{p}{(}\PYG{k+kt}{RealFrac}\PYG{+w}{ }\PYG{n}{r}\PYG{p}{,}\PYG{+w}{ }\PYG{k+kt}{Monad}\PYG{+w}{ }\PYG{p}{(}\PYG{n}{p}\PYG{+w}{ }\PYG{n}{m}\PYG{p}{),}\PYG{+w}{ }\PYG{k+kt}{Monad}\PYG{+w}{ }\PYG{n}{m}\PYG{p}{)}\PYG{+w}{ }\PYG{o+ow}{=\PYGZgt{}}\PYG{+w}{ }\PYG{k+kt}{ADEV}\PYG{+w}{ }\PYG{n}{p}\PYG{+w}{ }\PYG{n}{m}\PYG{+w}{ }\PYG{n}{r}\PYG{+w}{ }\PYG{o}{|}\PYG{+w}{ }\PYG{n}{p}\PYG{+w}{ }\PYG{o+ow}{\PYGZhy{}\PYGZgt{}}\PYG{+w}{ }\PYG{n}{r}\PYG{p}{,}\PYG{+w}{ }\PYG{n}{r}\PYG{+w}{ }\PYG{o+ow}{\PYGZhy{}\PYGZgt{}}\PYG{+w}{ }\PYG{n}{p}\PYG{+w}{ }\PYG{k+kr}{where}
\PYG{+w}{  }\PYG{n}{sample}\PYG{+w}{           }\PYG{o+ow}{::}\PYG{+w}{ }\PYG{n}{p}\PYG{+w}{ }\PYG{n}{m}\PYG{+w}{ }\PYG{n}{r}
\PYG{+w}{  }\PYG{n}{flip\PYGZus{}enum}\PYG{+w}{        }\PYG{o+ow}{::}\PYG{+w}{ }\PYG{n}{r}\PYG{+w}{ }\PYG{o+ow}{\PYGZhy{}\PYGZgt{}}\PYG{+w}{ }\PYG{n}{p}\PYG{+w}{ }\PYG{n}{m}\PYG{+w}{ }\PYG{k+kt}{Bool}
\PYG{+w}{  }\PYG{n}{flip\PYGZus{}reinforce}\PYG{+w}{   }\PYG{o+ow}{::}\PYG{+w}{ }\PYG{n}{r}\PYG{+w}{ }\PYG{o+ow}{\PYGZhy{}\PYGZgt{}}\PYG{+w}{ }\PYG{n}{p}\PYG{+w}{ }\PYG{n}{m}\PYG{+w}{ }\PYG{k+kt}{Bool}
\PYG{+w}{  }\PYG{n}{normal\PYGZus{}reparam}\PYG{+w}{   }\PYG{o+ow}{::}\PYG{+w}{ }\PYG{n}{r}\PYG{+w}{ }\PYG{o+ow}{\PYGZhy{}\PYGZgt{}}\PYG{+w}{ }\PYG{n}{r}\PYG{+w}{ }\PYG{o+ow}{\PYGZhy{}\PYGZgt{}}\PYG{+w}{ }\PYG{n}{p}\PYG{+w}{ }\PYG{n}{m}\PYG{+w}{ }\PYG{n}{r}
\PYG{+w}{  }\PYG{n}{normal\PYGZus{}reinforce}\PYG{+w}{ }\PYG{o+ow}{::}\PYG{+w}{ }\PYG{n}{r}\PYG{+w}{ }\PYG{o+ow}{\PYGZhy{}\PYGZgt{}}\PYG{+w}{ }\PYG{n}{r}\PYG{+w}{ }\PYG{o+ow}{\PYGZhy{}\PYGZgt{}}\PYG{+w}{ }\PYG{n}{p}\PYG{+w}{ }\PYG{n}{m}\PYG{+w}{ }\PYG{n}{r}
\PYG{+w}{  }\PYG{n}{expect}\PYG{+w}{           }\PYG{o+ow}{::}\PYG{+w}{ }\PYG{n}{p}\PYG{+w}{ }\PYG{n}{m}\PYG{+w}{ }\PYG{n}{r}\PYG{+w}{ }\PYG{o+ow}{\PYGZhy{}\PYGZgt{}}\PYG{+w}{ }\PYG{n}{m}\PYG{+w}{ }\PYG{n}{r}
\PYG{+w}{  }\PYG{n}{plus\PYGZus{}}\PYG{+w}{            }\PYG{o+ow}{::}\PYG{+w}{ }\PYG{n}{m}\PYG{+w}{ }\PYG{n}{r}\PYG{+w}{ }\PYG{o+ow}{\PYGZhy{}\PYGZgt{}}\PYG{+w}{ }\PYG{n}{m}\PYG{+w}{ }\PYG{n}{r}\PYG{+w}{ }\PYG{o+ow}{\PYGZhy{}\PYGZgt{}}\PYG{+w}{ }\PYG{n}{m}\PYG{+w}{ }\PYG{n}{r}
\PYG{+w}{  }\PYG{n}{times\PYGZus{}}\PYG{+w}{           }\PYG{o+ow}{::}\PYG{+w}{ }\PYG{n}{m}\PYG{+w}{ }\PYG{n}{r}\PYG{+w}{ }\PYG{o+ow}{\PYGZhy{}\PYGZgt{}}\PYG{+w}{ }\PYG{n}{m}\PYG{+w}{ }\PYG{n}{r}\PYG{+w}{ }\PYG{o+ow}{\PYGZhy{}\PYGZgt{}}\PYG{+w}{ }\PYG{n}{m}\PYG{+w}{ }\PYG{n}{r}
\PYG{+w}{  }\PYG{n}{exp\PYGZus{}}\PYG{+w}{             }\PYG{o+ow}{::}\PYG{+w}{ }\PYG{n}{m}\PYG{+w}{ }\PYG{n}{r}\PYG{+w}{ }\PYG{o+ow}{\PYGZhy{}\PYGZgt{}}\PYG{+w}{ }\PYG{n}{m}\PYG{+w}{ }\PYG{n}{r}
\PYG{+w}{  }\PYG{n}{minibatch\PYGZus{}}\PYG{+w}{       }\PYG{o+ow}{::}\PYG{+w}{ }\PYG{k+kt}{Int}\PYG{+w}{ }\PYG{o+ow}{\PYGZhy{}\PYGZgt{}}\PYG{+w}{ }\PYG{k+kt}{Int}\PYG{+w}{ }\PYG{o+ow}{\PYGZhy{}\PYGZgt{}}\PYG{+w}{ }\PYG{p}{(}\PYG{k+kt}{Int}\PYG{+w}{ }\PYG{o+ow}{\PYGZhy{}\PYGZgt{}}\PYG{+w}{ }\PYG{n}{m}\PYG{+w}{ }\PYG{n}{r}\PYG{p}{)}\PYG{+w}{ }\PYG{o+ow}{\PYGZhy{}\PYGZgt{}}\PYG{+w}{ }\PYG{n}{m}\PYG{+w}{ }\PYG{n}{r}
\PYG{+w}{  }\PYG{n}{exact\PYGZus{}}\PYG{+w}{           }\PYG{o+ow}{::}\PYG{+w}{ }\PYG{n}{r}\PYG{+w}{ }\PYG{o+ow}{\PYGZhy{}\PYGZgt{}}\PYG{+w}{ }\PYG{n}{m}\PYG{+w}{ }\PYG{n}{r}

\PYG{c+c1}{\PYGZhy{}\PYGZhy{} \PYGZsq{}forward\PYGZsq{} non\PYGZhy{}AD interpretation of primitives}
\PYG{k+kr}{instance}\PYG{+w}{ }\PYG{k+kt}{MonadDistribution}\PYG{+w}{ }\PYG{n}{m}\PYG{+w}{ }\PYG{o+ow}{=\PYGZgt{}}\PYG{+w}{ }\PYG{k+kt}{ADEV}\PYG{+w}{ }\PYG{k+kt}{IdentityT}\PYG{+w}{ }\PYG{n}{m}\PYG{+w}{ }\PYG{k+kt}{Double}\PYG{+w}{ }\PYG{k+kr}{where}
\PYG{+w}{  }\PYG{n}{sample}\PYG{+w}{           }\PYG{o+ow}{=}\PYG{+w}{ }\PYG{n}{uniform}\PYG{+w}{ }\PYG{l+m+mi}{0}\PYG{+w}{ }\PYG{l+m+mi}{1}
\PYG{+w}{  }\PYG{n}{flip\PYGZus{}enum}\PYG{+w}{        }\PYG{o+ow}{=}\PYG{+w}{ }\PYG{n}{bernoulli}
\PYG{+w}{  }\PYG{n}{flip\PYGZus{}reinforce}\PYG{+w}{   }\PYG{o+ow}{=}\PYG{+w}{ }\PYG{n}{bernoulli}
\PYG{+w}{  }\PYG{n}{normal\PYGZus{}reparam}\PYG{+w}{   }\PYG{o+ow}{=}\PYG{+w}{ }\PYG{n}{normal}
\PYG{+w}{  }\PYG{n}{normal\PYGZus{}reinforce}\PYG{+w}{ }\PYG{o+ow}{=}\PYG{+w}{ }\PYG{n}{normal}
\PYG{+w}{  }\PYG{n}{expect}\PYG{+w}{           }\PYG{o+ow}{=}\PYG{+w}{ }\PYG{n}{runIdentityT}
\PYG{+w}{  }\PYG{n}{exact\PYGZus{}}\PYG{+w}{           }\PYG{o+ow}{=}\PYG{+w}{ }\PYG{n}{return}
\PYG{+w}{  }\PYG{n}{plus\PYGZus{}}\PYG{+w}{ }\PYG{n}{esta}\PYG{+w}{ }\PYG{n}{estb}\PYG{+w}{  }\PYG{o+ow}{=}\PYG{+w}{ }\PYG{n}{pure}\PYG{+w}{ }\PYG{p}{(}\PYG{o}{+}\PYG{p}{)}\PYG{+w}{ }\PYG{o}{\PYGZlt{}*\PYGZgt{}}\PYG{+w}{ }\PYG{n}{esta}\PYG{+w}{ }\PYG{o}{\PYGZlt{}*\PYGZgt{}}\PYG{+w}{ }\PYG{n}{estb}
\PYG{+w}{  }\PYG{n}{times\PYGZus{}}\PYG{+w}{ }\PYG{n}{esta}\PYG{+w}{ }\PYG{n}{estb}\PYG{+w}{ }\PYG{o+ow}{=}\PYG{+w}{ }\PYG{n}{pure}\PYG{+w}{ }\PYG{p}{(}\PYG{o}{*}\PYG{p}{)}\PYG{+w}{ }\PYG{o}{\PYGZlt{}*\PYGZgt{}}\PYG{+w}{ }\PYG{n}{esta}\PYG{+w}{ }\PYG{o}{\PYGZlt{}*\PYGZgt{}}\PYG{+w}{ }\PYG{n}{estb}
\PYG{+w}{  }\PYG{n}{exp\PYGZus{}}\PYG{+w}{ }\PYG{n}{estx}\PYG{+w}{        }\PYG{o+ow}{=}\PYG{+w}{ }\PYG{k+kr}{do}
\PYG{+w}{    }\PYG{n}{n}\PYG{+w}{  }\PYG{o+ow}{\PYGZlt{}\PYGZhy{}}\PYG{+w}{ }\PYG{n}{poisson}\PYG{+w}{ }\PYG{n}{rate}
\PYG{+w}{    }\PYG{n}{xs}\PYG{+w}{ }\PYG{o+ow}{\PYGZlt{}\PYGZhy{}}\PYG{+w}{ }\PYG{n}{replicateM}\PYG{+w}{ }\PYG{n}{n}\PYG{+w}{ }\PYG{n}{estx}
\PYG{+w}{    }\PYG{n}{return}\PYG{+w}{ }\PYG{o}{\PYGZdl{}}\PYG{+w}{ }\PYG{n}{exp}\PYG{+w}{ }\PYG{n}{rate}\PYG{+w}{ }\PYG{o}{*}\PYG{+w}{ }\PYG{n}{product}\PYG{+w}{ }\PYG{p}{(}\PYG{n}{map}\PYG{+w}{ }\PYG{p}{(}\PYG{n+nf}{\PYGZbs{}}\PYG{n}{x}\PYG{+w}{ }\PYG{o+ow}{\PYGZhy{}\PYGZgt{}}\PYG{+w}{ }\PYG{n}{x}\PYG{+w}{ }\PYG{o}{/}\PYG{+w}{ }\PYG{n}{rate}\PYG{p}{)}\PYG{+w}{ }\PYG{n}{xs}\PYG{p}{)}
\PYG{+w}{    }\PYG{k+kr}{where}\PYG{+w}{ }\PYG{n}{rate}\PYG{+w}{ }\PYG{o+ow}{=}\PYG{+w}{ }\PYG{l+m+mi}{2}
\PYG{+w}{  }\PYG{n}{minibatch\PYGZus{}}\PYG{+w}{ }\PYG{n}{n}\PYG{+w}{ }\PYG{n}{m}\PYG{+w}{ }\PYG{n}{f}\PYG{+w}{ }\PYG{o+ow}{=}\PYG{+w}{ }\PYG{k+kr}{do}
\PYG{+w}{    }\PYG{n}{indices}\PYG{+w}{ }\PYG{o+ow}{\PYGZlt{}\PYGZhy{}}\PYG{+w}{ }\PYG{n}{replicateM}\PYG{+w}{ }\PYG{n}{m}\PYG{+w}{ }\PYG{p}{(}\PYG{n}{uniformD}\PYG{+w}{ }\PYG{p}{[}\PYG{l+m+mi}{1}\PYG{o}{..}\PYG{n}{n}\PYG{p}{])}
\PYG{+w}{    }\PYG{n}{vals}\PYG{+w}{    }\PYG{o+ow}{\PYGZlt{}\PYGZhy{}}\PYG{+w}{ }\PYG{n}{mapM}\PYG{+w}{ }\PYG{n}{f}\PYG{+w}{ }\PYG{n}{indices}
\PYG{+w}{    }\PYG{n}{return}\PYG{+w}{ }\PYG{o}{\PYGZdl{}}\PYG{+w}{ }\PYG{p}{(}\PYG{n}{fromIntegral}\PYG{+w}{ }\PYG{n}{n}\PYG{+w}{ }\PYG{o}{/}\PYG{+w}{ }\PYG{n}{fromIntegral}\PYG{+w}{ }\PYG{n}{m}\PYG{p}{)}\PYG{+w}{ }\PYG{o}{*}\PYG{+w}{ }\PYG{p}{(}\PYG{n}{sum}\PYG{+w}{ }\PYG{n}{vals}\PYG{p}{)}

\PYG{c+c1}{\PYGZhy{}\PYGZhy{} AD interpretation of primitives}
\PYG{k+kr}{instance}\PYG{+w}{ }\PYG{k+kt}{MonadDistribution}\PYG{+w}{ }\PYG{n}{m}\PYG{+w}{ }\PYG{o+ow}{=\PYGZgt{}}\PYG{+w}{ }\PYG{k+kt}{ADEV}\PYG{+w}{ }\PYG{p}{(}\PYG{k+kt}{ContT}\PYG{+w}{ }\PYG{k+kt}{ForwardDouble}\PYG{p}{)}\PYG{+w}{ }\PYG{n}{m}\PYG{+w}{ }\PYG{k+kt}{ForwardDouble}\PYG{+w}{ }\PYG{k+kr}{where}
\PYG{+w}{  }\PYG{n}{sample}\PYG{+w}{ }\PYG{o+ow}{=}\PYG{+w}{ }\PYG{k+kt}{ContT}\PYG{+w}{ }\PYG{o}{\PYGZdl{}}\PYG{+w}{ }\PYG{n+nf}{\PYGZbs{}}\PYG{n}{dloss}\PYG{+w}{ }\PYG{o+ow}{\PYGZhy{}\PYGZgt{}}\PYG{+w}{ }\PYG{k+kr}{do}
\PYG{+w}{    }\PYG{n}{u}\PYG{+w}{ }\PYG{o+ow}{\PYGZlt{}\PYGZhy{}}\PYG{+w}{ }\PYG{n}{uniform}\PYG{+w}{ }\PYG{l+m+mi}{0}\PYG{+w}{ }\PYG{l+m+mi}{1}
\PYG{+w}{    }\PYG{n}{dloss}\PYG{+w}{ }\PYG{p}{(}\PYG{n}{bundle}\PYG{+w}{ }\PYG{n}{u}\PYG{+w}{ }\PYG{l+m+mi}{0}\PYG{p}{)}
\PYG{+w}{  }\PYG{n}{flip\PYGZus{}enum}\PYG{+w}{ }\PYG{n}{dp}\PYG{+w}{ }\PYG{o+ow}{=}\PYG{+w}{ }\PYG{k+kt}{ContT}\PYG{+w}{ }\PYG{o}{\PYGZdl{}}\PYG{+w}{ }\PYG{n+nf}{\PYGZbs{}}\PYG{n}{dloss}\PYG{+w}{ }\PYG{o+ow}{\PYGZhy{}\PYGZgt{}}\PYG{+w}{ }\PYG{k+kr}{do}
\PYG{+w}{    }\PYG{n}{dl1}\PYG{+w}{ }\PYG{o+ow}{\PYGZlt{}\PYGZhy{}}\PYG{+w}{ }\PYG{n}{dloss}\PYG{+w}{ }\PYG{k+kt}{True}
\PYG{+w}{    }\PYG{n}{dl2}\PYG{+w}{ }\PYG{o+ow}{\PYGZlt{}\PYGZhy{}}\PYG{+w}{ }\PYG{n}{dloss}\PYG{+w}{ }\PYG{k+kt}{False}
\PYG{+w}{    }\PYG{n}{return}\PYG{+w}{ }\PYG{p}{(}\PYG{n}{dp}\PYG{+w}{ }\PYG{o}{*}\PYG{+w}{ }\PYG{n}{dl1}\PYG{+w}{ }\PYG{o}{+}\PYG{+w}{ }\PYG{p}{(}\PYG{l+m+mi}{1}\PYG{+w}{ }\PYG{o}{\PYGZhy{}}\PYG{+w}{ }\PYG{n}{dp}\PYG{p}{)}\PYG{+w}{ }\PYG{o}{*}\PYG{+w}{ }\PYG{n}{dl2}\PYG{p}{)}
\PYG{+w}{  }\PYG{n}{flip\PYGZus{}reinforce}\PYG{+w}{ }\PYG{n}{dp}\PYG{+w}{ }\PYG{o+ow}{=}\PYG{+w}{ }\PYG{k+kt}{ContT}\PYG{+w}{ }\PYG{o}{\PYGZdl{}}\PYG{+w}{ }\PYG{n+nf}{\PYGZbs{}}\PYG{n}{dloss}\PYG{+w}{ }\PYG{o+ow}{\PYGZhy{}\PYGZgt{}}\PYG{+w}{ }\PYG{k+kr}{do}
\PYG{+w}{    }\PYG{n}{b}\PYG{+w}{           }\PYG{o+ow}{\PYGZlt{}\PYGZhy{}}\PYG{+w}{ }\PYG{n}{bernoulli}\PYG{+w}{ }\PYG{p}{(}\PYG{n}{primal}\PYG{+w}{ }\PYG{n}{dp}\PYG{p}{)}
\PYG{+w}{    }\PYG{p}{(}\PYG{n}{l}\PYG{p}{,}\PYG{+w}{ }\PYG{n}{l\PYGZsq{}}\PYG{p}{)}\PYG{+w}{     }\PYG{o+ow}{\PYGZlt{}\PYGZhy{}}\PYG{+w}{ }\PYG{n}{fmap}\PYG{+w}{ }\PYG{n}{split}\PYG{+w}{ }\PYG{p}{(}\PYG{n}{dloss}\PYG{+w}{ }\PYG{n}{b}\PYG{p}{)}
\PYG{+w}{    }\PYG{k+kr}{let}\PYG{+w}{ }\PYG{n}{logpdf\PYGZsq{}}\PYG{+w}{ }\PYG{o+ow}{=}\PYG{+w}{ }\PYG{n}{tangent}\PYG{+w}{ }\PYG{p}{(}\PYG{n}{log}\PYG{+w}{ }\PYG{o}{\PYGZdl{}}\PYG{+w}{ }\PYG{k+kr}{if}\PYG{+w}{ }\PYG{n}{b}\PYG{+w}{ }\PYG{k+kr}{then}\PYG{+w}{ }\PYG{n}{dp}\PYG{+w}{ }\PYG{k+kr}{else}\PYG{+w}{ }\PYG{l+m+mi}{1}\PYG{+w}{ }\PYG{o}{\PYGZhy{}}\PYG{+w}{ }\PYG{n}{dp}\PYG{p}{)}
\PYG{+w}{    }\PYG{n}{return}\PYG{+w}{ }\PYG{p}{(}\PYG{n}{bundle}\PYG{+w}{ }\PYG{n}{l}\PYG{+w}{ }\PYG{p}{(}\PYG{n}{l\PYGZsq{}}\PYG{+w}{ }\PYG{o}{+}\PYG{+w}{ }\PYG{n}{l}\PYG{+w}{ }\PYG{o}{*}\PYG{+w}{ }\PYG{n}{logpdf\PYGZsq{}}\PYG{p}{))}
\PYG{+w}{  }\PYG{n}{normal\PYGZus{}reparam}\PYG{+w}{ }\PYG{n}{dmu}\PYG{+w}{ }\PYG{n}{dsig}\PYG{+w}{ }\PYG{o+ow}{=}\PYG{+w}{ }\PYG{k+kr}{do}
\PYG{+w}{    }\PYG{n}{deps}\PYG{+w}{ }\PYG{o+ow}{\PYGZlt{}\PYGZhy{}}\PYG{+w}{ }\PYG{n}{stdnorm}
\PYG{+w}{    }\PYG{n}{return}\PYG{+w}{ }\PYG{o}{\PYGZdl{}}\PYG{+w}{ }\PYG{p}{(}\PYG{n}{deps}\PYG{+w}{ }\PYG{o}{*}\PYG{+w}{ }\PYG{n}{dsig}\PYG{p}{)}\PYG{+w}{ }\PYG{o}{+}\PYG{+w}{ }\PYG{n}{dmu}
\PYG{+w}{    }\PYG{k+kr}{where}
\PYG{+w}{      }\PYG{n}{stdnorm}\PYG{+w}{ }\PYG{o+ow}{=}\PYG{+w}{ }\PYG{k+kt}{ContT}\PYG{+w}{ }\PYG{o}{\PYGZdl{}}\PYG{+w}{ }\PYG{n+nf}{\PYGZbs{}}\PYG{n}{dloss}\PYG{+w}{ }\PYG{o+ow}{\PYGZhy{}\PYGZgt{}}\PYG{+w}{ }\PYG{k+kr}{do}
\PYG{+w}{        }\PYG{n}{eps}\PYG{+w}{ }\PYG{o+ow}{\PYGZlt{}\PYGZhy{}}\PYG{+w}{ }\PYG{n}{normal}\PYG{+w}{ }\PYG{l+m+mi}{0}\PYG{+w}{ }\PYG{l+m+mi}{1}
\PYG{+w}{        }\PYG{n}{dloss}\PYG{+w}{ }\PYG{p}{(}\PYG{n}{bundle}\PYG{+w}{ }\PYG{n}{eps}\PYG{+w}{ }\PYG{l+m+mi}{0}\PYG{p}{)}
\PYG{+w}{  }\PYG{n}{normal\PYGZus{}reinforce}\PYG{+w}{ }\PYG{n}{dmu}\PYG{+w}{ }\PYG{n}{dsig}\PYG{+w}{ }\PYG{o+ow}{=}\PYG{+w}{ }\PYG{k+kt}{ContT}\PYG{+w}{ }\PYG{o}{\PYGZdl{}}\PYG{+w}{ }\PYG{n+nf}{\PYGZbs{}}\PYG{n}{dloss}\PYG{+w}{ }\PYG{o+ow}{\PYGZhy{}\PYGZgt{}}\PYG{+w}{ }\PYG{k+kr}{do}
\PYG{+w}{    }\PYG{n}{x}\PYG{+w}{           }\PYG{o+ow}{\PYGZlt{}\PYGZhy{}}\PYG{+w}{ }\PYG{n}{normal}\PYG{+w}{ }\PYG{p}{(}\PYG{n}{primal}\PYG{+w}{ }\PYG{n}{dmu}\PYG{p}{)}\PYG{+w}{ }\PYG{p}{(}\PYG{n}{primal}\PYG{+w}{ }\PYG{n}{dsig}\PYG{p}{)}
\PYG{+w}{    }\PYG{k+kr}{let}\PYG{+w}{ }\PYG{n}{dx}\PYG{+w}{      }\PYG{o+ow}{=}\PYG{+w}{  }\PYG{n}{bundle}\PYG{+w}{ }\PYG{n}{x}\PYG{+w}{ }\PYG{l+m+mi}{0}
\PYG{+w}{    }\PYG{p}{(}\PYG{n}{l}\PYG{p}{,}\PYG{+w}{ }\PYG{n}{l\PYGZsq{}}\PYG{p}{)}\PYG{+w}{     }\PYG{o+ow}{\PYGZlt{}\PYGZhy{}}\PYG{+w}{ }\PYG{n}{fmap}\PYG{+w}{ }\PYG{n}{split}\PYG{+w}{ }\PYG{p}{(}\PYG{n}{dloss}\PYG{+w}{ }\PYG{n}{dx}\PYG{p}{)}
\PYG{+w}{    }\PYG{k+kr}{let}\PYG{+w}{ }\PYG{n}{logpdf\PYGZsq{}}\PYG{+w}{ }\PYG{o+ow}{=}\PYG{+w}{  }\PYG{n}{tangent}\PYG{+w}{ }\PYG{o}{\PYGZdl{}}\PYG{+w}{ }\PYG{p}{(}\PYG{o}{\PYGZhy{}}\PYG{l+m+mi}{1}\PYG{+w}{ }\PYG{o}{*}\PYG{+w}{ }\PYG{n}{log}\PYG{+w}{ }\PYG{n}{dsig}\PYG{p}{)}\PYG{+w}{ }\PYG{o}{\PYGZhy{}}\PYG{+w}{ }\PYG{l+m+mf}{0.5}\PYG{+w}{ }\PYG{o}{*}\PYG{+w}{ }\PYG{p}{((}\PYG{n}{dx}\PYG{+w}{ }\PYG{o}{\PYGZhy{}}\PYG{+w}{ }\PYG{n}{dmu}\PYG{p}{)}\PYG{+w}{ }\PYG{o}{/}\PYG{+w}{ }\PYG{n}{dsig}\PYG{p}{)}\PYG{o}{\PYGZca{}}\PYG{l+m+mi}{2}
\PYG{+w}{    }\PYG{n}{return}\PYG{+w}{ }\PYG{p}{(}\PYG{n}{bundle}\PYG{+w}{ }\PYG{n}{l}\PYG{+w}{ }\PYG{p}{(}\PYG{n}{l\PYGZsq{}}\PYG{+w}{ }\PYG{o}{+}\PYG{+w}{ }\PYG{n}{l}\PYG{+w}{ }\PYG{o}{*}\PYG{+w}{ }\PYG{n}{logpdf\PYGZsq{}}\PYG{p}{))}
\PYG{+w}{  }\PYG{n}{expect}\PYG{+w}{ }\PYG{n}{prog}\PYG{+w}{ }\PYG{o+ow}{=}\PYG{+w}{ }\PYG{n}{runContT}\PYG{+w}{ }\PYG{n}{prog}\PYG{+w}{ }\PYG{n}{return}
\PYG{+w}{  }\PYG{n}{plus\PYGZus{}}\PYG{+w}{ }\PYG{n}{est\PYGZus{}da}\PYG{+w}{ }\PYG{n}{est\PYGZus{}db}\PYG{+w}{ }\PYG{o+ow}{=}\PYG{+w}{ }\PYG{n}{pure}\PYG{+w}{ }\PYG{p}{(}\PYG{o}{+}\PYG{p}{)}\PYG{+w}{ }\PYG{o}{\PYGZlt{}*\PYGZgt{}}\PYG{+w}{ }\PYG{n}{est\PYGZus{}da}\PYG{+w}{ }\PYG{o}{\PYGZlt{}*\PYGZgt{}}\PYG{+w}{ }\PYG{n}{est\PYGZus{}db}
\PYG{+w}{  }\PYG{n}{times\PYGZus{}}\PYG{+w}{ }\PYG{n}{est\PYGZus{}da}\PYG{+w}{ }\PYG{n}{est\PYGZus{}db}\PYG{+w}{ }\PYG{o+ow}{=}\PYG{+w}{ }\PYG{n}{pure}\PYG{+w}{ }\PYG{p}{(}\PYG{o}{*}\PYG{p}{)}\PYG{+w}{ }\PYG{o}{\PYGZlt{}*\PYGZgt{}}\PYG{+w}{ }\PYG{n}{est\PYGZus{}da}\PYG{+w}{ }\PYG{o}{\PYGZlt{}*\PYGZgt{}}\PYG{+w}{ }\PYG{n}{est\PYGZus{}db}
\PYG{+w}{  }\PYG{n}{exp\PYGZus{}}\PYG{+w}{ }\PYG{n}{estimate\PYGZus{}dx}\PYG{+w}{ }\PYG{o+ow}{=}\PYG{+w}{ }\PYG{k+kr}{do}
\PYG{+w}{    }\PYG{p}{(}\PYG{n}{x}\PYG{p}{,}\PYG{+w}{ }\PYG{n}{x\PYGZsq{}}\PYG{p}{)}\PYG{+w}{ }\PYG{o+ow}{\PYGZlt{}\PYGZhy{}}\PYG{+w}{ }\PYG{p}{(}\PYG{n}{fmap}\PYG{+w}{ }\PYG{n}{split}\PYG{+w}{ }\PYG{n}{estimate\PYGZus{}dx}\PYG{p}{)}
\PYG{+w}{    }\PYG{n}{s}\PYG{+w}{ }\PYG{o+ow}{\PYGZlt{}\PYGZhy{}}\PYG{+w}{ }\PYG{n}{exp\PYGZus{}}\PYG{+w}{ }\PYG{p}{(}\PYG{n}{fmap}\PYG{+w}{ }\PYG{n}{primal}\PYG{+w}{ }\PYG{n}{estimate\PYGZus{}dx}\PYG{p}{)}
\PYG{+w}{    }\PYG{n}{return}\PYG{+w}{ }\PYG{p}{(}\PYG{n}{bundle}\PYG{+w}{ }\PYG{n}{x}\PYG{+w}{ }\PYG{p}{(}\PYG{n}{s}\PYG{+w}{ }\PYG{o}{*}\PYG{+w}{ }\PYG{n}{x\PYGZsq{}}\PYG{p}{))}
\PYG{+w}{  }\PYG{n}{minibatch\PYGZus{}}\PYG{+w}{ }\PYG{n}{n}\PYG{+w}{ }\PYG{n}{m}\PYG{+w}{ }\PYG{n}{estimate\PYGZus{}df}\PYG{+w}{ }\PYG{o+ow}{=}\PYG{+w}{ }\PYG{k+kr}{do}
\PYG{+w}{    }\PYG{n}{indices}\PYG{+w}{ }\PYG{o+ow}{\PYGZlt{}\PYGZhy{}}\PYG{+w}{ }\PYG{n}{replicateM}\PYG{+w}{ }\PYG{n}{m}\PYG{+w}{ }\PYG{p}{(}\PYG{n}{uniformD}\PYG{+w}{ }\PYG{p}{[}\PYG{l+m+mi}{1}\PYG{o}{..}\PYG{n}{n}\PYG{p}{])}
\PYG{+w}{    }\PYG{n}{dfs}\PYG{+w}{ }\PYG{o+ow}{\PYGZlt{}\PYGZhy{}}\PYG{+w}{ }\PYG{n}{mapM}\PYG{+w}{ }\PYG{p}{(}\PYG{n+nf}{\PYGZbs{}}\PYG{n}{i}\PYG{+w}{ }\PYG{o+ow}{\PYGZhy{}\PYGZgt{}}\PYG{+w}{ }\PYG{n}{estimate\PYGZus{}df}\PYG{+w}{ }\PYG{n}{i}\PYG{p}{)}\PYG{+w}{ }\PYG{n}{indices}
\PYG{+w}{    }\PYG{n}{return}\PYG{+w}{ }\PYG{o}{\PYGZdl{}}\PYG{+w}{ }\PYG{p}{(}\PYG{n}{sum}\PYG{+w}{ }\PYG{n}{dfs}\PYG{p}{)}\PYG{+w}{ }\PYG{o}{*}\PYG{+w}{ }\PYG{p}{(}\PYG{n}{fromIntegral}\PYG{+w}{ }\PYG{n}{n}\PYG{+w}{ }\PYG{o}{/}\PYG{+w}{ }\PYG{n}{fromIntegral}\PYG{+w}{ }\PYG{n}{m}\PYG{p}{)}
\PYG{+w}{  }\PYG{n}{exact\PYGZus{}}\PYG{+w}{ }\PYG{o+ow}{=}\PYG{+w}{ }\PYG{n}{return}

\PYG{c+c1}{\PYGZhy{}\PYGZhy{} Derivative operator}
\PYG{n+nf}{diff}\PYG{+w}{ }\PYG{o+ow}{::}\PYG{+w}{ }\PYG{k+kt}{MonadDistribution}\PYG{+w}{ }\PYG{n}{m}\PYG{+w}{ }\PYG{o+ow}{=\PYGZgt{}}\PYG{+w}{ }\PYG{p}{(}\PYG{k+kt}{ForwardDouble}\PYG{+w}{ }\PYG{o+ow}{\PYGZhy{}\PYGZgt{}}\PYG{+w}{ }\PYG{n}{m}\PYG{+w}{ }\PYG{k+kt}{ForwardDouble}\PYG{p}{)}\PYG{+w}{ }\PYG{o+ow}{\PYGZhy{}\PYGZgt{}}\PYG{+w}{ }\PYG{k+kt}{Double}\PYG{+w}{ }\PYG{o+ow}{\PYGZhy{}\PYGZgt{}}\PYG{+w}{ }\PYG{n}{m}\PYG{+w}{ }\PYG{k+kt}{Double}
\PYG{n+nf}{diff}\PYG{+w}{ }\PYG{n}{f}\PYG{+w}{ }\PYG{n}{x}\PYG{+w}{ }\PYG{o+ow}{=}\PYG{+w}{ }\PYG{k+kr}{do}
\PYG{+w}{  }\PYG{n}{df}\PYG{+w}{ }\PYG{o+ow}{\PYGZlt{}\PYGZhy{}}\PYG{+w}{ }\PYG{n}{f}\PYG{+w}{ }\PYG{p}{(}\PYG{n}{bundle}\PYG{+w}{ }\PYG{n}{x}\PYG{+w}{ }\PYG{l+m+mi}{1}\PYG{p}{)}
\PYG{+w}{  }\PYG{n}{return}\PYG{+w}{ }\PYG{p}{(}\PYG{n}{tangent}\PYG{+w}{ }\PYG{n}{df}\PYG{p}{)}

\PYG{c+c1}{\PYGZhy{}\PYGZhy{} Example program l : R \PYGZhy{}\PYGZgt{} MR}
\PYG{n+nf}{l}\PYG{+w}{ }\PYG{o+ow}{::}\PYG{+w}{ }\PYG{k+kt}{ADEV}\PYG{+w}{ }\PYG{n}{p}\PYG{+w}{ }\PYG{n}{m}\PYG{+w}{ }\PYG{n}{r}\PYG{+w}{ }\PYG{o+ow}{=\PYGZgt{}}\PYG{+w}{ }\PYG{n}{r}\PYG{+w}{ }\PYG{o+ow}{\PYGZhy{}\PYGZgt{}}\PYG{+w}{ }\PYG{n}{m}\PYG{+w}{ }\PYG{n}{r}
\PYG{n+nf}{l}\PYG{+w}{ }\PYG{n}{theta}\PYG{+w}{ }\PYG{o+ow}{=}\PYG{+w}{ }\PYG{n}{expect}\PYG{+w}{ }\PYG{o}{\PYGZdl{}}\PYG{+w}{ }\PYG{k+kr}{do}
\PYG{+w}{  }\PYG{n}{b}\PYG{+w}{ }\PYG{o+ow}{\PYGZlt{}\PYGZhy{}}\PYG{+w}{ }\PYG{n}{flip\PYGZus{}reinforce}\PYG{+w}{ }\PYG{n}{theta}
\PYG{+w}{  }\PYG{k+kr}{if}\PYG{+w}{ }\PYG{n}{b}\PYG{+w}{ }\PYG{k+kr}{then}
\PYG{+w}{    }\PYG{n}{return}\PYG{+w}{ }\PYG{l+m+mi}{0}
\PYG{+w}{  }\PYG{k+kr}{else}
\PYG{+w}{    }\PYG{n}{return}\PYG{+w}{ }\PYG{p}{(}\PYG{o}{\PYGZhy{}}\PYG{n}{theta}\PYG{+w}{ }\PYG{o}{/}\PYG{+w}{ }\PYG{l+m+mi}{2}\PYG{p}{)}

\PYG{c+c1}{\PYGZhy{}\PYGZhy{} Run Stochastic Gradient Descent}
\PYG{n+nf}{sgd}\PYG{+w}{ }\PYG{o+ow}{::}\PYG{+w}{ }\PYG{k+kt}{MonadDistribution}\PYG{+w}{ }\PYG{n}{m}\PYG{+w}{ }\PYG{o+ow}{=\PYGZgt{}}\PYG{+w}{ }\PYG{p}{(}\PYG{k+kt}{ForwardDouble}\PYG{+w}{ }\PYG{o+ow}{\PYGZhy{}\PYGZgt{}}\PYG{+w}{ }\PYG{n}{m}\PYG{+w}{ }\PYG{k+kt}{ForwardDouble}\PYG{p}{)}
\PYG{+w}{       }\PYG{o+ow}{\PYGZhy{}\PYGZgt{}}\PYG{+w}{ }\PYG{k+kt}{Double}\PYG{+w}{ }\PYG{o+ow}{\PYGZhy{}\PYGZgt{}}\PYG{+w}{ }\PYG{k+kt}{Double}\PYG{+w}{ }\PYG{o+ow}{\PYGZhy{}\PYGZgt{}}\PYG{+w}{ }\PYG{k+kt}{Int}\PYG{+w}{ }\PYG{o+ow}{\PYGZhy{}\PYGZgt{}}\PYG{+w}{ }\PYG{n}{m}\PYG{+w}{ }\PYG{p}{[}\PYG{k+kt}{Double}\PYG{p}{]}
\PYG{n+nf}{sgd}\PYG{+w}{ }\PYG{n}{loss}\PYG{+w}{ }\PYG{n}{eta}\PYG{+w}{ }\PYG{n}{x0}\PYG{+w}{ }\PYG{n}{steps}\PYG{+w}{ }\PYG{o+ow}{=}
\PYG{+w}{  }\PYG{k+kr}{if}\PYG{+w}{ }\PYG{n}{steps}\PYG{+w}{ }\PYG{o}{==}\PYG{+w}{ }\PYG{l+m+mi}{0}\PYG{+w}{ }\PYG{k+kr}{then}
\PYG{+w}{    }\PYG{n}{return}\PYG{+w}{ }\PYG{p}{[}\PYG{n}{x0}\PYG{p}{]}
\PYG{+w}{  }\PYG{k+kr}{else}\PYG{+w}{ }\PYG{k+kr}{do}
\PYG{+w}{    }\PYG{n}{v}\PYG{+w}{ }\PYG{o+ow}{\PYGZlt{}\PYGZhy{}}\PYG{+w}{ }\PYG{n}{diff}\PYG{+w}{ }\PYG{n}{loss}\PYG{+w}{ }\PYG{n}{x0}
\PYG{+w}{    }\PYG{k+kr}{let}\PYG{+w}{ }\PYG{n}{x1}\PYG{+w}{ }\PYG{o+ow}{=}\PYG{+w}{ }\PYG{n}{x0}\PYG{+w}{ }\PYG{o}{\PYGZhy{}}\PYG{+w}{ }\PYG{n}{eta}\PYG{+w}{ }\PYG{o}{*}\PYG{+w}{ }\PYG{n}{v}
\PYG{+w}{    }\PYG{n}{xs}\PYG{+w}{ }\PYG{o+ow}{\PYGZlt{}\PYGZhy{}}\PYG{+w}{ }\PYG{n}{sgd}\PYG{+w}{ }\PYG{n}{loss}\PYG{+w}{ }\PYG{n}{eta}\PYG{+w}{ }\PYG{n}{x1}\PYG{+w}{ }\PYG{p}{(}\PYG{n}{steps}\PYG{+w}{ }\PYG{o}{\PYGZhy{}}\PYG{+w}{ }\PYG{l+m+mi}{1}\PYG{p}{)}
\PYG{+w}{    }\PYG{n}{return}\PYG{+w}{ }\PYG{p}{(}\PYG{n}{x0}\PYG{k+kt}{:}\PYG{n}{xs}\PYG{p}{)}

\PYG{n+nf}{main}\PYG{+w}{ }\PYG{o+ow}{::}\PYG{+w}{ }\PYG{k+kt}{IO}\PYG{+w}{ }\PYG{n+nb}{()}
\PYG{n+nf}{main}\PYG{+w}{ }\PYG{o+ow}{=}\PYG{+w}{ }\PYG{k+kr}{do}
\PYG{+w}{  }\PYG{n}{vs}\PYG{+w}{ }\PYG{o+ow}{\PYGZlt{}\PYGZhy{}}\PYG{+w}{ }\PYG{n}{sampleIO}\PYG{+w}{ }\PYG{o}{\PYGZdl{}}\PYG{+w}{ }\PYG{n}{sgd}\PYG{+w}{ }\PYG{n}{l}\PYG{+w}{ }\PYG{l+m+mf}{0.2}\PYG{+w}{ }\PYG{l+m+mf}{0.2}\PYG{+w}{ }\PYG{l+m+mi}{100}
\PYG{+w}{  }\PYG{n}{print}\PYG{+w}{ }\PYG{n}{vs}
\end{Verbatim}    
}

%% file: appendix/modularity-primitives-2.tex
\section{Extending ADEV}
\label{sec:extending}

Our paper builds the ADEV algorithm one piece at a time, repeatedly adding new types, constructs, and primitives to increase the expressiveness of the language, modularly extending the correctness proof at each step. In this section,
we demonstrate by example that ADEV can be modularly extended in many more useful directions (all of which are implemented in our Haskell prototype at \url{https://github.com/probcomp/adev}):

\begin{itemize}
    \item We add a new primitive for variance reduction based on control variates ($\baseline$, \ref{sub:baseline}).
    \item We add new constructs that let users expose structure in a loss function, enabling ADEV to exploit the ``credit assignment'' variance reduction technique from~\citet{schulman2015gradient} ($W\!P~\tau$ and $\addcost$, \ref{sub:scg}).
    \item We introduce new constructs for representing distributions with known density functions, and operations that use those densities to automatically construct gradient estimators ($D~\sigma, \reinforce$, \ref{sec:densities}). 
    \item We show how to add multi-sample gradient estimators from the Storchastic framework~\citep{krieken2021storchastic} and discuss pros and cons of Storchastic's vs. ADEV's interfaces ($\leaveoneout$, \ref{sub:storchastic}).
    \item We add a higher-order primitive for sequential Monte Carlo, with custom derivative logic that exploits \citet{scibior2021differentiable}'s differentiable particle filter estimator ($\smc$, \ref{sub:particle-filter}).
    \item We demonstrate how \texttt{stop-grad}-like operations can be justified if they are encapsulated within the implementations of certain primitives ($\importance$, \ref{sub:stop-grad}).
    \item We show how implicit reparameterization can be used to create gradient estimators for some distributions ($C~\RR$ and $\implicitdiff$, \ref{sub:implicit-diff}). 
    \item We show how weak or measure-valued derivatives can be incorporated as estimators ($\poissonweak$, \ref{sub:weak-deriv}). 
    \item We add a higher-order primitive for a reparametized rejection sampler from \citet{naesseth2017reparameterization} ($\reparamreject$, \ref{sub:reject})
\end{itemize}

\subsection{Controlling Variance with Baselines}
\label{sub:baseline}

Suppose $p : \RR \to P~X$ and we wish to estimate $\mathbb{E}_{x \sim p(\theta)}[f(x)]$ and its derivative with respect to $\theta$, 
for some function $f : X \to \RR$.

For some estimators (e.g., the REINFORCE estimator), the variance of the gradient estimate may grow with the magnitude of $f$. 
In these cases it can be useful to ``center'' the loss function $f$: 
instead of passing $\lambda x. f(x)$ to $p$'s gradient estimator, we pass $\lambda x. f(x) - c(\theta)$ for some \textit{baseline} $c$, yielding a (hopefully lower-variance) estimate of $\frac{d}{d\theta}\left(\mathbb{E}_{x \sim p_\theta}[f(x)] - c(\theta)\right)$, to which we must re-add $c'(\theta)$ to obtain an estimate of $\frac{d}{d\theta}\mathbb{E}_{x \sim p_\theta}[f(x)]$. 

In ADEV, we can expose this technique to users via a primitive $\baseline : P~\RR \to \RR \to \eRR$. Semantically, $\sem{\baseline~p~b} = \sem{\mbe~p}$, but its built-in dual-number derivative is distinct:

\begin{center}
\begin{algorithm}[H]
\DontPrintSemicolon
\SetKwProg{Fn}{}{:}{end}
\Fn{$\baseline_\der(\widetilde{dp}: P_\mathcal{D}~\ad{\RR}, db : \ad{\RR})$}{
$dl\sim \widetilde{dp}(\lambda dx. \exact_\mathcal{D}~ (dx -_\mathcal{D} db))$ \;
\Return $dl +_\mathcal{D}~db$
}
\end{algorithm}  
\end{center}

\noindent Given $\widetilde{dp} : (\ad{\RR} \to \deRR) \to \deRR$ and $db : \ad{\RR}$, the derivative first calls $dp$ on $\lambda dr. \exact_\mathcal{D} (dr -_{\mathcal{D}}~db)$ to obtain a dual-number loss estimate $dl$, then returns $dl +_\mathcal{D}~db$.

\subsection{Accounting for the Dependency Graph: Stochastic Computation Graphs}
\label{sub:scg}

ADEV generalizes the stochastic computation graphs (SCG) framework~\citep{schulman2015gradient} from computation 
graphs to higher-order probabilistic programs. On programs that use a combination of REINFORCE- and REPARAM-based
primitives, the resulting ADEV estimators often resemble the SCG estimator for a particular graph. 
One aspect of SCGs that vanilla ADEV fails to capture, however, is their tracking of \textit{dependence relationships} between primitive random choices and additive terms in the loss. 

For example, consider the program $$\lambda \theta : \II. \mbe (\haskdo~\{x \gets \flipreinforce~\theta; y \gets \flipreinforce~\theta; \return~(c_1(x)+c_2(x,y))\})$$ for two loss functions $c_1 : \BB \to \RR$ and $c_2 : \BB \times \BB \to \RR$. In the SCG framework, 
we might represent this program as a graph with \textit{four} nodes: two stochastic nodes, for $x$ and $y$, and two \textit{cost nodes}, for $c_1$ and $c_2$, with a directed edge from $x$ to $c_1$ and $c_2$, and from $y$ to $c_2$. This graph structure captures the fact that the term $c_1(x)$ does not depend on $y$. The SCG estimator exploits this fact to reduce the variance of the resulting estimator,
which samples $x$ and $y$ from their Bernoulli distributions, and computes $(\frac{d}{d\theta}\log \text{Bern}(x; \theta))(c_1(x) + c_2(x, y)) + (\frac{d}{d\theta} \log \text{Bern}(y; \theta))(c_2(x, y))$. Note that the derivative of $y$'s log density is multiplied only by $c_2$, i.e., only by the portion of the loss function for which $y$ should ``get credit.'' By contrast, ADEV sees the term $c_1(x) + c_2(x, y)$ as a monolithic value, and uses the sound but generally higher-variance estimator $(\frac{d}{d\theta}\log \text{Bern}(x; \theta))(c_1(x) + c_2(x, y)) + (\frac{d}{d\theta} \log \text{Bern}(y; \theta))(c_1(x) + c_2(x, y))$. 

We can fix this by making the additive structure of the loss function explicit. We replace the monad $P~\tau$ with the monad $W\!P~\tau$, which uses the writer monad transformer to explicitly track an accumulated loss as a program executes. More precisely, $\sem{W\!P~\tau} = \sem{P\,(\tau \times \RR)}$. The unit of the monad, $\return_{W\!P} : \tau \to W\!P~\tau$, deterministically returns its argument and the accumulated loss $0$. To sequence computations, we write
$\haskdo_{W\!P} \{x \gets t; m\}$, which first runs $t$ to generate $(x, l)$, then runs $\haskdo_{W\!P} \{m(x)\}$ to generate $(y, l')$, and finally returns $(y, l + l')$. We extend the AD macro to cover $W\!P~\tau$, setting $\ad{W\!P~\tau} := \ad{P~\tau}$, with the exact same translations for $\return_{W\!P}$ and $\haskdo_{W\!P}$ as we had for $\return$ and $\haskdo$. The difference is in our \textit{correctness requirement} for the new translations: a function $g : \RR \to P_\mathcal{D}~\tau$ is a derivative of a function $f : \RR \to W\!P~\tau$ if
whenever $(h, j) \in R_{\tau \to \eRR}$, we have that for all $\theta \in \RR$, $g(\theta)(j(\theta)) : \deRR$ is an unbiased dual-number estimator of the value and derivative of 
$\lambda \theta. \mathbb{E}_{(x, w) \sim f(\theta)}[\mathbb{E}_{y \sim h(\theta)(x)}[y] + w]$. 

In other words: a program of type $W\!P~\tau$ represents a distribution over \textit{pairs} $(x, w) \in \sem{\tau} \times \RR$, and
when we think about the expectation of a function $f : \sem{\tau} \to \RR$ under this distribution, we always treat $w$ additively, computing $\mathbb{E}[f(x) + w]$.  This development allows us to define the primitive $\addcost : \RR \to W\!P~1$, which, given a number $w$, deterministically returns $((), w)$. The \textit{built-in derivative} for $\addcost$ accepts as input a dual-number version $dw : \ad{\RR}$ of $w$, and a loss-to-go $\widetilde{dl} : 1 \to \deRR$. It samples $dl \sim \widetilde{dl}$ and then returns $dl +_\mathcal{D} dw$. 

Why does this development help solve the problem discussed above? It allows us to rewrite our example program as $$\lambda \theta : \II. \mbe (\haskdo_{W\!P}~\{x \gets \flipreinforce~\theta; \addcost~(c_1(x)); y \gets \flipreinforce~\theta; \addcost~(c_2(x, y)); \return~0\})$$
(assuming that we lift the primitive $\flipreinforce$ to be of type $\II \to W\!P~\BB$, which we can do by having it return both the Boolean value it flips and the accumulated loss $0$). The structure of the program now makes explicit that $c_1$ is added to the loss \textit{before} $y$ is sampled, and so $y$ cannot affect the value of that term. And indeed, if we apply ADEV to our modified program\textemdash using the built-in derivative of $\addcost$ described above\textemdash we recover the same lower-variance estimator that
the SCG framework yields. At the same time, we have not \textit{broken} anything else about ADEV: it still supports higher-order functions, as well as the other extensions described in this section, which together enable a broad class of estimators, some of which the SCG framework cannot compositionally express.

The dependency tracking enabled by this simple writer monad captures only sequential dependencies: everything earlier in a program can affect everything later in a program.\footnote{Storchastic~\citep{krieken2021storchastic}, a PyTorch framework for Stochastic Computation Graphs, makes a similar design decision, exposing an $\addcost$ primitive. However, because PyTorch explicitly builds a graph of Tensors, it is possible to obtain a conservative overapproximation of which nodes in particular affect the new cost being added. Our source-to-source transformation does not have the same property.} More sophisticated monads could be used to track more interesting dependency relationships, but would require users to make independence in their programs explicit, e.g. using a special combinator $p \otimes q$ to compute $\haskdo \{ x \gets p; y \gets q; \return~(x, y)\}$ in parallel, instead of writing the program sequentially with $\haskdo$. Such combinators, explicitly 
representing conditional independence relationships, do show up in existing probabilistic programming systems (e.g., \texttt{plate} and
\texttt{markov} in Pyro~\citep{bingham2019pyro}, or \texttt{Map}, \texttt{Unfold}, and \texttt{Recurse} in Gen~\citep{cusumano2019gen}),
so perhaps this is a reasonable way forward, to combine the benefits of graph-based dependency tracking with the expressive power of 
higher-order probabilistic programming. Future work could also investigate ways to automatically detect conditional independence 
relationships using ideas from information-flow analysis, or by adapting existing type-directed program slicing techniques~\citep{gorinova2021conditional}.

 
\subsection{Density-Carrying Distributions}\label{sec:densities}
So far, we have defined \textit{different} primitives whose built-in derivatives
\textit{each} implement the REINFORCE gradient estimator, but for different distributions.
Can this duplication be avoided? Can we express the REINFORCE estimation strategy as its
own primitive? Intuitively, REINFORCE applies to any distribution that we can sample from 
and that we can differentiate the log density of (modulo the dominated convergence conditions discussed 
in Section~\ref{sec:continuous}). 

Following existing work on probabilistic programming~\citep{lew2019trace}, we can add
types $D\,\sigma$ of \textit{distributions with densities} for ground types $\sigma$.
For each ground type $\sigma$, we define a reference measure $\mu_\sigma$, and our semantics 
interprets the type $\sem{D\,\sigma} = P\,\sem{\sigma} \times (\sem{\sigma} \to \RR_{\geq 0})$
as the space of \textit{pairs} of measures with density functions. Under AD, 
we have $\ad{D\,\sigma} = D\,\sigma \times (\ad{\sigma} \to \ad{\RR})$: the distribution 
is left alone, but the density function is differentiated. Our logical relation $\mathcal{R}_{D\,\tau}$
relates a parameterized distribution $p : \RR \to P\,\sem{\sigma} \times (\sem{\sigma} \to \RR_{\geq 0})$ to its derivative $dp : \RR \to (P\,\sem{\sigma} \times (\sem{\sigma} \to \RR_{\geq 0})) \times (\ad{\sigma} \to \RR \times \RR)$ if: (1) for all $\theta$, $\pi_1(p(\theta))$ has density $\pi_2(p(\theta))$ with respect to $\mu_\sigma$, (2) the density $\pi_2(p(\theta))$ is $\mu_\sigma$-almost-everywhere non-zero, (3) $\pi_1 \circ dp = p$, and (4) $(\pi_2 \circ p, \pi_2 \circ dp) \in \mathcal{R}_{\sigma \to \RR_{\geq 0}}$. In words, the density needs to match the distribution, the distribution needs to have full support (non-zero density), and the derivative of the density needs to be correct.

Using this, we can implement 
a primitive $\reinforce  : D\,\sigma \to P\,\sigma$. The semantics of $\reinforce$ is to produce a $P\,\sigma$ representing the same distribution as the $D\,\sigma$ does (formally, $\sem{\reinforce} = \pi_1$). However, the built-in derivative for the resulting $P\,\sigma$ 
uses the sampler, density, and density derivative to implement the REINFORCE estimator for
the distribution in question:

\begin{center}
\begin{algorithm}[H]
\DontPrintSemicolon
\SetKwProg{Fn}{}{:}{end}
\Fn{$\reinforce_\der(dp : D\,\sigma \times (\ad{\sigma} \to \ad{\RR}), \widetilde{dl} : \ad{\sigma} \to \deRR)$}{
$x \sim \fst(\fst(dp))$\;
$(l, \delta l) \sim \widetilde{dl}(x)$\;
$(\_, \delta d) \gets \log_\mathcal{D}~ ((\snd~dp)(x))$\;
\Return $(l, \delta d \cdot l + \delta l)$
}
\end{algorithm}  
\end{center}

Furthermore, we can follow the design of probabilistic programming languages like Gen~\citep{cusumano2019gen}, Pyro~\citep{bingham2019pyro}, and ProbTorch~\citep{stites2021learning} to provide programming constructs 
for building new values of type $D\,\sigma$ from primitives. For example, given $p : D\,\sigma_1$ and $k : \sigma_1 \to D\,\sigma_2$,
it is straightforward to create the ``dependent product measure'' $p~\otimes\!\!=~k : D\,(\sigma_1 \times \sigma_2)$ representing the
distribution that arises if $x \sim p$, $y \sim k(x)$, and $(x, y)$ is returned. (The usual monadic bind is 
more difficult to implement, since the density of the resulting program is a (generally intractable) integral 
over all possible values of $x \in \sem{\sigma_1}$. This version side-steps the problem by remembering $x$, 
so the density is just the product of $p$'s density and $k$'s density. \citet{lew2019trace} use this concatenative 
bind operation to define a \textit{graded} monad for probabilistic programs, which enables the compositional
program-like construction of probability distributions and corresponding density functions over records, lists, and 
sum types that record the choices made by programs with sequencing, looping, and branching.)

\subsection{Adapting Estimators from Storchastic}
\label{sub:storchastic}

Storchastic~\citep{krieken2021storchastic} is a PyTorch framework for 
gradient estimation on stochastic computation graphs~\citep{schulman2015gradient}.
As in ADEV, Storchastic users can choose different gradient estimators at each 
primitive sampling statement, and can add new gradient estimation strategies modularly;
in Storchastic this is done by specifying a four-tuple of a \textit{proposal}, \textit{weighting function},
\textit{gradient function}, and \textit{control variate}, satisfying certain properties.\footnote{
Note that ADEV gives a specification for a correct custom derivative for new primitives of \textit{any} type,
including higher-order primitives. Storchastic's interface, by contrast, only allows adding new primitives 
of type $P~\tau$, where $\tau$ is a ground type.
}
Gradient estimation strategies suitable for use with Storchastic can generally also 
be incorporated into ADEV via the introduction of new primitives.

For example, consider the Leave-One-Out score function estimator that~\citet{krieken2021storchastic}
give as their example method. Let $N \in \NN_{\geq 2}$, and let $p(x; \theta)$ be a probability mass function on a space $\sem{\tau}$ 
(for simplicity, we consider a discrete space, with $\ad{\tau} = \tau$, but reals would work too) 
parameterized by $\theta \in \RR$.
Then the following is a valid built-in derivative for a primitive that samples from $p$, using the leave-one-out gradient estimator:

\begin{center}
\begin{algorithm}[H]
\DontPrintSemicolon
\SetKwProg{Fn}{}{:}{end}
\Fn{${p_\texttt{LEAVE\_ONE\_OUT}}_\der(d\theta : \ad{\RR}, \widetilde{dl} : \ad{\tau} \to \deRR)$}{
\For{$i \in \{1, \dots, N\}$}{
$x_i \sim p(\cdot; \fst(d\theta))$ \;
$(l_i, \delta l_i) \sim \widetilde{dl}(x)$
}
\For{$i \in \{1, \dots, N\}$}{
$b_i \gets \frac{1}{N-1} \sum_{j \neq i} l_i$ \;
$\nabla_i \gets (\snd (\log_\mathcal{D} (p_\mathcal{D}(x_i, d\theta)))) \cdot (l_i - b_i) + \delta l_i$
}
\Return $(\frac{1}{N} \sum_{i=1}^N l_i, \frac{1}{N} \sum_{i=1}^N \nabla_i)$ \;
}
\end{algorithm}  
\end{center}

Like the $\flipenum$ example we gave in the main paper, this primitive's gradient estimator involves
evaluating the ``rest of the program'' $\widetilde{dl}$ on multiple values $x_{1:N}$. This is accomplished
in ADEV using continuations. By contrast, in Storchastic the samples are packed into a vector, and 
the rest of the program is executed on that vector, yielding a vector of losses. Storchastic's approach
may have the benefit of computing the loss on the various samples \textit{in parallel}, depending on the vector
operations supported by the user's hardware. However, it is also less robust than ADEV's continuation-based
approach. For example, Storchastic's version of this estimator will fail if the user's program samples 
from this primitive and uses the result to compute the condition of a Python $\texttt{if}$ statement.

Although the above is carried out for a specific $p$ and $N$, a higher-order primitive $\leaveoneout : \NN \to D\,\sigma \to  P\,\sigma$ can be formulated using the density-carrying types from the previous section. It behaves like $\reinforce$ (and, as its type suggests, is a drop-in replacement) but uses the multi-sample leave-one-out estimator described above.


\subsection{Differentiable Particle Filters}
\label{sub:particle-filter}

Suppose we wish to estimate the derivative (w.r.t. $\theta \in \RR$) of a high-dimensional integral $\int f_\theta(\mathbf{x}) \mu_\theta(d\mathbf{x})$ for some parameterized $\sigma$-finite measure
$\mu_\theta$ over vectors $\mathbf{x}$. 
In ADEV, we could write a probabilistic program for estimating the integral, e.g. via the use of a 
randomized algorithm like Sequential Monte Carlo, and differentiate the expected value of the program. 
However, naively applying ADEV to that estimator may yield 
a \textit{derivative} estimator that (although unbiased) has very high variance. 
This is analogous to a situation that arises 
in standard, deterministic AD, where an iterative computation for estimating 
a fixed point (e.g., Newton's method for finding the root of a function) may behave 
poorly under automatic differentiation. One benefit of recent theoretical developments
in \textit{higher-order} deterministic AD is that operations like root-finding can be
treated as higher-order primitives, and the theory can be used to guide the development
of custom built-in derivatives that are more accurate~\citep{sherman2021}.

ADEV's theoretical framework similarly provides a specification for built-in derivatives of higher-order 
primitives. As such, we can expose algorithms like sequential Monte Carlo as primitives whose built-in 
derivatives employ specialized unbiased gradient estimation strategies. As an example, consider 
the algorithm for estimating SMC gradients unbiasedly recently proposed by~\citet{scibior2021differentiable}.

To encode this algorithm in ADEV, for each ground type $\sigma$, we define a new primitive $\smc_\sigma : (\text{List}~\sigma \to \RR_{\geq 0}) \to (\sigma \to D\,\sigma) \to (\text{List}~\sigma \to \eRR) \to \NN \to \NN \to \eRR$.\footnote{There is no conceptual difficulty in extending our core language with types $\text{List}~\tau$ for lists. The AD macro operates functorially on the nil and cons constructors, just as it does on products.
Although we do not have general recursion, a \textbf{fold} operation can easily be exposed as a primitive.} In order, the arguments are:
\begin{itemize}
    \item $p : \text{List}~\sigma \to \RR_{\geq 0}$, a density function for the target measure $\mu$. We assume that restricted to lists of length $i$, $p$ is a density with respect to the product reference measure $\mu_\sigma^i$ (see previous subsection for a discussion of reference measures). The sequence of measures defined for each length $i$ constitute the sequence of target measures against which sequential Monte Carlo will be run.
    
    \item $q : \sigma \to D\,\sigma$, a transition proposal for the sequential Monte Carlo algorithm. The type $D\,\sigma$ is the density-carrying distribution type defined in the previous subsection.
    
    \item $\widetilde{f} : \text{List}~\sigma \to \eRR$, an unbiased estimator of the integrand $f$.
    
    \item $N : \NN$, the number of SMC steps to run (i.e., the dimension of the space over which to integrate).
    
    \item $K : \NN$, the number of SMC particles to use.
\end{itemize}

When run forward, $\smc$ estimates $\int_{\sem{\sigma}^N} f(\mathbf{x}) \mu(d\mathbf{x})$ by running a particle filter,
using the user-specified proposal, and weighting particles according to the user-specified target density and the proposal density that is provided as part of $q$. This yields a weighted collection of $K$ particles, each of which has a weight $w^{(j)}$ and an associated vector $\mathbf{x}^{(j)}$ (for $j \in \{1, \dots, K\}$). For each particle, we run $\widetilde{f}$ to get an unbiased estimate $\hat{f}^{(j)}$ of $f(\mathbf{x}^{(j)})$, then compute $\frac{1}{K}\sum_{j=1}^K w^{(j)} \cdot \hat{f}^{(j)}$ to get an unbiased estimate of the integral in question. 

The built-in derivative begins by running the same particle filter as in the primal computation, to arrive at a collection of weighted particles. For each particle, it runs the \textit{derivative} of $\widetilde{f}$ to obtain estimates $(\hat{f}^{(j)}, \delta\hat{f}^{(j)})$ of both $f$ and its derivative with respect to $\theta$ at $\mathbf{x}^{(j)}$. For each particle it also computes the derivative $\delta l^{(j)}$ of the log of the target density at $\mathbf{x}^{(j)}$, using $p$'s derivative. It then computes $\frac{1}{K}\sum_{j=1}^K w^{(j)} \cdot (\delta l^{(j)} \cdot \hat{f}^{(j)} + \delta\hat{f}^{(j)})$, which Theorem 1 of~\citet{scibior2021differentiable} shows is an unbiased estimate of the derivative of the integral in question.

Note although this SMC derivative estimator is wrapped in a black-box primitive, rather than being derived compositionally 
from a program implementing a particle filter, we can \textit{use} it compositionally to derive new hybrid estimators. For example, the integrand $f$ can be defined compositionally by an ADEV program as the expectation of some probabilistic process (possibly one that \textit{also} uses $\smc$!).

\subsection{Stop-Gradient}
\label{sub:stop-grad}

Many existing works on the compositional derivation of gradient estimators (e.g., DICE~\citep{foerster2018dice}) make heavy use
of the \textit{stop-gradient} operator, which in our context can be understood as a forced cast from the type $\RR$ to the type $\RR^*$ (whose derivatives are not tracked). Of course, such a cast cannot be soundly added to our language; there is no way to attach a built-in derivative to it that would satisfy our logical relation at the type $\RR \to \RR^*$. However, we can add estimators to ADEV that internally erase gradient information, so long as we validate that they are sound. As a simple example, consider the primitive $\importance : D\,\sigma \times D\,\sigma \to P\,\sigma$, which takes as input a distribution $p$ and a distribution $q$ (both with densities attached\textemdash see Section~\ref{sec:densities}), and outputs a $P\,\sigma$ representing $p$, whose built-in derivative performs
importance sampling using $q$ to estimate the derivative of an expectation with respect to $p$:

\begin{center}
\begin{algorithm}[H]
\DontPrintSemicolon
\SetKwProg{Fn}{}{:}{end}
\Fn{$\importance_\der(dp : \ad{D\,\sigma}, dq : \ad{D\,\sigma}, \widetilde{dl} : \ad{\sigma} \to \deRR)$}{
$q \gets \fst(dq)$\;
$x \sim \fst(q)$\;
$dw \gets \snd(dp)(x) \div_\mathcal{D}~((\snd~q)(x),0)$\; 
$dl \sim \widetilde{dl}(x)$\;
\Return $dw \times_\mathcal{D} dl$
}
\end{algorithm}  
\end{center}

The goal is to compute $\frac{d}{d\theta}\mathbb{E}_{x \sim q_\theta}[\frac{p_\theta(x)}{q_\theta(x)} \cdot f_\theta(x)]$. 
But note that for any parameter $\eta$, $\mathbb{E}_{x \sim q_\eta}[\frac{p_\theta(x)}{q_\eta(x)} \cdot f_\theta(x)]$ is the same
value, because the proposal distribution in importance sampling does not affect the expected value. 
Therefore, the derivative is equal to $\frac{d}{d\theta}\mathbb{E}_{x \sim q_\eta}[\frac{p_\theta(x)}{q_\eta(x)} \cdot f_\theta(x)]$ for any $\eta$, and in this expression, we can push the derivative inside the expectation (under the usual regularity conditions): $\mathbb{E}_{x \sim q_\eta}[\frac{d}{d\theta}(\frac{p_\theta(x)}{q_\eta(x)} \cdot f_\theta(x))]$. This is precisely what our built-in derivative estimates, sampling $x$, and then computing the product of the importance weight with the loss, but using dual numbers. We strip $dq$ of its gradient information, to get $q$, because we do not care about the proposal's dependence on the parameter.

When $p = q$, we recover $\reinforce$ from Section~\ref{sec:densities}.

\subsection{Gradients via Implicit Differentiation}
\label{sub:implicit-diff}

For any measure $\nu_\theta$ on $\RR$ which has a density $p(x;\theta)$, when we have access to an analytic version of the inverse Cumulative Distributive Function (CDF) $F_\theta$ of $\nu_\theta$, and if it is continuously differentiable, we can use it as a reparametrization for the REPARAM method. 
For univariate distributions, one can also obtain an alternative expression for the gradient using the CDF directly.
Let $g(\epsilon; \theta):=F_\theta^{-1}(\epsilon)$ where $\epsilon\sim \sample$. Then, using the fact $\grad_x F_\theta(x)= p(x,\theta)$, one can show that $ \grad_\theta g(\epsilon; \theta)=-\frac{\grad_\theta F(x;\theta)}{p(x;\theta)}$.
This is called implicit differentiation  \cite{figurnov2018implicit, jankowiak2018pathwise} in the literature. 
Using this fact, we can extend ADEV with new primitives:

\begin{center}
\begin{algorithm}[H]
\DontPrintSemicolon
\SetKwProg{Fn}{}{:}{end}
\Fn{${p_\texttt{IMPLICIT}}_\der(d\theta : \ad{\RR}, \widetilde{dl} : \ad{\RR} \to \deRR)$}{
$x\sim p(\fst~d\theta)$\;
$\grad \gets\big(x,-\snd ~F_\der((x,0),d\theta)\div \snd~F_\der((x,1),(\fst ~d\theta,0))\big) $\;
$dx\sim \widetilde{dl}(\grad)$ \;
\Return $dx$ \;
}
\end{algorithm}  
\end{center}

In words, the derivative will sample an $x$ from the distribution $p(-;\theta)$, compute the gradient part $\delta x= -\frac{\grad_\theta F(x;\theta)}{p(x;\theta)}$, and give the dual number $(x,\delta x)$ to the loss. We can see that this follows a similar reparametrization strategy to our $\normalreparam$ primitive, but the sampling from the reparametrized $\sample$ distribution is "implicit".
Following \cite{figurnov2018implicit, jankowiak2018pathwise}, one such instance is the Gamma distribution $\gammaimplicit:\posreal\times\posreal\to\pmonad \RR$.
From the Gamma distribution, it is relatively easy to add the Beta and Dirichlet distributions.
Other examples include the von Mises and Student’s  distributions, as well as univariate mixtures.

More generally, similarly to the type $D~\sigma$, we can a new type $C~\RR$ of distributions with densities and CDFs. Our semantics would interpret the new type as $\sem{C~\RR}= \pmonad\RR\times(\RR\to\RR_{\geq 0})\times (\RR\to \RR_{\geq 0})$. Under AD, we have $\ad{C~\RR}= C~\RR\times (\ad{\RR}\to \ad{R})\times (\ad{\RR}\to \ad{R})$. Similarly to $D~\sigma$, our logical relation $\mathcal{R}_{C\RR}$ relates a parametrized distribution $p : \RR \to P\,\RR \times (\RR \to \RR_{\geq 0})\times (\RR \to \RR_{\geq 0})$ to its derivative $dp : \RR \to (P\,\RR \times (\RR \to \RR_{\geq 0})\times (\RR \to \RR_{\geq 0})) \times (\ad{\RR} \to \RR \times \RR) \times (\ad{\RR} \to \RR \times \RR)$ if: (1) for all $\theta$, $\pi_1(p(\theta))$ has density $\pi_2(p(\theta))$ with respect to $\lambda$, (2) $\pi_3(p(\theta))$ is the CDF of $\pi_1(p(\theta))$, (3) the density $\pi_2(p(\theta))$ is $\mu_\RR$-almost-everywhere non-zero, (4) $\pi_1 \circ dp = p$, (5) $(\pi_2 \circ p, \pi_2 \circ dp) \in \mathcal{R}_{\RR \to \RR_{\geq 0}}$, (6)  $(\pi_3 \circ p, \pi_3 \circ dp) \in \mathcal{R}_{\RR \to \RR_{\geq 0}}$. That is, the density needs to match the distribution, the CDF should match the density, the distribution needs to have full support (non-zero density), and the derivative of the density and the CDF need to be correct. 
Using this, we can wrap-up the construction above as a higher-order primitive $\implicitdiff: C~\RR\to \pmonad\RR$.

\subsection{Gradients via Weak Derivatives}
\label{sub:weak-deriv}

Another distinct estimation strategy from the literature uses the weak derivative method \cite{pflug1996optimization,heidergott2000measure}. The idea is that the gradient of a probability density $\grad_\theta p(x;\theta)$ will not be a probability density in general, but can be decomposed as a weighted difference of two probability densities $\grad_\theta p(x;\theta)= c^+_\theta p^+(x;\theta)-c^-_\theta p^-(x;\theta)$. In fact, it is always possible to choose $c^+_\theta = c^-_\theta$, and such a triple $(c_\theta,p^+,p^-)$ is called a \textit{weak derivative} of $p$.

We then derive an estimator for the gradient of the expectation of a loss $l$ under $p$ by estimating the expectation of $l$ under $p^+$, under $p^-$, subtracting the first one to the second, and multiplying by the normalizing factor $c_\theta$.

Based on this estimation strategy, we can extend ADEV with new primitives. As an example, we look at a Poisson distribution $\mathcal{P}$, for which the weak derivative can be written as $(1,\mathcal{P}+1,\mathcal{P})$, where $\mathcal{P}+1$ is a notation for pushforward measure of $\mathcal{P}$ by the function $x\mapsto x+1$.

\begin{center}
\begin{algorithm}[H]
\DontPrintSemicolon
\SetKwProg{Fn}{}{:}{end}
\Fn{$\poissonweak_\der(d\theta : \ad{\posreal}, \widetilde{dl} : \NN \to \deRR)$}{
$(\theta, \delta \theta) \gets d\theta$\;
$x^-\sim \poissonweak(\theta)$\;
$x^+ \gets x^-+1$\; 
$(y^+,\_)\sim \widetilde{dl}(x^+)$\;
$(y^-,\delta y^-)\sim \widetilde{dl}(x^-)$\;
$\grad\text{est} \gets y^+-y^-$\;
\Return $(y^-,\delta y^-+(\grad\text{est} \times \delta\theta))$ \;
}
\end{algorithm}  
\end{center}

We have used correlated samples $(x_1, x_2)$ as it usually lowers variance, and we used the fact $p^-=p$ in this specific situation.
We could also write a more general version that estimates the expectation under the measures $p^+, p^-$ using $N$ samples instead of just one. 
This estimation strategy is quite general as, similarly to the score estimator, it does not require the loss to be differentiable w.r.t. its argument. Therefore, the local-domination property allowing the exchange of integral and derivative is still sufficient to ensure the correctness of primitives using the weak derivative estimator, as long as we have a valid weak-derivative triple.

\subsection{Hybrid Estimator: Gradient through Rejection Sampling}\label{sub:reject}

For several distributions, we don't have access to analytic CDFs or inverse CDFs, or these can be computationally expensive. In such a case, a simple rejection-sampling algorithm can be a good strategy for sampling from the desired distribution. The problem is to be able to somehow differentiate through the rejection sampler. \citet{naesseth2017reparameterization}'s solution, which we now recall, can be added to ADEV as a new higher-order primitive $\reparamreject: D~\sigma\times (\sigma\to \sigma) \times D~\sigma\times D~\sigma \times\posreal \to \pmonad\sigma$.
The arguments are to be interpreted as follows. $\reparamreject((S,s),h,(P,p),(Q,q),M)$ assumes that 
(1) $H_*S=Q$ and (2) $p\leq M\times q$. The distribution of interest is $P$. The interpretation of $\reparamreject$ is given as follows.

\begin{center}
\begin{algorithm}[H]
\DontPrintSemicolon
\SetKwProg{Fn}{}{:}{end}
\Fn{$\reparamreject(s:D~\sigma, h:\sigma\to\sigma, p:D~\sigma,q:D~\sigma, M:\posreal^*)$}{
$i\gets 0$\;
\Repeat{
$u_i < \frac{(\snd ~p)(h(\epsilon_i))}{M\times(\snd~q)(h(\epsilon_i))}$
}{
$i \gets i+1$\;
$\epsilon_i\sim \fst~s$\;
$u_i\sim\sample$\;
}
\Return $\epsilon_i$\;
}
\end{algorithm}  
\end{center}

The derivative of this new primitive is given as follows.

\begin{center}
\begin{algorithm}[H]
\DontPrintSemicolon
\SetKwProg{Fn}{}{:}{end}
\Fn{$\reparamreject_\der(ds:\ad{D~\sigma}, dh:\ad{\sigma}\to\ad{\sigma}, dp:\ad{D~\sigma}, dq:\ad{D~\sigma}, M:\posreal^*,  \widetilde{dl} : \ad{\sigma} \to \deRR)$}{
$h\gets \lambda x.\fst~dh(x,0)$\;
$\pi \gets \reparamreject(\fst~ds, h, \fst~dp,\fst~dq, M)$\;
$(s_{density},p_{density},q_{density})\gets \big(\snd(\fst(ds)),\snd(\fst(dp)),\snd(\fst(dq))\big)$\;
$\pi_{density}\gets \lambda \epsilon. s_{density}(\epsilon) \times \frac{p_{density}(h(\epsilon))}{q_{density}(h(\epsilon))}$\;
$d \pi_{density}\gets \lambda d\epsilon. \big((\snd~ds)(d\epsilon)\big) \times_\der \Big(\big((\snd~ dp)(dh(d\epsilon))\big) \div_\der \big((\snd ~dq)(dh(d\epsilon))\big)\Big) $\;
$d\pi \gets \big((\pi,\pi_{density}), d \pi_{density}
\big)$ \;
$\text{new\_loss} \gets \lambda dx. \widetilde{dl}(dh(dx))$\;
$dx\sim \reinforce_\der(d\pi, \text{new\_loss})$\;
\Return $dx$\;
}
\end{algorithm}
\end{center}

The idea is that one can show that $\mathbb{E}_{p(\theta)}[l] =\mathbb{E}_{\pi(\theta)}[l\circ h_\theta]$ where $\pi(\theta)$ is the posterior distribution of the rejection sampling algorithm. We then obtain an unbiased estimate of the gradient of the LHS by using the REINFORCE estimator on the RHS. Crucially, $\pi$ is normalized and \citet{naesseth2017reparameterization} were able to compute a simple analytic form for the density of $\pi(\theta)$, which we need for the REINFORCE estimator. 
As an example, \citet{naesseth2017reparameterization} showed how one can recover an efficient rejection sampler for the Gamma distribution from \citep{marsaglia2000simple}.

Finally, we cannot directly show the correctness of $\reparamreject$ in ADEV as such, because the semantics and logical relations are not checking when conditions (1) and (2) are valid. This can be fixed by a similar technique as used before with $D~\sigma$ and $C~\RR$ by defining a new type $F~\sigma$ such that  $\sem{F~\sigma} =\sem{D~\sigma \times (\sigma\to\sigma)\times D~\sigma \times D~\sigma \times \posreal^*}$. $F~\sigma$ is then given a new logical relation $\mathcal{R}_{F~\sigma}$ that would additionally encompass conditions (1) and (2) above. When constructing a term of type $F~\sigma$, one would simply have to ensure that these extra conditions are satisfied.  

%% file: appendix/add-figures.tex
\section{Additional figures}
\label{appx:figures}

Figure~\ref{fig:03primitives} presents the interpretation of the primitives from Section~\ref{sec:combinator} and of their AD-translation.
Similarly, Figure~\ref{fig:probabilistic_semantics} presents the interpretation of the primitives from Section~\ref{sec:continuous} and of their AD-translation.
Figure~\ref{fig:verif_constants} presents the interpretation of the new component of the AD translation of our primitives appearing in Sections~\ref{sec:continuous} and \ref{sec:general}.
Figure~\ref{fig:log_rel_recap} presents our logical relations at all types.


\input{figures/03figures/primitives}
 \input{figures/05figures/primitives}
\input{figures/verif_constant_translation}
\input{figures/09figures/log_rel}

%% file: figures/03figures/primitives.tex
\begin{figure}[t]
\fbox{
  \parbox{.97\textwidth}{
  \begin{center}
      
  Semantics of types:
  
  \begin{tabular}{ll}
      $\sem{\eRR}$ &= $\{\mu~|~\mu\text{ a probability measure on }(\RR,\mathcal{B}(\RR))\}$ \\
      $\sem{\eRR_\der}$ &= $\{\mu~|~\mu\text{ a probability measure on }(\RR\times\RR,\mathcal{B}(\RR\times \RR))\}$
  \end{tabular}
  \vspace{.2cm}
  
  Semantics of primitives:
   \end{center} 
     \footnotesize{
\begin{tabular}{llll}
    \begin{minipage}{.22\linewidth}
    \hspace{-3mm}
     \begin{algorithm}[H]
\DontPrintSemicolon
\SetKwProg{Fn}{}{:}{end}
\Fn{$\fst_*(d\mu:\deRR)$}{
 $(x,\delta x)\sim d\mu$\; 
 \Return $x$
}
\end{algorithm}  
\end{minipage} 
& 
\begin{minipage}{.22\linewidth}
     \begin{algorithm}[H]
\DontPrintSemicolon
\SetKwProg{Fn}{}{:}{end}
\Fn{$\snd_*(d\mu:\deRR)$}{
 $(x,\delta x)\sim d\mu$\; 
 \Return $\delta x$
}
\end{algorithm}  
\end{minipage} 
     & 
\begin{minipage}{.24\linewidth}
     \begin{algorithm}[H]
\DontPrintSemicolon
\SetKwProg{Fn}{}{:}{end}
\Fn{$\exact(x:\RR)$}{
 \Return $x$
}
\end{algorithm}  
\begin{algorithm}[H]
\DontPrintSemicolon
\SetKwProg{Fn}{}{:}{end}
\Fn{$\fst~\exact_\der(x:\RR\times\RR)$}{
 \Return $x$
}
\end{algorithm}  
\end{minipage}
&
\begin{minipage}{.25\linewidth}
\begin{algorithm}[H]
\DontPrintSemicolon
\SetKwProg{Fn}{}{:}{end}
\Fn{$\etimes(\widetilde{x}:\eRR,\widetilde{y}:\eRR)$}{
$r\sim\widetilde{x}$\;
$s\sim\widetilde{y},$\; 
\Return $r\times s$
}
\end{algorithm}  
\end{minipage}
\end{tabular}
\begin{tabular}{lll}
\hline\\
\hspace{-.4cm}
\begin{minipage}{.45\linewidth}
     \begin{algorithm}[H]
\DontPrintSemicolon
\SetKwProg{Fn}{}{:}{end}
\Fn{$\minibatch(M : \NN, m : \NN, f : \NN \to \RR)$}{
\uIf{$M=0$}{
\Return 0\;
}
\uElseIf{$m=0$}{
\Return $\sum_{i=1}^Mf(i)$\;
}
\Else{
\For{$j=1$ to $m$}{
$i_j\sim Unif(\{1,\ldots,M\})$\;
}
\Return $\frac{M}{m}\sum_{i=1}^mf(i_j)$\;
}
}
\end{algorithm}  
\end{minipage}
     & 
\hspace{-2.0cm}
\begin{minipage}{.4\linewidth}
\begin{algorithm}[H]
\DontPrintSemicolon
\SetKwProg{Fn}{}{:}{end}
\Fn{$\fst~ \minibatch_\der(M : \NN, m : \NN,$\;
$df : \NN \to \RR\times\RR)$}{
\uIf{$M=0$}{
\Return (0,0)\;
}
\uElseIf{$m=0$}{
\Return $\sum_{i=1}^M df(i)$\;
}
\Else{
\For{$j=1$ to $m$}{
$i_j\sim Unif(\{1,\ldots,M\})$\;
}
\Return $\frac{M}{m}\sum_{i=1}^mdf(i_j)$\;
}
}
\end{algorithm}     
\end{minipage}
&
\hspace{-1.5cm}

\begin{minipage}{.33\linewidth}
     \begin{algorithm}[H]
\DontPrintSemicolon
\SetKwProg{Fn}{}{:}{end}
\Fn{$\eexp_\der(\widetilde{dx}:= \fst~dx:\deRR)$}{
  $(r,\delta r)\sim \widetilde{dx}$\;
  $s\sim \eexp(\fst_*\widetilde{dx})$\;
  \Return $(s,\delta r\times s)$ 
}
\end{algorithm}  
\begin{algorithm}[H]
\DontPrintSemicolon
\SetKwProg{Fn}{}{:}{end}
\Fn{$\etimes_\der(\widetilde{dx}:=\fst~dx:\deRR,\widetilde{dy}:= \fst~dy:\deRR)$}{
$dr\sim \widetilde{dx}$\;
$ds\sim \widetilde{dy},$\; 
\Return $dr\times_\der ds$
}
\end{algorithm}  
\end{minipage}
\end{tabular}
\begin{tabular}{lll}
\hline \\
\hspace{-.3cm}
\begin{minipage}{.25\linewidth}
     \begin{algorithm}[H]
\DontPrintSemicolon
\SetKwProg{Fn}{}{:}{end}
\Fn{$\eplus(\widetilde{x}:\eRR,\widetilde{y}:\eRR)$}{
$b\sim flip(0.5)$\;
\eIf{$b$}{
$r\sim \widetilde{x}$\;
\Return $2\times r$}{
$r\sim \widetilde{y}$\;
\Return $2\times r$
}
}
\end{algorithm}  
\end{minipage}
&
\hspace{-.8cm}
\begin{minipage}{.44\linewidth}
     \begin{algorithm}[H]
\DontPrintSemicolon
\SetKwProg{Fn}{}{:}{end}
\Fn{$\fst~\eplus_\der(\widetilde{dx}:=\fst~dx:\deRR,\widetilde{dy}:= \fst~dy:\deRR)$}{
$b\sim flip(0.5)$\;
\eIf{$b$}{
$dr\sim \widetilde{dx}$\;
\Return $(2,0)\times_\der dr$}{
$dr\sim \widetilde{dy}$\;
\Return $(2,0)\times_\der dr$
}
}
\end{algorithm}  
\end{minipage}
&
\hspace{-.8cm}
\begin{minipage}{.34\linewidth}
     \begin{algorithm}[H]
\DontPrintSemicolon
\SetKwProg{Fn}{}{:}{end}
\Fn{$\eexp(\widetilde{x}:\eRR)$}{
 $\lambda = 2$\;
$n\sim Poisson(\lambda)$\;
\For{$i=1$ to $n$}{
$x_i\sim \widetilde{x}$\;
}
\Return $exp(\lambda)\times\prod_{i=1}^n\frac{x_i}{\lambda}$
}
\end{algorithm}  
\end{minipage}
\end{tabular}
 }}}
\caption{Semantics of the Combinator DSL, including built-in derivatives for each primitive. We only give the "dual-number" part of the derivative, hence the $\fst$ appearing in some definitions and on some of their arguments. }
\label{fig:03primitives}
\end{figure}

%% file: figures/05figures/primitives.tex
\begin{figure}[t]
\fbox{
  \parbox{.97\textwidth}{
  \small{
 \begin{tabular}{cc}
 \hspace{-.2cm}
\begin{minipage}{.5\linewidth}
     \begin{algorithm}[H]
\DontPrintSemicolon
\SetKwProg{Fn}{}{:}{end}
\Fn{$\fst~\geometricreinforce_\der(dp : \II\times \RR, \widetilde{dl}: \NN\to\deRR \times (S\to \RR\times \RR))$}{
  $n\sim
    \text{Geometric}(\fst~dp)$ \;
      $(l_1,l_2)\sim \fst~\widetilde{dl}~n  $  \;
      $dr \gets \pows_\der ((1,0)-_\der dp)~ n$  \;
      $dlp \gets\logs_\der~(dr \times_\der dp)$  \;
      $\delta logpdf \gets \snd~dlp$ \;
      \Return $(l_1,l_2+l_1\times \delta logpdf)$
}
\end{algorithm}  
\end{minipage}
     &
     \hspace{-.3cm}
\begin{minipage}{.5\linewidth}
\begin{algorithm}[H]
\DontPrintSemicolon
\SetKwProg{Fn}{}{:}{end}
\Fn{$\fst~\normalreinforce_\der(d\mu : \RR\times \RR, d\sigma : \RR\times \RR,$\;
$\widetilde{dl}: \RR^*\to\deRR \times (S\to \RR\times \RR))$}{
$x\sim \mathcal{N}(\fst~d\mu,\fst~d\sigma)$ \;
        $   dx \gets (x,0)  $ \;
        $(l_1,l_2)\sim~\fst~(\widetilde{dl}~x)$ \;
        $  dlp_1 \gets (-1,0) \times_\der \logs_\der d\sigma  $ \;
        $  d\epsilon \gets \pows~((dx -_\der d\mu)\div_\der d\sigma)~2$ \;
        $  dlp_2 \gets (0.5, 0) \times_\der d\epsilon  $ \;
        $  \delta logpdf \gets  \snd~(dlp_1-_\der dlp_2)   $ \;
        \Return $(l_1,l_2+l_1\times \delta logpdf)$
}
\end{algorithm}     
\end{minipage}
\end{tabular}

\begin{center}
 \begin{tabular}{cc}
 \hline \\
  \hspace{-2.8cm}
 \begin{minipage}{.5\linewidth}
     \begin{algorithm}[H]
\DontPrintSemicolon
\SetKwProg{Fn}{}{:}{end}
\Fn{$\fst~\normalreparam_\der(d\mu : \RR\times \RR, d\sigma : \RR\times \RR, \widetilde{dl}: \RR\times\RR\to\deRR \times (S\to \RR\times \RR))$}{
$\epsilon\sim \mathcal{N}(0,1)$ \;
$dx\sim \fst~\widetilde{dl}\left((\epsilon,0) ~\times_\der ~d\sigma~ +_\der ~d\mu\right)$ \;
\Return $dx$
}
\end{algorithm}  
\end{minipage}
&
 \begin{minipage}{.5\linewidth}
     \begin{algorithm}[H]
\DontPrintSemicolon
\SetKwProg{Fn}{}{:}{end}
\Fn{$\fst~\sample_\der(\widetilde{dl}: \II^* \to\deRR \times (S\to \RR\times \RR))$}{
$x\sim \sample$ \;
$dy\sim \fst~(\widetilde{dl}(x))$ \;
\Return $dy$
}
\end{algorithm}  
\end{minipage}
\hspace{-3cm}
\end{tabular}
\end{center}
 }}}
\caption{Built-in derivatives for our new probabilistic primitives. Here, we only give the dual-number translation of those primitives, not the extra witness produced by the final \ad{-} translation.}
\label{fig:probabilistic_semantics}
\end{figure}

%% file: figures/verif_constant_translation.tex
\begin{figure}[t]
\fbox{
  \parbox{.97\textwidth}{
  \small{

   \begin{tabular}{cc}
\begin{minipage}{.5\linewidth}
     \begin{algorithm}[H]
\DontPrintSemicolon
\SetKwProg{Fn}{}{:}{end}
\Fn{$\snd~\flipenum_\der(dp : \II\times \RR, dl: \BB\to \deRR \times (S\to \RR\times\RR))$}{
$\lambda s.$\;
 \Return $((\snd~(dl~ \textbf{True}))~s) \times_\der dp + ((\snd~(dl~ \textbf{False}))~s) \times_\der ((1,0) -_\der dp)$
}
\end{algorithm}  
\end{minipage}
     & 
     \hspace{-1cm}
\begin{minipage}{.55\linewidth}
\begin{algorithm}[H]
\DontPrintSemicolon
\SetKwProg{Fn}{}{:}{end}
\Fn{$\snd~\flipreinforce_\der(dp : \II\times \RR, dl: \BB\to \deRR \times (S\to \RR\times\RR))$}{
$\lambda s.$\;
 \Return $((\snd~(dl~ \textbf{True}))~s) \times_\der dp +_\der ((\snd~(dl~ \textbf{False}))~s) \times_\der ((1,0) -_\der dp)$
}
\end{algorithm}     
\end{minipage}
\end{tabular}

 \begin{tabular}{cc}
 \hline \\
\begin{minipage}{.45\linewidth}
     \begin{algorithm}[H]
\DontPrintSemicolon
\SetKwProg{Fn}{}{:}{end}
\Fn{$\snd~\geometricreinforce_\der(dp : \RR\times \RR,$\; 
$dl: \NN\to \deRR \times (S\to \RR\times\RR))$}{
 $\lambda s.$\;
  $(s_1, s_2) \gets \splits ~s$\;
  \eIf{$s_1\geq 0$}{
  $n \gets \textbf{floor} ~s_1$\;
  \Return $geom\_pdf(n; ~dp) \times_\der ((\snd ~(dl~ n)) ~s_2)$\;}{
  \Return $(0,0)$}
}
\end{algorithm}  
\end{minipage}
     & 
\begin{minipage}{.5\linewidth}
\begin{algorithm}[H]
\DontPrintSemicolon
\SetKwProg{Fn}{}{:}{end}
\Fn{$\snd~\normalreinforce_\der(d\mu : \RR\times \RR, d\sigma : \posreal\times \RR$,~
$dl: \RR^* \to \deRR \times (S\to \RR\times\RR))$}{
$\lambda s.$\;
  $(s_1, s_2) \gets \splits ~s$\;
  \Return $\mathcal{N}_\der(\forget{s_1}; ~d\mu, d\sigma) \times_\der ((\snd~(dl~ \forget{s_1}))~ s_2)$
}
\end{algorithm}     
\end{minipage}
\end{tabular}

\begin{center}
 \begin{tabular}{cc}
 \hline \\
 \hspace{-.4cm}
 \begin{minipage}{.5\linewidth}
     \begin{algorithm}[H]
\DontPrintSemicolon
\SetKwProg{Fn}{}{:}{end}
\Fn{$\snd~\normalreparam_\der(d\mu : \RR\times \RR, d\sigma : \posreal\times \RR$, 
$dl: \RR\times\RR\to \deRR\times (S\to \RR\times\RR))$}{
$\lambda s.$\;
  $(s_1, s_2) \gets \splits ~s$\;
  \Return $\mathcal{N}_\der(\forget{s_1}; ~0, 1) \times_\der ((\snd~(dl ~(\forget{s_1} \times_\der d\sigma +_\der d\mu)))~ s_2)$
}
\end{algorithm}  
\end{minipage}
& 
 \hspace{-.7cm}
 \begin{minipage}{.5\linewidth}
     \begin{algorithm}[H]
\DontPrintSemicolon
\SetKwProg{Fn}{}{:}{end}
\Fn{$\snd~\sample_\der(dl: \II^* \to \deRR\times (S\to \RR\times\RR))$}{
$\lambda s.$\; 
$(s_1, s_2) \gets \splits~ s$\;
\eIf{$s_1 \not\in (0, 1)$}{
\Return $(0, 0)$\;}{
\Return $((\snd~(dl ~\forget{s_1})) ~s_2)$\;}
}
\end{algorithm}  
\end{minipage}
\hspace{-1cm}
\end{tabular}
\end{center}
 }}}
\caption{The revised translation of our probabilistic primitives returns a pair. The first component is given in the 2 figures above. Here we only give the second component, the witness for checking the weak domination property. $\splits:S\to S\times S$ produces 2 random seeds out of 1 and $\forget{-}:S\to\RR$ is just an inclusion. Semantically, $S$ will be the measurable space $(\RR,\mathcal{B}(\RR))$ and $\sem{\splits}$ will be any measurable isomorphism between $\RR$ and $\RR\times \RR$.}
\label{fig:verif_constants}
\end{figure}





%% file: figures/09figures/log_rel.tex
\begin{figure}[tb]
\small{
\fbox{
  \parbox{.97\textwidth}{
  \begin{align*}
      \rel_\RR &= \big\{(f : \RR \to \RR, g : \RR \to \RR \times \RR) \mid f \text{ differentiable} \wedge \forall \theta \in \RR. g(\theta) = (f(\theta), f'(\theta))\big\} \\
      \rel_{\RR_{>0}} &= \big\{(f : \RR \to \RR_{>0}, g : \RR \to \RR_{>0} \times \RR) \mid (\iota_\RR \circ f, \langle \iota_\RR, id\rangle \circ g) \in \rel_\RR\big\} \\
      \rel_{\II} &= \big\{(f : \RR \to \II, g : \RR \to \II \times \RR) \mid (\iota_\RR \circ f, \langle \iota_\RR, id\rangle \circ g) \in \rel_\RR\big\} \\
      \rel_\NN &= \big\{(f : \RR \to \NN, g : \RR \to \NN) \mid f \text{ is constant } \wedge f = g\big\} \\
      \rel_{\tau_1 \times \tau_2} &= \big\{(f : \RR \to \tau_1 \times \tau_2, g : \RR \to \ad{\tau_1} \times \ad{\tau_2}) \mid\\ 
      &\qquad(\pi_1 \circ f, \pi_1 \circ g) \in \rel_{\tau_1} \wedge (\pi_2 \circ f, \pi_2 \circ g) \in \rel_{\tau_2}\big\} \\
      \rel_{\tau_1 \to \tau_2} &= \big\{(f : \RR \to \tau_1 \to \tau_2, g : \RR \to \ad{\tau_1 \to \tau_2}) \mid \\
      &\qquad\forall (j, k) \in \rel_{\tau_1}. (\lambda r. f(r)(j(r)), \lambda r. g(r)(k(r))) \in \rel_{\tau_2}\big\} \\
      \rel_\BB &= \big\{(f : \RR \to \BB, g : \RR \to \BB) \mid f \text{ is constant } \wedge f = g\big\} \\
        \rel_{\pmonad ~\type} &= \big\{(f:\RR\to \sem{\pmonad~\type},g:\RR\to (\sem{\ad{\type}}\to \deRR)\to \deRR) \mid \\
        &\qquad (\lambda \theta.\lambda \widetilde{l}:\sem{\tau}\to\eRR.\sqint \widetilde{l}(x)f(\theta)(dx),g)\in \rel_{(\tau\to\eRR)\to\eRR}\big\}\\
        \rel_{\eRR} &= \big\{(f:\RR\to\eRR, g:\RR\to(\deRR\times (S\to\RR\times \RR)))\mid h_i := \lambda \theta. \lambda s. \pi_i((\pi_2\circ g)(\theta)(s))\\
      &\qquad\wedge \forall \theta.\int_{\RR} h_1(\theta)(s)ds = \mathbb{E}_{x\sim f(\theta)}[x] = \mathbb{E}_{x\sim{\pi_1}_* (\pi_1 \circ g)(\theta)}[x]\\
      &\qquad\wedge \forall \theta.\int_{\RR} h_2(\theta)(s) ds = \mathbb{E}_{x\sim{\pi_2}_* (\pi_1 \circ g)(\theta)}[x] \\
      &\qquad\wedge (\lambda \theta. \lambda s. h_1(\theta)(s), \lambda \theta. \lambda s. (h_1, h_2)(\theta)(s)) \in \rel_{S \to \RR} \big\} \\
    \rel_{\mathbb{K}^*} &= \big\{(f : \RR \to \sem{\mathbb{K}^*}, g : \RR \to \sem{\mathbb{K}^*}) \mid f \text{ is constant } \wedge f = g\big\}
  \end{align*}
  }}
    \vspace{-2mm}
    }
\caption{Definition of the dual-number logical relation at each type.}
\label{fig:log_rel_recap}
\end{figure}

%% file: appendix/qbs.tex
\section{Quasi Borel spaces and logical relations}
\label{appx:qbs_logrel}

We give the semantics of our language in the category of Quasi-Borel spaces \citep{heunen2017convenient}, and then present a categorical view on our logical relations argument. 

\subsection{Quasi-Borel Spaces}

\textbf{Definition and basic properties.}
A QBS $(X,\plots_X)$ is a pair of a set $X$ and a set $\plots_X \subseteq \RR\to X$ of so-called random elements, satisfying 3 conditions:
\begin{enumerate}
    \item Constant functions are random elements
    \item If $f:\RR\to X$ is a random element and $g:\RR\to\RR$ is measurable, then $g;f$ is a random element
    \item A countable family of random elements $f_i:A_i\to X$ on a partition $\{A_i\}$ of $\RR$ gives a random element $f:\RR\to X, x\in A_i\mapsto f_i(x)$.
\end{enumerate}

A morphism $f:(X,\plots_X)\to (Y,\plots_Y)$ between two QBS is a function $f:X\to Y$ such that for all $g\in\plots_X$, $g;f\in\plots_Y$. We write $QBS(X,Y)$ for the set of QBS-morphisms between $(X,\plots_X)$ and $(Y,\plots_Y)$.

QBS essentially inherit all of the nice properties of sets: closure under products, coproducts, function spaces, by being the usual construction on the underlying sets and an appropriate one on the sets of random elements. 

For instance, we have $\plots_{X\times Y}\{(f,g)\mid f\in\plots_X,g\in\plots_Y\}$ and $\plots_{X\to Y}=\{f:\RR\to X\times Y\mid \lambda (r,x). f(r)(x) \in QBS(\RR\times X,Y)\}$.

Abstractly, this is a consequence of QBS forming a category of concrete sheaves. 
As such, it is a Grothendieck quasitopos, and therefore Cartesian-closed, complete and cocomplete.

QBS also enjoys the very desirable property of supporting a commutative probability monad $P$ \citep{heunen2017convenient} which allows us to interpret our monadic constructs and primitives.

\textbf{Semantics of our language.} The new interpretation follows an evident analogue in QBS of our previous Set-interpretation.
In more detail, we have
\begin{align*}
    \sem{\RR}&=(\RR,\{f:\RR\to\RR\mid f\text{ measurable}\}) \\
    \sem{\NN} &= (\NN, \{f:\RR\to\NN \text{ measurable}\}) \\
    \sem{\BB} &= (1+1, \{f:\RR\to(1+1) \text{ measurable}\}) \\
    \sem{\type_1\times \type_2} &=\sem{\type_1}\times \sem{\type_2} \\
    \sem{\type_1\to \type_2} &=\sem{\type_1}\to \sem{\type_2} \\
    \sem{\pmonad \type}&=P~\sem{\type} \\
\end{align*}
The deterministic primitives are interpreted as the standard measurable functions, and more generally the inductive interpretation follows the exact same structure as the Set-interpretation.

\subsection{Logical relations, categorically}

We now give a more categorical approach to our logical relations argument, in the same vein as recent work on categorical glueing for logical relations \citep{huot2020correctness,katsumata2013relating,vakar2019domain, mitchell1992notes}.

Let $\SubSet$ be the following category. An object is a pair of sets $(A,B)$ such that $A\subseteq B$, and a morphism $f:(A,B)\to (C,D)$ is a function $f:B\to D$ such that $f(A)\subseteq C$. The projection to the second component induces a functor $\pi:\SubSet\to\Set$ that is a fibration for logical relations in the sense of  \citep{katsumata2013relating}.
We define the functor $F:\QBS\times\QBS\to\Set$ as $F((X,Y)):= \QBS(\RR,X\times Y)$ where $\QBS(\RR,X\times Y)$ is the underlying set of functions.

As $\QBS$ is a bi-Cartesian closed category and $F$ is product preserving, the pullback of $\pi$ along $F$ induces another fibration for logical relations $p:\Gl\to \QBS\times\QBS$. An object in $\Gl$ is a triple $(X,Y,R)$ where $R$ is a subset of the set $\QBS(\RR,X\times Y)$. 
A morphism $(X_1,Y_1,R)\to (X_2,Y_2,S)$ in $\Gl$ is a pair of QBS-morphisms $f:X_1\to X_2,g:Y_1\to Y_2$ such that for all $h\in R$, we have $h;(f\times g)\in S$. 

We now interpret our language in $\QBS\times\QBS$ with $\sem{\type}_{new}=\sem{\type}\times \sem{\ad{\type}}$ and $\sem{\ter}_{new}=\sem{\ter}\times \sem{\ad{\ter}}$.
We consider the monad $P\times (((-)\Rightarrow \deRR)\Rightarrow \deRR)$ on $\QBS\times\QBS$. 
We are now set up to use the main theorem from \citep{katsumata2013relating}. 
To do so, we chose the interpretation of our base types $G$ in $\Gl$ to be $(\sem{G},\sem{\ad{G}},\rel_G)$, and the semantics of $\pmonad~G$ to be $(\sem{\pmonad G},\sem{\ad{\pmonad G}},\rel_{\pmonad G})$. We easily check that the return of the monad on $\QBS\times\QBS$ has a lift in $\Gl$ at all base types.
Finally, it only remains to show that every primitive $c$ has a lift from $\QBS\times\QBS$ to $\Gl$. 
This part is as described in the main body of the paper. 
In particular, it fails for primitives sampling from continuous distributions if we don't revise $\rel$ as in Section~\ref{sec:continuous}.
At this point, we also need to consider the revised translation for $\ad{-}$ that tracks the witness for which weak-domination has to be checked.
Once we have shown that every primitive preserves the logical relation, we recover the fundamental lemma of logical relations as a direct corollary of the main theorem from \citep{katsumata2013relating},